\documentclass[12pt,twoside]{article}
\usepackage[a4paper, left=1in, right=1in, top=1in, bottom=1in]{geometry}
\usepackage{amsmath, amsthm, amssymb, amsfonts}
\usepackage{tabulary}
\usepackage{bbm}
\usepackage{multirow}

\usepackage{amsmath, amsthm, amssymb, amsfonts}
\usepackage{tabulary}
\usepackage{bbm}
\usepackage{multirow}

\usepackage[dvipsnames]{xcolor}
\usepackage{enumitem}
\usepackage{booktabs}
\usepackage{xcolor}
\usepackage{mathtools}
\usepackage{upgreek}
\usepackage{listings}
\usepackage{natbib}
\usepackage{array}
\usepackage{enumitem}

\usepackage{bm,rotating,multirow,dsfont,graphicx}

\newtheorem{theorem}{Theorem}[section]
\newtheorem{lemma}[theorem]{Lemma}

\usepackage{enumitem}

\newtheorem{definition}{Definition}

\usepackage{algorithm} 
\usepackage{algpseudocode} 
\algtext*{EndFor}  
\usepackage{caption}
\usepackage{subcaption}

\renewcommand{\thealgorithm}{\arabic{algorithm}}

\newcommand{\M}{\mathcal{M}}
\newcommand{\by}{\boldsymbol{y}}

\newcommand{\bdt}{\boldsymbol{\theta}}

\newcommand{\commentt}[1]{}

\def\1{{\mathcal 1}}

\usepackage{hyperref}
\hypersetup{citecolor=blue}

\title{\textbf{Regularized Small Area Estimation  with  Graph  Laplacian Benchmarking  priors}}

\author{
\textbf{Jairo F\'uquene-Pati\~no}\thanks{\texttt{E-mail: jafuquenepatino@ucdavis.edu.}}, 
\textbf{and}
\textbf{Brenda Betancourt} \thanks{\texttt{E-mail: bbetanc2@gmu.edu.}} \\[0.4em]
\small
$^{*}$ Department of Statistics and Data Science, University of California, Davis, CA\\
\small $^{\dagger}$ Department of Statistics, George Mason University, Fairfax, VA}

\providecommand{\keywords}[1]{
  \small	
  \textbf{\textit{Keywords---}} #1
}

\begin{document}

\pagestyle{myheadings}
\markboth{
\small\itshape \textbf{Jairo F\'uquene-Pati\~no and Brenda Betancourt}
}{
\small\itshape \textbf{Regularized SAE with  Benchmarking priors}
}

    \maketitle
    
    \date{}


\begin{abstract}

Small area estimation (SAE) often requires both borrowing information
across areas and benchmarking estimates to reliable aggregates. We develop
a Bayesian framework that addresses these two objectives jointly through a
new family of Benchmarking priors. The priors are induced by a benchmark-constrained 
regularization problem. The resulting family includes a Benchmarking Prior that incorporates the benchmarking restrictions without additional regularization across areas, and Single and Multi-View Laplacian Benchmarking Priors that introduce regularization through graph Laplacians constructed from area similarities based on external covariate information. For posterior computation under these degenerate priors, we
develop tailored MCMC algorithms based on a reduced parameterization of the
constraint space. We assess the proposed models using a data-based simulation and apply the framework to estimate Average Household Size (AHS) at the municipality level in Colombia in 2025. In this application, graph-based
regularization improves model performance, with the
Multi-View models generally producing more precise municipality-level
estimates.

\end{abstract}



\section{Introduction}\label{sec:introduction}


In many official statistics applications, survey sample sizes are too small to support reliable direct estimates for subnational areas, such as counties, municipalities, and provinces, or for demographic subpopulations, such as gender and age groups. Small area estimation (SAE) addresses this problem by combining direct survey estimates with auxiliary information to borrow strength across areas. The Fay-Herriot (FH) model \citep{fay_1979} provides a widely used area-level framework for this purpose. 

Consider $m$ small areas, and let $y_i$ be a direct survey estimator for the small area mean parameter $\theta_i$ in the $i$-th area for $i=1,...,m$. The FH model considers sampling and linking models for the direct estimates and parameter of interest, respectively.
The sampling model for the direct estimates can be written as follows
\begin{align}
\mathbf{y}
    &=\boldsymbol{\theta}+\boldsymbol{\varepsilon},
    \label{FH_model1}\\
\boldsymbol{\varepsilon}
    &\sim
    N_m(\mathbf{0},\mathbf{D}),
    \notag
\end{align}

where $\mathbf{y}=(y_1,\ldots,y_m)^\top$, $\boldsymbol{\theta}=(\theta_1,\ldots,\theta_m)^\top$, and  $\mathbf{D}
=\operatorname{diag}(D_1,\ldots,D_m)$ for sampling variances $D_1,\ldots,D_m$ assumed known.

The linking model is
\begin{align}
\boldsymbol{\theta}
    &=\mathbf{X}\boldsymbol{\beta}+\mathbf{u},
    \label{FH_model2}\\
\mathbf{u}
    &\sim
    N_m(\mathbf{0},\sigma_u^2\mathbf{I}_m),
    \notag
\end{align}
where $\mathbf{X}=(\mathbf{x}_1,\ldots,\mathbf{x}_m)^\top$ is the $m\times p$ design matrix , $\boldsymbol{\beta}=(\beta_1,...,\beta_p)^\top \in\mathbb R^p$ is the corresponding vector of regression coefficients, and $\mathbf{u}=(u_1,\ldots,u_m)^\top$ is a vector of area-level random effects assumed to be independent of the sampling errors $\boldsymbol{\varepsilon}$. For each area, $\mathbf{x}_i=(x_{i1},...,x_{ip})^\top\in \mathbb{R}^p$ contains auxiliary covariates that are treated as known and measured without error and are commonly obtained from administrative records, census data, or other external sources. In practice, the sampling variances $D_i$ are generally estimated from the survey data and subsequently treated as known in fitting the FH model.

Bayesian and frequentist extensions of the FH framework are routinely used by National Statistical Offices (NSOs) to produce estimates of key indicators for small geographic and demographic domains. These models borrow strength across areas by combining direct survey estimates with auxiliary information from administrative records, population censuses, and other data sources. The hierarchical structure of the FH model shrinks noisy direct estimates toward values predicted by the common regression structure. In this sense, classical SAE models implicitly regularize area-level estimates through hierarchical modeling assumptions. A central challenge is how to specify this structure so that information is shared across areas in a way that reflects meaningful similarities and heterogeneity among them. A large body of work has addressed this problem by introducing structured dependence through geographic or administrative relationships, leading to spatial and spatio-temporal SAE models in which neighboring or administratively related areas are assumed to have similar random effects \citep{ross2015bayesian, mercer2015space}. While geographic and administrative relationships provide a convenient basis for borrowing strength, they may not adequately capture similarities relevant to the outcome of interest. Neighboring areas can differ substantially in demographic or socioeconomic characteristics, while geographically distant areas may be substantively similar. Models that rely on fixed spatial structures may therefore smooth across heterogeneous areas or fail to borrow strength across non-neighboring but similar domains. This limitation is particularly relevant in modern SAE applications, where increasingly rich auxiliary information is available to characterize relationships among areas.

The role of regularization in SAE is closely related to the principle of borrowing strength across areas and connects naturally with developments in Bayesian hierarchical modeling and modern machine learning. Although graph and spectral methods are well established in statistics and machine learning, their integration into SAE remains limited. Recent work by \citet{fuquenefully} uses auxiliary information to define similarity relationships among areas through a weighted graph. Spectral clustering is applied to the weighted similarity graph to identify groups of similar areas \citep{von2007tutorial}, and the resulting cluster assignments are used to construct a graph Laplacian \citep{belkin2006manifold}. This Laplacian is incorporated into the regularization of the small area parameters using an objective function motivated by Regularized Task Kernel Learning (RTKL) and Laplacian Regularized Least Squares (LapRLS) \citep{evgeniou2005learning,miller2023additive}. The resulting regularized parameters are then incorporated into a Bayesian SAE model, allowing information to be shared according to the learned cluster structure. In many official statistics applications, small area estimates must also satisfy benchmarking restrictions that ensure consistency with reliable aggregates available at higher levels of aggregation \citep{ghosh2015benchmarked}. Thus, regularization and benchmarking often arise as related components of SAE. 

Benchmarking has traditionally been imposed as an adjustment to small area estimates, although fully Bayesian approaches have been developed to incorporate benchmarking into the inferential framework through modifications to the likelihood or posterior distribution \citep{zhang2020fully,Okonek2024,sugasawa2024bayesian}. The framework of \citet{fuquenefully} provides a concrete example in which both regularization and benchmarking are considered within the same analysis. However, the proposed SAE models rely on a discrete cluster structure identified through spectral clustering, and benchmarking is imposed as a post-processing step through posterior projection. 

In this work, we develop a general framework for regularized SAE with
Benchmarking priors that integrates similarity and benchmarking structures
within the model. First, we incorporate pairwise similarities derived from
external covariates directly through weighted graph Laplacians. In contrast
to the discrete cluster structure of \citet{fuquenefully}, this formulation
preserves graded relationships among areas and extends the use of
regularization to SAE applications in which similarities may not be well
represented by discrete clusters. Second, we incorporate
benchmarking restrictions directly into the regularization process. Starting from a constrained regularization problem, we define a transformation of the small area parameters that incorporates the regularization structure while satisfying the benchmarking restrictions. This transformation induces a new family of Benchmarking priors that are degenerate distributions supported on a lower-dimensional constraint space. To the best of our knowledge, this is the first Bayesian SAE approach in which graph-based regularization and benchmarking restrictions are jointly incorporated into the prior distribution of the small area parameters. 

To accommodate the degeneracy of these priors, we develop a reduced
parameterization and tailored MCMC algorithms that ensure the benchmarking
restrictions are satisfied at every posterior draw. In addition to the
methodological and computational contributions, we present a novel
application for the estimation of average household size in Colombia using
a combination of survey, census, and external socioeconomic information. The remainder of the paper is organized as follows. Section \ref{sec_intro} develops the general framework for regularized SAE under benchmarking constraints and introduces the proposed family of Benchmarking priors, including specifications based on single and multiple graph Laplacians. Section \ref{sec_computations} develops the computational methods for the proposed Benchmarking priors and corresponding SAE models. Specifically, in Section \ref{computations} we draw on existing computational methods for Gaussian Markov random fields under constraints to develop the MCMC algorithms presented in Section \ref{MCMC_algorithms}. Algorithms for constructing single and multiple similarity graph Laplacians are presented in Section \ref{Laplacian_computations}. Section \ref{sec_Application} presents the application to the estimation of Average Household Size (AHS) across municipalities in Colombia in 2025. Concluding remarks are provided in Section  \ref{sec_Conclusions}. Additional derivations and a simulation study are provided in the Supplementary Material.

\section{Bayesian SAE with Regularization and Benchmarking priors} \label{sec_intro}

In official statistics, small area estimates are rarely used in isolation. They are often required to satisfy benchmarking constraints that ensure coherence with known aggregates such as national or regional totals \citep{ghosh2015benchmarked}. In classical SAE practice, benchmarking is often imposed as a post-hoc adjustment to point estimates with uncertainty measures modified only heuristically, if at all. This separation confounds how benchmarking alters borrowing of strength across areas, and affects posterior dependence and uncertainty quantification. These considerations motivate approaches that incorporate benchmarking directly into the inferential framework. In particular, we consider constrained regularization as a natural mechanism for achieving this integration. 

Our starting point is the regularization framework of \citet{fuquenefully}, which combines regularization
of the small area parameters with linear benchmarking restrictions. We first
generalize the corresponding constrained optimization problem to accommodate a
broader class of regularization structures. We then use its solution as a
transformation of the latent small area parameters and derive the probability
distribution induced by this transformation. This provides the basis for the
Benchmarking priors introduced in Section~\ref{subsec:The_priors}.


\subsection{Constrained quadratic optimization via the KKT conditions}

Following the constrained regularization framework of \citet{fuquenefully}, let $\boldsymbol{\theta}$ denote a baseline vector of small area parameters and
let $\boldsymbol{\theta}^{\M}$ denote its regularized counterpart under model
$\M$. Consider $k$ linear benchmarking restrictions written as
\begin{equation}
    \boldsymbol{W}\boldsymbol{\theta}^{\M}=\boldsymbol{p},
    \label{benchmark_constraint}
\end{equation}
where $\boldsymbol{W}\in\mathbb{R}^{k\times m}$ is a known matrix of
benchmarking weights with $\operatorname{rank}(\boldsymbol{W})=k\leq m$,
and $\boldsymbol{p}\in\mathbb{R}^{k}$ contains the corresponding known
benchmark values.
In practice, these quantities may be obtained from
external sources such as administrative records, population censuses, or
surveys \citep{bell_2013}.

The regularization is introduced through the
quadratic term \((\bdt^{\M})^\textsf{T}\boldsymbol{\Phi}\bdt^{\M}\), where \(\boldsymbol{\Phi}\) determines the structure of the regularization across areas. In \citet{fuquenefully}, 
the regularization matrix, \(\boldsymbol{\Phi}\), is a cluster-based graph Laplacian constructed from groups obtained through spectral clustering. In that case, the penalty encourages small areas assigned to the same cluster to have similar parameter values. 
Proposition~3.1 of \citet{fuquenefully} derives the constrained solution for this particular regularization structure. The resulting transformation is subsequently applied to posterior draws to obtain benchmarked estimates. Here, we extend the constrained optimization result to a general symmetric positive semidefinite regularization matrix \(\boldsymbol{\Phi}\). This allows
a wider range of relationships among small areas to be represented through the
choice of $\boldsymbol{\Phi}$. We discuss specific choices of $\boldsymbol{\Phi}$ in Section \ref{subsec:The_priors}.

Combining the regularization penalty with the benchmarking restrictions in \eqref{benchmark_constraint}, we consider the constrained quadratic optimization problem
\begin{equation} \label{eq:constraint}
\bdt^{\M\textsf{-B}}
=
\underset{\boldsymbol{\theta}^{\M}\in\mathbb{R}^{m}}
{\operatorname{argmin}}
\left\{
\phi(\boldsymbol{\theta}-\boldsymbol{\theta}^{\M})^\textsf{T}
(\boldsymbol{\theta}-\boldsymbol{\theta}^{\M}) +  (1-\phi)(\bdt^{\M})^\textsf{T}\boldsymbol{\Phi}(\bdt^{\M})
\right\},
\end{equation}
subject to
\begin{equation}
\boldsymbol{W}\boldsymbol{\theta}^{\M}=\boldsymbol{p}. \notag
\end{equation}
Assume that
\[
\qquad
\boldsymbol{\Phi}=\boldsymbol{\Phi}^\top \succeq 0,
\qquad
\boldsymbol{W}\in\mathbb{R}^{k\times m},
\qquad
\operatorname{rank}(\boldsymbol{W})=k\leq m.
\]
The first term penalizes deviations of \(\boldsymbol{\theta}^{\M}\) from \(\boldsymbol{\theta}\), while the second imposes the regularization structure determined by \(\boldsymbol{\Phi}\). The parameter \(\phi\in(0,1]\) controls the relative contribution of these two terms.

Because the objective function is convex and the benchmarking constraints are linear, the Karush-Kuhn-Tucker (KKT) conditions characterize the optimal solution \citep{karush1939minima,kuhn1951nonlinear}. In this case, the KKT conditions yield a system of linear equations that jointly determines the constrained small area parameters and the Lagrange multipliers associated with the benchmarking restrictions. The KKT system for the constrained quadratic optimization problem is

\begin{equation}
\begin{pmatrix}
2\phi\boldsymbol{Q}_{\phi} & \boldsymbol{W}^\top\\
\boldsymbol{W} & \boldsymbol{0}
\end{pmatrix}
\begin{pmatrix}
\bdt^{\M\textsf{-B}}\\
\boldsymbol{\lambda}
\end{pmatrix}
=
\begin{pmatrix}
2\phi\boldsymbol{\theta}\\
\boldsymbol{p}
\end{pmatrix}.
\label{eq:kkt-system}
\end{equation}
where $\boldsymbol{\lambda}\in\mathbb{R}^{k}$ is the vector of Lagrange multipliers and 
$\boldsymbol{Q}_{\phi}
=
\boldsymbol{I}_m+
\frac{1-\phi}{\phi}\boldsymbol{\Phi}.$
Using the inverse of the $2\times2$ block matrix in
\eqref{eq:kkt-system}, we obtain
\begin{equation*}
\begin{pmatrix}
\bdt^{\M\textsf{-B}}\\
\boldsymbol{\lambda}
\end{pmatrix}
=
\begin{pmatrix}
(2\phi\boldsymbol{Q}_{\phi})^{-1}
+
(2\phi\boldsymbol{Q}_{\phi})^{-1}
\boldsymbol{W}^\top
\boldsymbol{Z}
\boldsymbol{W}
(2\phi\boldsymbol{Q}_{\phi})^{-1}
&
-(2\phi\boldsymbol{Q}_{\phi})^{-1}
\boldsymbol{W}^\top
\boldsymbol{Z}
\\[1ex]
-\boldsymbol{Z}
\boldsymbol{W}
(2\phi\boldsymbol{Q}_{\phi})^{-1}
&
\boldsymbol{Z}
\end{pmatrix}
\begin{pmatrix}
2\phi\boldsymbol{\theta}\\
\boldsymbol{p}
\end{pmatrix}.
\label{eq:inverse-kkt}
\end{equation*}
where
$\boldsymbol{Z}
=
-
\left(
\boldsymbol{W}
(2\phi\boldsymbol{Q}_{\phi})^{-1}
\boldsymbol{W}^\top
\right)^{-1}$.  Therefore, the constrained solution is
\begin{equation}
\bdt^{\M\textsf{-B}}
=
\boldsymbol{Q}_{\phi}^{-1}
\boldsymbol{\theta}
+
\boldsymbol{Q}_{\phi}^{-1}
\boldsymbol{W}^\top
\left(
\boldsymbol{W}
\boldsymbol{Q}_{\phi}^{-1}
\boldsymbol{W}^\top
\right)^{-1}
\left(
\boldsymbol{p}
-
\boldsymbol{W}
\boldsymbol{Q}_{\phi}^{-1}
\boldsymbol{\theta}
\right).
\label{eq:final-solution}
\end{equation}

To interpret $\bdt^{\M\textsf{-B}}$, first consider the corresponding
unconstrained regularized solution,
\begin{equation}
\boldsymbol{\theta}_{\phi}^{\M}
=
\boldsymbol{Q}_{\phi}^{-1}\boldsymbol{\theta}.
\label{eq:unconstrained}
\end{equation}
Using \eqref{eq:unconstrained}, the constrained solution in
\eqref{eq:final-solution} can be written as
\begin{equation}
\bdt^{\M\textsf{-B}}
=
\boldsymbol{\theta}_{\phi}^{\M}
+
\boldsymbol{Q}_{\phi}^{-1}
\boldsymbol{W}^\top
\left(
\boldsymbol{W}
\boldsymbol{Q}_{\phi}^{-1}
\boldsymbol{W}^\top
\right)^{-1}
\left(
\boldsymbol{p}
-
\boldsymbol{W}\boldsymbol{\theta}_{\phi}^{\M}
\right).
\label{eq:projection-form}
\end{equation}

Therefore, the second term in \eqref{eq:projection-form} adjusts the unconstrained
regularized solution to satisfy the benchmarking constraints. More precisely,
$\bdt^{\M\textsf{-B}}$ is the
$\boldsymbol{Q}_{\phi}$-metric projection of the unconstrained
regularized solution $\boldsymbol{\theta}_{\phi}^{\M}$ onto the
constraint set,
$$\left\{\boldsymbol{\theta}_{\phi}^{\M}\in\mathbb{R}^m:
\boldsymbol{W}\boldsymbol{\theta}_{\phi}^{\M}=\boldsymbol{p}\right\}.
$$
In \citet{fuquenefully}, the constrained transformation is applied to posterior samples as a post-estimation adjustment to obtain benchmarked estimates. In contrast, we use the transformation in \eqref{eq:projection-form} to define the benchmarked small area parameters. This mapping induces a distribution for
$\bdt^{\M\textsf{-B}}$ supported on the benchmarking constraint set.
We use this induced distribution to define the Benchmarking priors in the next
section.

\subsection{Model construction under benchmarking constraints}
\label{subsec:The_FHSC}

We now consider the constrained regularization framework in the context of the Fay-Herriot (FH) model. Our goal is to incorporate the benchmarking constraints directly into the model for the small area parameters. For this purpose, we use the constrained parameter vector $\bdt^{\M\textsf{-B}}$ derived in \eqref{eq:final-solution} in place of the usual small area parameter vector in the FH sampling model.

For this model construction, it is convenient to rewrite
\eqref{eq:final-solution} as
\begin{align}
\label{benchmark_PP_solution}
\bdt^{\M\textsf{-B}}
&=
\left(
\boldsymbol{Q}_{\phi}^{-1}
+
\boldsymbol{Q}_{\phi}^{-1}
\boldsymbol{W}^{\mathsf T}
\boldsymbol{L}_{\phi}
\boldsymbol{W}
\boldsymbol{Q}_{\phi}^{-1}
\right)\boldsymbol{\theta}
-
\boldsymbol{Q}_{\phi}^{-1}
\boldsymbol{W}^{\mathsf T}
\boldsymbol{L}_{\phi}\boldsymbol{p},
\end{align}
where
\[
\boldsymbol{L}_{\phi}
=
-
\left(
\boldsymbol{W}
\boldsymbol{Q}_{\phi}^{-1}
\boldsymbol{W}^{\mathsf T}
\right)^{-1}.
\]
This representation expresses the constrained parameter vector in terms of the usual small area parameter vector $\boldsymbol{\theta}$. We use $\boldsymbol{\theta}^{\M\textsf{-B}}$ as the small area parameter in the Fay-Herriot sampling model, leading to the benchmarked area-level model in Definition \ref{def:FH_B}.

 \begin{definition}
 \label{def:FH_B} 
The FH  model  with Benchmarking (FH-B) can be written as follows:
\begin{align}
\label{model_FHB_m}
\begin{split}
    \by & =  \bdt^{\textsf{FH-B}} + \boldsymbol{e},\\
    \bdt^{\textsf{FH-B}}  & = \boldsymbol{A}_{\phi}\bdt + \boldsymbol{b}_{\phi},\\
\bdt &= \mathbf{X}\boldsymbol{\beta}+\mathbf{u} ,\\
\end{split}
\end{align} 
 where $\boldsymbol{A}_{\phi} = \boldsymbol{Q}_{\phi}^{-1}+ \boldsymbol{Q}_{\phi}^{-1}
\boldsymbol{W}^T\boldsymbol{L}_{\phi}\boldsymbol{W} \boldsymbol{Q}_{\phi}^{-1}$, $\boldsymbol{Q}_{\phi}
=
\boldsymbol{I}_m
+
\frac{1-\phi}{\phi}\boldsymbol{\Phi}$, 
 $\boldsymbol{b}_{\phi}=-  \boldsymbol{Q}_{\phi}^{-1}\boldsymbol{W}^T\boldsymbol{L}_{\phi}\boldsymbol{p}$
  with $\boldsymbol{A}_{\phi}  \in \mathbb{R}^{m \times m}$ and  $\boldsymbol{L}_{\phi}=-\left(
\boldsymbol{W}
\boldsymbol{Q}_{\phi}^{-1}
\boldsymbol{W}^\top
\right)^{-1}$.  The variables $\mathbf{X}$, $\mathbf{u}$ and $\boldsymbol{e}$ are as in  (\ref{FH_model1}) and (\ref{FH_model2}) and satisfy the same assumptions.
\end{definition}

Definition \ref{def:FH_B} changes where benchmarking enters the Fay-Herriot hierarchy. The parameter \(\bdt^{\textsf{FH-B}}\) is constrained to satisfy the benchmarking restrictions. Consequently, the benchmarking restriction becomes a property of the distribution assigned to the small area parameters. The following lemma characterizes this distribution and provides the basis for the proposed family of \textit{Benchmarking priors}.

\begin{lemma}
\label{lemma1}
Definition \ref{def:FH_B} induces a family of \textit{Benchmarking priors} for \(\bdt^{\textsf{FH-B}}\), indexed by the choices of \(\boldsymbol{\Phi}\) and \(\phi\). Each member of this family follows a degenerate Normal distribution given by
\begin{align}
\label{eq:benchmarking_prior}
\bdt^{\textsf{FH-B}}
\mid \mathbf{X}, \boldsymbol{\beta}, \sigma_u^2, \phi
\sim
\mathcal{N}_{\mathcal S}^{\mathrm{deg}}
\left(
\boldsymbol{A}_{\phi}\mathbf{X}\boldsymbol{\beta}
+\boldsymbol{b}_{\phi},
\;
\sigma_u^2
\boldsymbol{A}_{\phi}\boldsymbol{A}_{\phi}^{\mathsf{T}}
\right),
\end{align}
with support
\[
\mathcal S
=
\left\{
\boldsymbol{\theta}\in\mathbb{R}^m:
\boldsymbol{W}\boldsymbol{\theta}=\boldsymbol{p}
\right\}.
\]
where $\boldsymbol{A}_{\phi}$, $\mathbf{X}$,
$\boldsymbol{\beta}$, $\boldsymbol{b}_{\phi}$, and $\sigma_u^2$
are as in Definition \ref{def:FH_B}. The Benchmarking prior has
the following two properties:

\begin{enumerate}[label=\roman*.]

\item The prior mean satisfies the external benchmarking constraints,
\[
\boldsymbol{W}(\boldsymbol{A}_{\phi}\mathbf{X}\boldsymbol{\beta}
+\boldsymbol{b}_{\phi})
=
\boldsymbol{p}.
\]

\item The variance-covariance matrix is singular.

\end{enumerate}

\end{lemma}

 \begin{proof}
See Supplementary Material, Section \ref{theory}.
    \end{proof}

\subsection{Benchmarking prior specifications} \label{subsec:The_priors} 

We now describe how different benchmarking priors and corresponding SAE
models are obtained through different choices of $\boldsymbol{\Phi}$ and
$\phi$ in the FH-B model of Definition \ref{def:FH_B}.  The simplest case arises when no additional
structure among the small areas is incorporated beyond the auxiliary variables
in the FH model and the benchmarking constraints. Setting $\phi=1$ gives
$\boldsymbol{Q}_{\phi}=\boldsymbol{I}_m$. We refer to the resulting prior in
Lemma \ref{lemma1} as the \textit{Benchmarking Prior (BP)} with corresponding
 \textit{FH model with Benchmarking under the BP (FH-BP)}.

Other specifications of the proposed family incorporate additional
relationships among the small areas through the choice of
$\boldsymbol{\Phi}$. In particular, building on \citet{fuquenefully},
we consider regularization based on a single graph Laplacian and on
multiple graph Laplacians representing different sources of similarity
among the small areas. In these specifications, the Laplacian determines
the regularization structure and the benchmarking constraints determine
the support of the resulting prior. These choices lead to the Single
and Multi-View Laplacian Benchmarking priors. Table \ref{tab:models} summarizes 
the Fay-Herriot model specifications under the regularization structures considered in this section, without
and with benchmarking. We next describe in detail the construction of the Benchmarking priors and their corresponding Fay-Herriot models.

\begin{table}[htbp] \footnotesize
\renewcommand{\arraystretch}{1.7}
\setlength{\tabcolsep}{5pt}
\begin{tabular}{lccc}
\toprule
     \textbf{Fay-Herriot Model} & $\boldsymbol{Q}_{\phi}$ & $\phi$ & \textbf{Constraint}  \\    \toprule
Standard (FH)  & $\mathbb I_m$ & 1 & None \\
Laplacian Prior (FH-L) & $\mathbb I_m
+
\frac{1-\phi}{\phi}\textsf{L}_{C^{*}}^{\textsf{S}}$ & $(0,1]$ & None \\
Multi-View Laplacian Prior (FH-MVL) & $\mathbb I_m
+
\frac{1-\phi}{\phi}\textsf{L}_{C^{*}}^{\textsf{MV}}$ & $(0,1]$  &None\\
Benchmarking Prior (FH-BP) & $\mathbb I_m$ & 1&  $\boldsymbol{W}\boldsymbol{\theta}^{\M}=\boldsymbol{p}$ \\
Laplacian Benchmarking Prior (FH-L-BP) & $\mathbb I_m
+
\frac{1-\phi}{\phi}\textsf{L}_{C^{*}}^{\textsf{S}}$ & $(0,1]$ & $\boldsymbol{W}\boldsymbol{\theta}^{\M}=\boldsymbol{p}$  \\
Multi-View Laplacian Benchmarking Prior (FH-MVL-BP) & $\mathbb I_m
+
\frac{1-\phi}{\phi}\textsf{L}_{C^{*}}^{\textsf{MV}}$ & $(0,1]$ & $\boldsymbol{W}\boldsymbol{\theta}^{\M}=\boldsymbol{p}$  \\
\bottomrule
\end{tabular} 
\caption{\small Fay-Herriot model specifications under different regularization
structures, without and with benchmarking. $\textsf{L}_{C^{*}}^{S}$ and $\textsf{L}_{C^{*}}^{\textsf{MV}}$ denote
Single and Multi-View graph Laplacians, respectively.}
\label{tab:models}
\end{table}

\subsubsection{Laplacian Benchmarking Prior (L-BP)}
\label{subsec:The_priors2}

\citet{fuquenefully} uses auxiliary information and direct estimates to construct pairwise similarities among the small areas. 
These similarities are used to form a weighted graph for spectral clustering, and the resulting cluster assignments are then used to construct a block-diagonal graph Laplacian for regularizing the small area parameters \citep{belkin2006manifold, von2007tutorial}. Thus, the similarity information enters the regularization through the discrete cluster structure. In this work, we use the pairwise similarities to define the regularization structure directly, without an intermediate clustering step. Specifically, the similarities are used as edge weights in a connected neighborhood graph, and the symmetric normalized Laplacian of this graph is incorporated into the regularization of the small area parameters. This construction preserves differences in the strength of the pairwise similarities among connected areas, allowing the degree of borrowing across areas to reflect their similarity in the auxiliary information. This approach may be particularly useful when demographic, socioeconomic, or other area-level characteristics suggest varying degrees of similarity as opposed to a clear partition of the areas into distinct groups.

Specifically, we first compute the pairwise similarities among the areas
using the available auxiliary information and the direct estimates. We then
select the smallest neighborhood size $C^{*}$ that produces a connected
graph. Let $\textsf{W}_{C^{*}}^{\textsf{S}}$ denote the weighted adjacency
matrix of the resulting graph with edge weights determined by the pairwise
similarities, and let $\textsf{D}_{C^{*}}^{\textsf{S}}$ denote the
corresponding degree matrix. The symmetric normalized graph Laplacian is
given by
\[
\textsf{L}_{C^{*}}^{\textsf{S}}
=
\boldsymbol{I}_m- (\textsf{D}_{C^{*}}^{\textsf{S}})^{-1/2}\textsf{W}_{C^{*}}^{\textsf{S}} (\textsf{D}_{C^{*}}^{\textsf{S}})^{-1/2}.
\]
The corresponding unconstrained regularized solution in
\eqref{eq:unconstrained} can be written as follows:
\begin{align}
\label{LapRLS_2}
\boldsymbol{\theta}_{\phi}^{\M}
&=
(\boldsymbol{Q}_{\phi}^{\textsf{L}_{C^{*}}^{\textsf{S}}})^{-1}
\boldsymbol{\theta},
&
\boldsymbol{Q}_{\phi}^{\textsf{L}_{C^{*}}^{\textsf{S}}}
&=
\mathbb I_m
+
\frac{1-\phi}{\phi}\textsf{L}_{C^{*}}^{\textsf{S}},
\end{align}

Under this specification, the graph Laplacian regularizes the small area
parameters according to the similarity structure represented by the weighted
graph. The parameter $\phi$ controls the strength of the regularization with
smaller values increasing the contribution of the Laplacian penalty, and
values closer to one reducing its influence. Using
$\boldsymbol{Q}_{\phi}^{\textsf{L}_{C^{*}}^{\textsf{S}}}$ in Lemma
\ref{lemma1} and Definition \ref{def:FH_B} gives the
\textit{Laplacian Benchmarking Prior (L-BP)} and the corresponding
\textit{FH model with Benchmarking under the L-BP (FH-L-BP)}.

\subsubsection{Multi-View Laplacian Benchmarking Prior (MVL-BP)}
 \label{subsec:multi_view} 

The L-BP represents similarities among areas through a single normalized symmetric graph
Laplacian constructed from a given set of auxiliary information. In many
SAE applications, however, the available auxiliary information can be
organized into distinct views, each describing a different relationship
among the areas. For example, demographic characteristics, economic
conditions, administrative information, and geographic proximity may
produce different similarity structures. Constructing a single graph
from all of these variables does not distinguish the contributions of the
different views. Multi-view graph regularization provides an alternative by
constructing a separate graph Laplacian for each view and then combining
these Laplacians into a single Multi-View Laplacian \citep{kumar2011co}. 
This construction allows each view to contribute separately to the
combined Laplacian through its own similarity structure and associated
weight.

Suppose that $J$ views of the auxiliary information are available,
indexed by $j=1,\ldots,J$. Let $C^{*}$ denote the common neighborhood
size across the views, selected using the union of the view-specific
neighborhood graphs (see Section \ref{Laplacian_computations} for details).
For each view $j$, pairwise similarities are used to construct the
corresponding symmetric normalized Laplacian
$\textsf{L}_{C^{*},j}^{\textsf{S}}$. We define the Multi-View normalized
Laplacian as
\begin{equation}
\label{eq:MVL}
\textsf{L}_{C^{*}}^{\textsf{MV}}
=
\sum_{j=1}^{J}q_j
\textsf{L}_{C^{*},j}^{\textsf{S}},
\qquad
q_j\geq 0,
\qquad
\sum_{j=1}^{J}q_j=1,
\end{equation}
where $q_j$ determines the contribution of view $j$ to the combined
similarity structure.

Setting $\boldsymbol{\Phi}=\textsf{L}_{C^{*}}^{\textsf{MV}}$ in the
general regularization framework gives
\begin{align}
\label{eq:LapRLS_3}
\boldsymbol{\theta}_{\phi}^{\M}
&=
(\boldsymbol{Q}_{\phi}^{\textsf{L}_{C^{*}}^{\textsf{MV}}})^{-1}
\boldsymbol{\theta},
&
\boldsymbol{Q}_{\phi}^{\textsf{L}_{C^{*}}^{\textsf{MV}}}
&=
\boldsymbol{I}_m
+
\frac{1-\phi}{\phi}
\textsf{L}_{C^{*}}^{\textsf{MV}}.
\end{align}
The parameter $\phi$ controls the overall strength of the graph-based
regularization, while the weights $q_1,\ldots,q_{J}$ determine the
relative contribution of the different views to the similarity structure.

Using $\boldsymbol{Q}_{\phi}^{\textsf{L}_{C^{*}}^{\textsf{MV}}}$ in
Lemma \ref{lemma1} and Definition \ref{def:FH_B} gives the
\textit{Multi-View Laplacian Benchmarking Prior (MVL-BP)} and the
corresponding \textit{FH model with Benchmarking under the MVL-BP
(FH-MVL-BP)}. The MVL-BP extends the L-BP by allowing the regularization
structure to reflect multiple sources of auxiliary information while
preserving the same benchmarking restrictions on the small area parameters.

\section{Computation}
\label{sec_computations}

Posterior computation under the proposed Benchmarking priors requires 
accounting for both the linear benchmarking constraints and the graph-based 
regularization structure. In Section
\ref{computations}, we introduce a reduced parameterization that expresses the
Benchmarking priors in terms of an unconstrained $(m-k)$-dimensional latent
vector that facilitates posterior computation under the benchmarking constraints. 
Section \ref{MCMC_algorithms} presents the corresponding
MCMC algorithm for posterior sampling. Finally, Section
\ref{Laplacian_computations} describes the algorithms for constructing the
Single and Multi-View Normalized Graph Laplacians for the L-BP and MVL-BP,
respectively.

\subsection{Reduced parameterization}
\label{computations}

The Benchmarking priors introduced in Section \ref{subsec:The_priors} are 
degenerate Gaussian distributions supported on the constraint set defined by 
the benchmarking restrictions. This singular structure requires a different 
computational treatment from that used for the corresponding unconstrained 
Fay-Herriot models. To obtain a nonsingular representation of the Benchmarking 
prior, we express its stochastic component in the null space of the benchmarking
matrix $\boldsymbol W$. This follows the computational approach of
\citet{bolin2021efficient} for Gaussian distributions subject to linear
constraints. The resulting parameterization removes the singularity from 
posterior computation and guarantees that every posterior draw satisfies the 
benchmarking restrictions.

Specifically, define the kernel (or null space) of $\boldsymbol W$ as
\begin{equation}
\ker(\boldsymbol W)
=
\left\{
\boldsymbol z\in\mathbb R^m:
\boldsymbol W\boldsymbol z=\boldsymbol 0
\right\}.
\end{equation}
If $\operatorname{rank}(\boldsymbol W)=k$, then
$\dim \ker(\boldsymbol W)=m-k$. Let
$\boldsymbol B\in\mathbb R^{m\times(m-k)}$ have orthonormal columns
spanning $\ker(\boldsymbol W)$. Thus,
\[
\boldsymbol W\boldsymbol B=\boldsymbol 0,
\qquad
\boldsymbol B^\mathsf{T}\boldsymbol B=\boldsymbol I_{m-k}.
\]

Following \cite{bolin2021efficient},  the matrix $\boldsymbol B$ can be obtained from the singular value
decomposition of $\boldsymbol W$. Specifically, let
\begin{equation}\label{decomposition_1}
\boldsymbol W
=
\boldsymbol U
\begin{pmatrix}
\boldsymbol S & \boldsymbol 0
\end{pmatrix}
\boldsymbol V^{\mathsf T}, \;\;
\boldsymbol V
=
\begin{pmatrix}
\boldsymbol V_1 & \boldsymbol V_0
\end{pmatrix},
\end{equation}
where
\[
\bm U\in\mathbb{R}^{k\times k},\qquad
\bm S\in\mathbb{R}^{k\times k},\qquad
\bm V_1\in\mathbb{R}^{m\times k},\qquad
\bm V_0\in\mathbb{R}^{m\times(m-k)}.
\]
Here, $\bm U$ is an orthogonal matrix whose columns are the left
singular vectors of $\bm W$, $\bm S$ is a diagonal matrix containing the $k$
positive singular values of $\bm W$, and $\bm V=(\bm V_1\;\bm V_0)\in
\mathbb{R}^{m\times m}$ is an orthogonal matrix whose columns are
the right singular vectors of $\bm W$. In particular, $\bm V_1$ contains
the right singular vectors associated with the positive singular
values and $\bm V_0$ contains those associated with the zero
singular values. Consequently, the columns of $\bm V_0$ form an
orthonormal basis for $\ker(\boldsymbol W)$, so we can then take
$\boldsymbol B=\boldsymbol V_0$. It follows that  $\boldsymbol B \boldsymbol B^{\mathsf T}$ is the orthogonal projection matrix onto
$\ker(\boldsymbol W)$, with 
$$\boldsymbol B \boldsymbol B^{\mathsf T}=\bm I_m-\bm W^{\mathsf T}(\bm W \bm W^{\mathsf T})^{-1}\bm W.$$ 
For a given regularization matrix $\boldsymbol Q_\phi$, define
\begin{equation}
\boldsymbol K_\phi
=
\boldsymbol B^{\mathsf T}
\boldsymbol Q_\phi
\boldsymbol B.
\label{eq:Kphi}
\end{equation}
Since $\boldsymbol Q_\phi\succ 0$ and $\boldsymbol B$ has full column
rank, $\boldsymbol K_\phi\succ0$. The matrix
$\boldsymbol A_\phi$ in Lemma \ref{lemma1} can then be expressed as
\begin{equation}
\boldsymbol A_\phi
=
\boldsymbol B
\boldsymbol K_\phi^{-1}
\boldsymbol B^{\mathsf T}.
\label{eq:A-basis}
\end{equation}
This representation expresses the stochastic component of the
Benchmarking prior in the $(m-k)$-dimensional null space of
$\boldsymbol W$.

Using this representation, an equivalent reduced parameterization of
$\boldsymbol{\theta}^{\textsf{FH-B}}$ is
\begin{equation}
\boldsymbol{\theta}^{\textsf{FH-B}}
=
\boldsymbol A_\phi\boldsymbol X\boldsymbol\beta
+
\boldsymbol b_\phi
+
\boldsymbol B\boldsymbol\omega,
\label{eq:reduced}
\end{equation}
where $\boldsymbol\omega\in\mathbb R^{m-k}$ is the latent random effect
vector parameterizing variation within the null space of $\boldsymbol W$, with
\begin{equation}
\boldsymbol\omega
\mid
\boldsymbol X,\boldsymbol\beta,\sigma_u^2,\phi
\sim
\mathcal N_{m-k}
\left(
\boldsymbol 0,
\sigma_u^2\boldsymbol K_\phi^{-2}
\right).
\label{eq:reduced2}
\end{equation}
Indeed,
\[
\operatorname{Var}(\boldsymbol B\boldsymbol\omega)
=
\sigma_u^2
\boldsymbol B\boldsymbol K_\phi^{-2}\boldsymbol B^{\mathsf T}
=
\sigma_u^2
\boldsymbol A_\phi\boldsymbol A_\phi^{\mathsf T},
\]
so that \eqref{eq:reduced}--\eqref{eq:reduced2} recover the covariance
of the Benchmarking prior in Lemma \ref{lemma1}. 

Moreover, since 
$\boldsymbol W\boldsymbol A_\phi=\boldsymbol 0$,
$\boldsymbol W\boldsymbol b_\phi=\boldsymbol p$, and
$\boldsymbol W\boldsymbol B=\boldsymbol 0$,
the reduced parameterization also preserves the benchmarking restriction.
Thus, posterior computation can be carried out in terms of the
unconstrained $(m-k)$-dimensional vector $\boldsymbol\omega$ and the
benchmarking restriction is satisfied for every $\boldsymbol\omega$.

\subsection{Posterior computation}
\label{MCMC_algorithms}

To complete the Bayesian specification, we assign priors to
$\boldsymbol{\beta}$, $\sigma_u^2$, and $\phi$. We use an inverse-Gamma prior
\[
\sigma_u^2\sim
\operatorname{IG}(a_{\sigma_u^2},b_{\sigma_u^2})
\]
and a Beta prior
\[
\phi\sim\operatorname{Beta}(a_\phi,b_\phi).
\]
For the regression coefficients, we use the weakly informative prior
\[
\boldsymbol{\beta}
\mid \boldsymbol{\Omega}_{\boldsymbol{\beta}}
\sim
\mathcal{N}_p
\left(
\boldsymbol{0},
\boldsymbol{\Omega}_{\boldsymbol{\beta}}
\right),
\qquad
\boldsymbol{\Omega}_{\boldsymbol{\beta}}
=
\operatorname{diag}(\lambda_1^{-1},\ldots,\lambda_p^{-1}),
\]
where the scales $\lambda_j$ are assigned independent
Beta-Prime$(a_{\lambda_j},b_{\lambda_j})$ priors
(see, for instance, \citep{fuquene2025ultra}).
For posterior computation, we use the Gamma mixture representation
\[
\lambda_j\mid\rho_j
\sim\operatorname{Gamma}(a_{\lambda_j},\rho_j),
\qquad
\rho_j\sim\operatorname{Gamma}(b_{\lambda_j},1).
\]

The joint posterior distribution for the models in Table \ref{tab:models} using the equivalent reduced parameterization is given by

{
\begin{align} 
\label{eq:joint_d}
&p( \boldsymbol{\omega}, \boldsymbol{\beta}, \sigma_u^2, \phi \mid \by)
\propto
\exp
\left\{
-\frac12
\left(
\bm y-\bm b_{\phi}
-\bm A_{\phi}\bm X\bm\beta
-\bm B\bm \omega
\right)^T
\bm D^{-1}
\left(
\bm y-\bm b_{\phi}
-\bm A_{\phi}\bm X\bm\beta
-\bm B \boldsymbol{\omega}
\right)
\right\}
\nonumber\\
&\hspace{2cm}\quad\times
\mid  \boldsymbol{\Omega}_{\boldsymbol{\beta}} \mid^{-1/2}  \mid \bm K_{\phi} \mid (\sigma_u^2)^{-(m-k)/2} (\sigma_u^2)^{-(a_{\sigma_u^2}+1)}
\exp
\left\{
-\frac{1}{2}
\boldsymbol{\beta}^T
\boldsymbol{\Omega}_{\boldsymbol{\beta}}^{-1}
\boldsymbol{\beta}
\right\}
\nonumber\\
& \hspace{3cm} \times
\exp
\left\{
-\frac{1}{2\sigma_u^2}
\boldsymbol{\omega}^T
\bm K_{\phi} \bm K_{\phi}^{\top}
\boldsymbol{\omega}
\right\}
\exp\left(
-\frac{ b_{\sigma_u^2}}{\sigma_u^2}
\right)\pi(\phi)\pi(\boldsymbol{\Omega}_{\boldsymbol{\beta}}).
\end{align}
}

Algorithm \ref{mcmc:alg1} summarizes the MCMC procedure for posterior sampling
under the proposed models. The reduced representation in
\eqref{eq:reduced}--\eqref{eq:reduced2} yields Gibbs updates for
$\boldsymbol{\beta}$, $\sigma_u^2$, and $\boldsymbol{\omega}$.
The regularization parameter $\phi$ is updated using a random-walk
Metropolis step on the logit scale to ensure that $0<\phi<1$.
The proposal variance is adapted during burn-in and fixed thereafter.

By the reduced parameterization in Section \ref{computations}, every posterior
draw produced by Algorithm \ref{mcmc:alg1} satisfies the benchmarking
restriction. Unlike the posterior projection approach of
\citet{fuquenefully}, no additional benchmarking step is required after
posterior sampling. The main computational burden arises from computations associated
with the $(m-k)$-dimensional reduced system, particularly the factorization
required for sampling $\boldsymbol{\omega}$ and the evaluation of
$\boldsymbol K_\phi$ as $\phi$ is updated. The Single and Multi-View graph
Laplacians are constructed once, prior to MCMC (see Section \ref{Laplacian_computations}).

\begin{algorithm}[htbp]
 \small
\caption{MCMC algorithm under the proposed model with benchmarking priors}
\label{mcmc:alg1}
\begin{algorithmic}[1]

\setlength{\itemsep}{0.4em}
\State Construct $\boldsymbol{B}$, an orthonormal basis of
$\ker(\boldsymbol{W})$, and compute
$\boldsymbol{\Phi}_{\bm B}
=
\bm B^T\boldsymbol{\Phi}\bm B$.

\State Verify that
$\boldsymbol{W}\boldsymbol{A}_{\phi}=\boldsymbol{0}$ and
$\boldsymbol{W}\boldsymbol{b}_{\phi}=\boldsymbol{p}$.

\State Initialize $\boldsymbol{\beta}^{(0)}$,
$\boldsymbol{\omega}^{(0)}$,
$(\sigma_u^{2})^{(0)}$, $\boldsymbol{\Omega}_{\boldsymbol{\beta}}^{(0)}$, $\rho_j^{(0)}$ and $\phi^{(0)}$. Compute
\[
\bm K_{\phi^{(0)}}
=
\bm I_{m-k}
+
\frac{1-\phi^{(0)}}{\phi^{(0)}}
\boldsymbol{\Phi}_{\bm B},
\]
\[
\bm A_{\phi^{(0)}}
=
\bm B
\bm K_{\phi^{(0)}}^{-1}
\bm B^T,
\]
and $\boldsymbol{b}_{\phi^{(0)}}$.

\For{$s=1,\ldots,S$}

\State Draw
\[
\bm\beta^{(s)}
\sim
N_p\left(
\bm\mu_{\bm\beta^{(s)}},
\bm\Sigma_{\bm\beta^{(s)}}
\right),
\]
where
\begin{align*}
\bm\Sigma_{\bm\beta^{(s)}}
&=
\left(
\bm X^T
\bm A_{\phi^{(s-1)}}
\bm D^{-1}
\bm A_{\phi^{(s-1)}}
\bm X
+ (\boldsymbol{\Omega}_{\boldsymbol{\beta}}^{-1})^{(s-1)}\right)^{-1},\\
\bm\mu_{\bm\beta^{(s)}}
&=
\bm\Sigma_{\bm\beta^{(s)}}
\bm X^T
\bm A_{\phi^{(s-1)}}
\bm D^{-1}
\left(
\bm y
-\bm b_{\phi^{(s-1)}}
-\bm B\boldsymbol{\omega}^{(s-1)}
\right).
\end{align*}
\State Draw 
\begin{align*}
\lambda_j^{(s)}
&\sim
\operatorname{Gamma}
\left(
a_{\lambda_j}+\frac{1}{2},
\frac{1}{2}(\beta_j^{(s)})^2 + \rho_j ^{(s-1)}\right)  &
\rho_j^{(s)}
&\sim
\operatorname{Gamma}
\left(
a_{\lambda_j}+b_{\lambda_j},
1+ \lambda_j^{(s)} \right),
\end{align*}
for $j=1,\ldots,p$ and compute  $\boldsymbol{\Omega}_{\boldsymbol{\beta}}^{(s)}
=
\operatorname{diag}(1/\lambda_1^{(s)},\ldots,1/\lambda_p^{(s)})$.
\State Draw
\[
\boldsymbol{\omega}^{(s)}
\sim
N_{m-k}\left(
\bm\mu_{\boldsymbol{\omega}^{(s)}},
\bm\Sigma_{\boldsymbol{\omega}^{(s)}}
\right),
\]
where
\begin{align*}
\bm\Sigma_{\boldsymbol{\omega}^{(s)}}
&=
\left[
\bm B^T\bm D^{-1}\bm B
+
\frac{1}{(\sigma_u^2)^{(s-1)}}
\bm K_{\phi^{(s-1)}}
\bm K_{\phi^{(s-1)}}^T
\right]^{-1},\\
\bm\mu_{\boldsymbol{\omega}^{(s)}}
&=
\bm\Sigma_{\boldsymbol{\omega}^{(s)}}
\bm B^T\bm D^{-1}
\left(
\bm y
-\bm b_{\phi^{(s-1)}}
-\bm A_{\phi^{(s-1)}}
\bm X\bm\beta^{(s)}
\right).
\end{align*}

\State Draw
\[
(\sigma_u^2)^{(s)}
\sim
\operatorname{Inverse-Gamma}
\left(
a_{\sigma_u^2}+\frac{m-k}{2},
\;
b_{\sigma_u^2}
+
\frac12
(\bm\omega^{(s)})^T
\bm K_{\phi^{(s-1)}}
\bm K_{\phi^{(s-1)}}^T
\bm\omega^{(s)}
\right).
\]

\State Update $\phi^{(s)}$ using a random-walk Metropolis step on the logit 
scale, with log-posterior contribution
\[
\ell(\phi)
=
\log|\bm K_{\phi}|
-
\frac12
\bm r_{\phi}^T
\bm D^{-1}
\bm r_{\phi}
-
\frac{1}{2(\sigma_u^2)^{(s)}}
(\bm\omega^{(s)})^T
\bm K_{\phi}
\bm K_{\phi}^T
\bm\omega^{(s)}
+
\log\left\{
\operatorname{Beta}
(\phi;a_{\phi},b_{\phi})
\right\},
\]
where
\[
\bm r_{\phi}
=
\bm y
-
\bm b_{\phi}
-
\bm A_{\phi}
\bm X\bm\beta^{(s)}
-
\bm B\bm\omega^{(s)}.
\]
Update $\bm K_{\phi^{(s)}}$, $\bm A_{\phi^{(s)}}$, and
$\bm b_{\phi^{(s)}}$.

\State Compute
\[
\bdt^{\textsf{FH-B},(s)}
=
\bm A_{\phi^{(s)}}
\bm X\bm\beta^{(s)}
+
\bm b_{\phi^{(s)}}
+
\bm B\boldsymbol{\omega}^{(s)}.
\]
\EndFor
\end{algorithmic}
\end{algorithm}

\subsection{Algorithms for Normalized Similarity Graph Laplacians with external covariates}    
\label{Laplacian_computations}

To construct the Single Normalized Graph Laplacian 
$\textsf{L}_{C^{*}}^{\textsf{S}}$ and Multi-View Normalized Graph Laplacian 
$\textsf{L}_{C^{*}}^{\textsf{MV}}$ introduced in Sections 
\ref{subsec:The_priors2} and \ref{subsec:multi_view}, we propose 
Algorithms \ref{alg:SC} and \ref{alg:SC2}. Both algorithms construct connected 
neighborhood graphs from pairwise similarities among the $m$ small areas based on 
the direct estimates and external auxiliary information, $\boldsymbol{x}^{*}_i$, for area $i$. 
In contrast to approaches that use graph Laplacians to obtain discrete groups 
through spectral clustering, such as \citet{fuquenefully}, the proposed 
algorithms retain the weighted pairwise similarity structure among areas 
without introducing an intermediate clustering step. The neighborhood size 
$C$ is selected adaptively as the smallest value that produces a connected 
graph, instead of being specified a priori.

We first describe the construction of the Single Graph Laplacian in Algorithm \ref{alg:SC}. Given the direct estimates $y_i$, a single external covariate $x_i^*$ (not included among the regression covariates in the 
Fay-Herriot model), and a maximum neighborhood size $C_{\max}$, we construct a symmetric 
$C$-nearest-neighbor graph and its corresponding Laplacian graph for each 
$C=1,\ldots,C_{\max}$. For each candidate graph Laplacian, we calculate the ordered eigenvalues
\[
0=\lambda_1(C)
\leq
\lambda_2(C)
\leq\cdots\leq
\lambda_m(C)
\]
For a nonnegative undirected graph, the graph is connected if and only if 
$\lambda_2(C)>0$. Accounting for numerical precision through a small tolerance 
$\varepsilon>0$, we select
\begin{align}
\label{C_start}
C^*
=
\min
\left\{
C:
\lambda_2(C)>\varepsilon
\;\text{and}\;
\lambda_2(C+1)>\varepsilon
\right\}.
\end{align}
Thus, $C^*$ is the smallest neighborhood size for which the graph is connected 
under the numerical tolerance and remains connected after increasing the 
neighborhood size by one. The Laplacian, $\textsf{L}_{C^{*}}^{\textsf{S}}$, corresponding to the selected neighborhood 
size $C^*$ is then used to specify the FH-L and FH-L-BP models in Table \ref{tab:models}.

\begin{algorithm}[htbp] 
\small
  \begin{algorithmic}
  \setlength{\itemsep}{0.4em}
    \State \textbf{Input: direct estimate $y_{i}$, external covariate $x^{*}_{i}$ for $i=1,...,m$,
maximum neighborhood size $C_{\max}$
and numerical tolerance $\varepsilon$.}

\State Compute the similarity matrix using the external covariate $\boldsymbol{x}^{*}$,  
$\textsf{S}=(s_{i,\ell})_{i,\ell=1,...,m} \in \mathbb{R}^{m \times m}$.

\For{$C=1,\ldots,C_{\max}$}
\begin{enumerate}[label=(\arabic*)]
    \item Compute the $C$ nearest neighbors
    $\mathcal{N}_C(i)$ for every $i=1,\ldots,m$ where $\mathcal{N}_C(i)$ denotes the set of the $C$ units having the largest
similarities with unit $i$.     
    \item Construct the symmetric adjacency matrix $\textsf{A}_{C}^{\textsf{sym}}=(a_{i\ell}^C)$ where $
    a_{i\ell}^C
    =
    \mathbb{I}
    \left\{
    \ell\in\mathcal{N}_C(i)
    \;\text{or}\;
    i\in\mathcal{N}_C(\ell)
    \right\}$.
    \item Construct the weighted adjacency matrix $\textsf{W}_{C}^{\textsf{S}}=\textsf{A}_{C}^{\textsf{sym}} \odot \textsf{S}$.      
    \item Compute the degree matrix
    $\textsf{D}_{C}^{\textsf{S}}=\operatorname{diag}
    \left(
    \textsf{W}_{C}^{\textsf{S}}\mathbf{1}_m
    \right)$.    
    \item Compute $\textsf{L}_{C}^{\textsf{S}}
=
\boldsymbol{I}_m- (\textsf{D}_{C}^{\textsf{S}})^{-1/2}\textsf{W}_{C}^{\textsf{S}} (\textsf{D}_{C}^{\textsf{S}})^{-1/2}$.
    \item Compute the second-smallest eigenvalue
    $\lambda_2(C)$ of $\textsf{L}_{C}^{\textsf{S}}$.    
\end{enumerate}
\EndFor

\State Select $C^{*}$ using $C^*
=
\min
\left\{
C:
\lambda_2(C)>\varepsilon
\;\text{and}\;
\lambda_2(C+1)>\varepsilon
\right\}.$

\State \textbf{Output: Laplacian matrix, $\textsf{L}_{C^{*}}^{\textsf{S}}
=
\boldsymbol{I}_m- (\textsf{D}_{C^{*}}^{\textsf{S}})^{-1/2}\textsf{W}_{C^{*}}^{\textsf{S}} (\textsf{D}_{C^{*}}^{\textsf{S}})^{-1/2}.$}
  \end{algorithmic}
    \caption{Construction of the symmetric Normalized Graph Laplacian for the Laplacian Benchmarking Prior.}
     \label{alg:SC}
\end{algorithm}

A similar process is carried out to construct the Multi-View Laplacian, which allows multiple external covariates to be incorporated 
as a separate view.  For each candidate neighborhood size $C=1,\ldots,C_{\max}$, we construct 
a symmetric $C$-nearest-neighbor graph for each view. To select a common 
neighborhood size $C^*$ across views, we consider the union of the 
view-specific neighborhood graphs. 

Specifically, for each $C$, let $\textsf{A}^{\textsf{MV}}_{C}=(a_{i\ell,C}^{\textsf{MV}})$ denote the adjacency matrix of the union graph, where $a_{i\ell,C}^{\textsf{MV}}=1$ if areas $i$ and $\ell$ are neighbors in at least one view, and $a_{i\ell,C}^{\textsf{MV}}=0$ otherwise. Let $n(C)$ denote the number of connected components of the union graph 
$\textsf{A}^{\textsf{MV}}_{C}$, which is connected if and only if 
$n(C)=1$. We select the common neighborhood size as
\begin{align}
\label{C_start_MV}
C^{*}
=
\min
\left\{
C:
n(C)=1
\;\text{and}\;
n(C+1)=1
\right\}.
\end{align}
As in Algorithm \ref{alg:SC}, $C^*$ is the smallest 
neighborhood size for which the union graph is connected and remains 
connected after increasing $C$ by one. Given $C^*$, for each view $j$ we construct the weighted adjacency matrix, the corresponding degree matrix, and the view-specific normalized graph Laplacian. Finally, the Laplacian $\textsf{L}_{C^*}^{\textsf{MV}}$ is obtained as a weighted sum of the view-specific normalized graph Laplacians. This result is used to specify the FH-MVL and FH-MVL-BP models in Table \ref{tab:models}.

\begin{algorithm}[htbp] 
\small
  \begin{algorithmic}
    \State \textbf{Input: direct estimate $y_{i}$ for $i=1,...,m$, external covariates $\boldsymbol{x}^{*}_{j}=(x_{1,j}^{*},...,x_{m,j}^{*})$ for $j=1,...,J$, weights $q_1,\ldots,q_J$,
and maximum neighborhood size $C_{\max}$.}

\For{$j=1,\ldots,J$}
    \State Compute the similarity matrix for each external covariate
$x_j^*$, $\textsf{S}_j=(s_{i,\ell}^j)_{i,\ell=1,\ldots,m}
\in\mathbb{R}^{m\times m}$.
\EndFor

\For{$C=1,\ldots,C_{\max}$}
 \begin{enumerate}[label=(\arabic*)]
\item For each $j=1,\ldots,J$:
\begin{enumerate}[label=(\alph*)]
    \item Determine the $C$ nearest neighbors
        $\mathcal{N}_{j,C}(i)$ for each unit $i$.
    \item Construct the symmetric adjacency matrix $\textsf{A}^{\textsf{MV}}_{j,C}=(a_{i\ell}^{(j,C)})$ where
        $
        a_{i\ell}^{(j,C)}
        =
        \mathbb{I}
        \left\{
        \ell\in\mathcal{N}_{j,C}(i)
        \;\text{or}\;
        i\in\mathcal{N}_{j,C}(\ell)
        \right\}.
        $
\end{enumerate}

\item Construct the union adjacency matrix $
\textsf{A}^{\textsf{MV}}_{C}
=
\mathbb{I}
\left\{
\sum_{j=1}^{J} \textsf{A}^{\textsf{MV}}_{j,C} >0
\right\}.
$

\item Compute the number of connected components
    $n(C)$ of $\textsf{A}^{\textsf{MV}}_{C}$.
 \end{enumerate}
\EndFor

\State Select $C^{*}=
\min
\left\{
C:
n(C)=1
\;\text{and}\;
n(C+1)=1
\right\}$.

\For{$j=1,\ldots,J$}
  \begin{enumerate}[label=(\alph*)]
   \item Construct $\textsf{W}^{\textsf{MV}}_{j,C^{*}}= \textsf{A}^{\textsf{MV}}_{j,C^{*}} \odot \textsf{S}_j$.
   \item Construct $\textsf{D}^{\textsf{MV}}_{j,C^{*}}= \operatorname{diag}
    \left(
\textsf{W}^{\textsf{MV}}_{j,C^{*}}\mathbf{1}_m
    \right)$.
    \item Compute the view-specific normalized Laplacian 
    \[\textsf{L}^{\textsf{S}}_{C^{*},j} = \boldsymbol{I}_m - (\textsf{D}^{\textsf{MV}}_{j,C^{*}})^{-1/2} \textsf{W}^{\textsf{MV}}_{j,C^{*}} (\textsf{D}^{\textsf{MV}}_{j,C^{*}})^{-1/2}\].
 \end{enumerate}
 \EndFor    
 
 \vspace{-1.5em}
 
\State \textbf{Output: Multi-View Normalized Laplacian $\textsf{L}_{C^{*}}^{\textsf{MV}}=
\sum_{j=1}^{J}q_j \textsf{L}_{C^{*},j}^{\textsf{S}}$.}

  \end{algorithmic}
    \caption{Construction of the Multi-View Normalized Laplacian for the Multi-View Laplacian Benchmarking Prior.}
     \label{alg:SC2}
\end{algorithm}

\section{Application: Average Household Size (AHS) in the municipalities of Colombia} \label{sec_Application}

The Average Household Size (AHS) is a key indicator of social organization \citep{esteve2024global}. Reliable estimates of AHS are important for official statistics because household size and composition affect how resources, income, and opportunities are distributed across individuals. Despite its importance for policy and development planning, reliable estimates of AHS are often unavailable at fine geographic scales during long intercensal periods, particularly in low- and middle-income countries. This is the case in Colombia, where the National Statistical Office (NSO) and the Ministry of Health currently produce AHS estimates only at the national or departmental level. National-level estimates are also produced to meet United Nations reporting requirements \citep{UN2026}. While such aggregates are informative, estimates at finer subnational resolution such as municipalities, rural versus urban areas, or socioeconomic strata are essential for targeted policy interventions and equitable resource allocation. Motivated by this need for finer geographic resolution, our case study focuses on estimating AHS for \(m=291\) municipalities in Colombia.

We draw on data from the most recent Demographic and Health Survey (DHS) conducted in Colombia in 2025 \citep{DHS2025}. For each municipality $i$, we compute a direct estimator of AHS ($y_i$) using standard direct survey estimators \citep{Hajek1971}. The Households Per Person (HPP) (\(x_i\)), obtained from the 2018 Colombia Population Census \citep{DANE2018}, exhibits a relatively strong linear relationship with AHS and is therefore selected as a covariate for the municipality-level Fay-Herriot model (left panel of Figure \ref{fig:data}).
As is typical in SAE settings, direct estimates of variances are highly variable for many municipalities due to limited local sample sizes. Therefore the direct variance estimates are smoothed via the Generalized Variance Function (GVF) method \citep{wolter1985introduction}. Specifically, to obtain the estimates ($D_i$) of the 
variances of the AHS, we implement the GVF approach of \cite{ you2023application}. We fit a log-linear model relating the direct variance estimates of AHS to the municipality-level sample sizes, using the logarithm of each quantity.

\begin{figure}[h!]
\centering
\begin{minipage}[t]{0.48\textwidth}
  \centering
  \includegraphics[width=\textwidth]{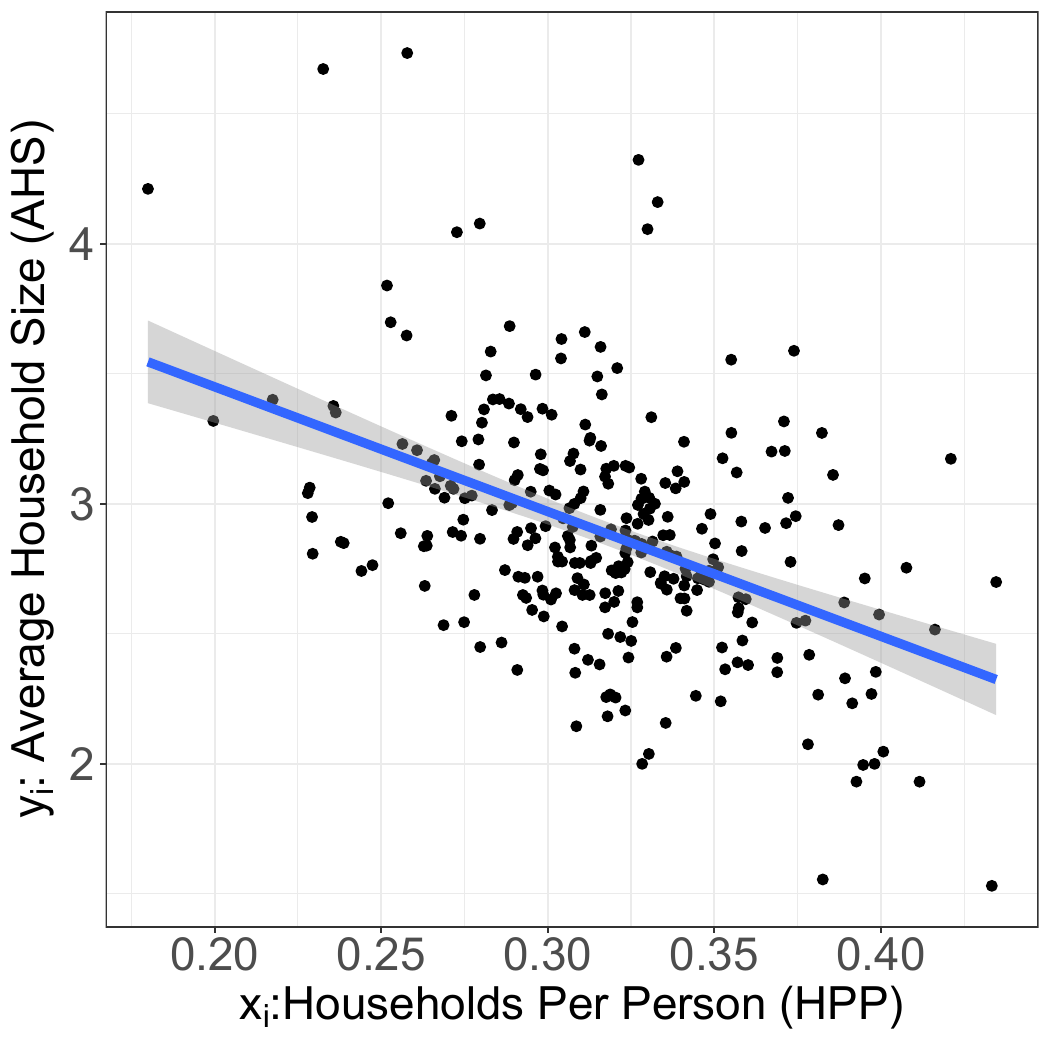}
\end{minipage}
  \hspace{-0.1cm}
\begin{minipage}[t]{0.48\textwidth}
  \centering
  \includegraphics[width=\textwidth]{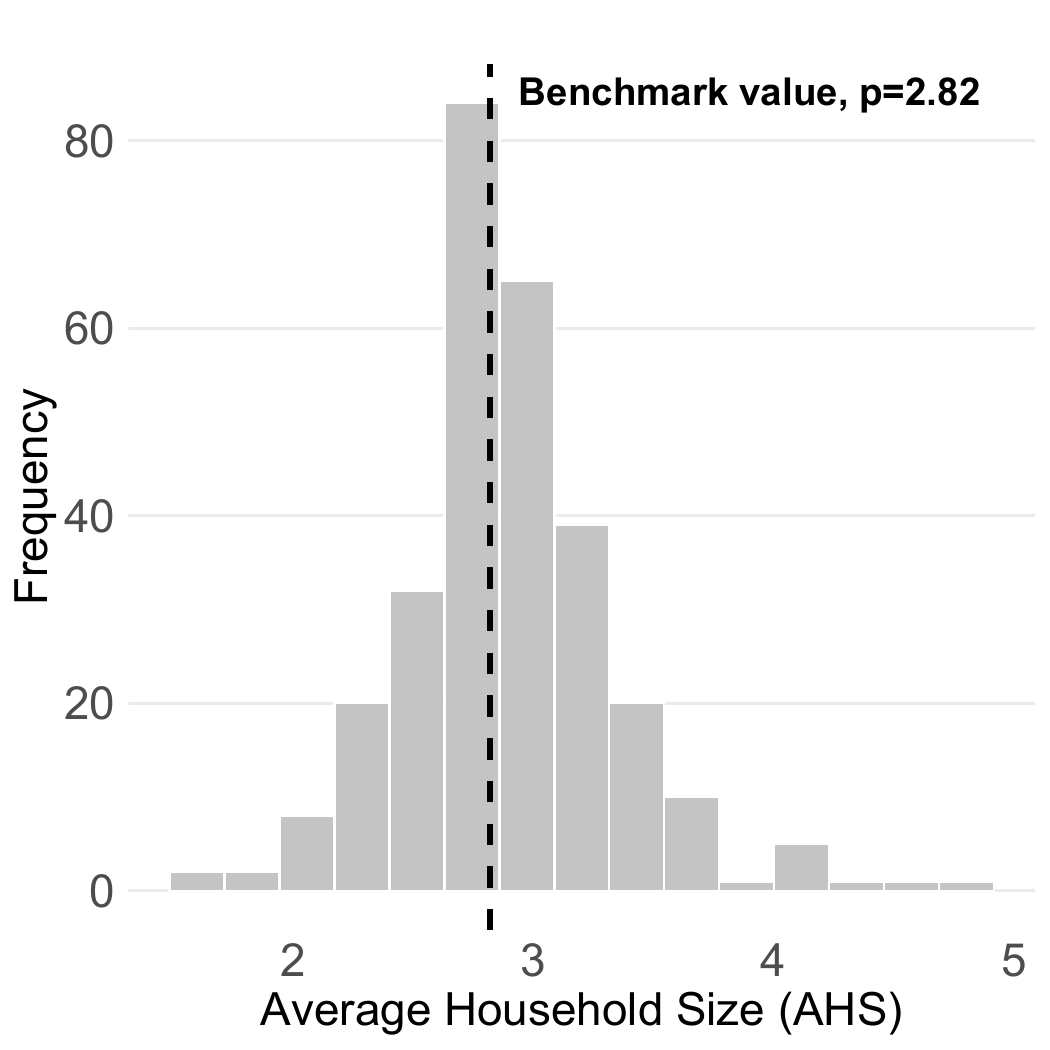}
\end{minipage}
\caption{\small Direct estimates and covariate values:  $y_{i}$: Average Household Size (AHS) vs $x_{i}$: Households Per Person (HPP) (left). Direct estimates of AHS and benchmark value according to the 2025 National Quality of Life Survey in Colombia.
}\label{fig:data}
\end{figure}

According to the 2025 National Quality of Life Survey (NQLS) in Colombia \citep{DANE_ECV_2025}, the national AHS was estimated at 2.82, as shown in the right panel of Figure \ref{fig:data}. We therefore impose a benchmarking constraint requiring the weighted municipality-level estimates to agree with this national estimate. Specifically, we set $\boldsymbol{w}^T \bdt =2.82$, where $ \boldsymbol{w}$ denotes the vector of weights used to aggregate municipality-level AHS estimates to the national level.

Substantial prior evidence suggests that the household size in Colombia is closely tied to underlying poverty conditions \citep{nunez2005determinantes,Profamilia}, a pattern also documented in other low- and middle-income countries, including Pakistan \citep{mahmood2026multidimensional} and Mexico \citep{garza2016determinants}. In particular, AHS is expected to be strongly associated with  Urban and Rural Multidimensional Poverty Indexes (U-MPI and R-MPI) \citep{alkire2015multidimensional, FAO_OPHI_2022_RMPI}. These relationships motivate the use of poverty measures as auxiliary information for constructing similarity structures among municipalities to facilitate borrowing of strength across areas with similar socioeconomic characteristics.

\subsection{Normalized Similarity Graph Laplacians with Multidimensional Poverty Indexes}

Motivated by the association between AHS and poverty indicators, we consider auxiliary information derived from the Multidimensional Poverty Index (MPI) \citep{alkire2015multidimensional}, available from the 2005 National Census \citep{DANE2005} and the 2014 Census of Agriculture \citep{CNA2014}. To implement the FH-L and FH-L-BP models, Algorithm \ref{alg:SC} requires an external auxiliary covariate to construct the normalized graph Laplacian. Both the Urban MPI (U-MPI) and Rural MPI (R-MPI) provide potential sources of auxiliary information for estimating AHS at the municipality level. Since several Colombian rural municipalities are included in the DHS survey \citep{DHS2025}, we use R-MPI and the direct AHS estimates as input to Algorithm \ref{alg:SC}. For the FH-MVL and FH-MVL-BP models, we use both R-MPI and U-MPI as
external auxiliary covariates to construct the Multi-View Normalized Graph
Laplacian using Algorithm \ref{alg:SC2}, with weights $q_1=q_2=0.5$. This allows the Multi-View construction to incorporate complementary socioeconomic information from the two poverty measures.

To construct the similarity graphs in Algorithms \ref{alg:SC} and \ref{alg:SC2}, we use a radial-kernel Gram matrix \citep{gonen2011multiple}, following the specifications in \citet{belkin2006manifold} and \citet{von2007tutorial}. For view $j$, the similarity between municipalities $i$ and $\ell$ is defined as
\[
s_{i,\ell}^{j}
=
\exp\left\{
-\frac{(x^{*}_{i,j}-x^{*}_{\ell,j})^{2}+(y_i-y_\ell)^{2}}
{2\sigma_s^2}
\right\},
\qquad \sigma_s^2=1.
\]

To select the neighborhood size \(C^*\) in Algorithm \ref{alg:SC}, we set the tolerance to \(\varepsilon=10^{-8}\). For the Multi-View graph in Algorithm \ref{alg:SC2}, the neighborhood size is selected based on the connectivity of the union graph induced by U-MPI and R-MPI. Figure \ref{fig:network} displays the similarity graphs obtained from Algorithms \ref{alg:SC} and \ref{alg:SC2}, with selected neighborhood sizes $C^{*}=3$ and $C^{*}=2$ for the Single and Multi-View constructions, respectively. Because the two views can provide complementary connections among municipalities, connectivity of the union graph is achieved at the smaller neighborhood size $C^{*}=2$.

\begin{figure}[h!]
\centering
 \hspace{-1cm}
\begin{minipage}[t]{0.55\textwidth}
  \hspace{1.5cm} Single, $C^{*}=3$ \vspace{-1.0cm}
\\
  \centering
  \includegraphics[width=\textwidth]{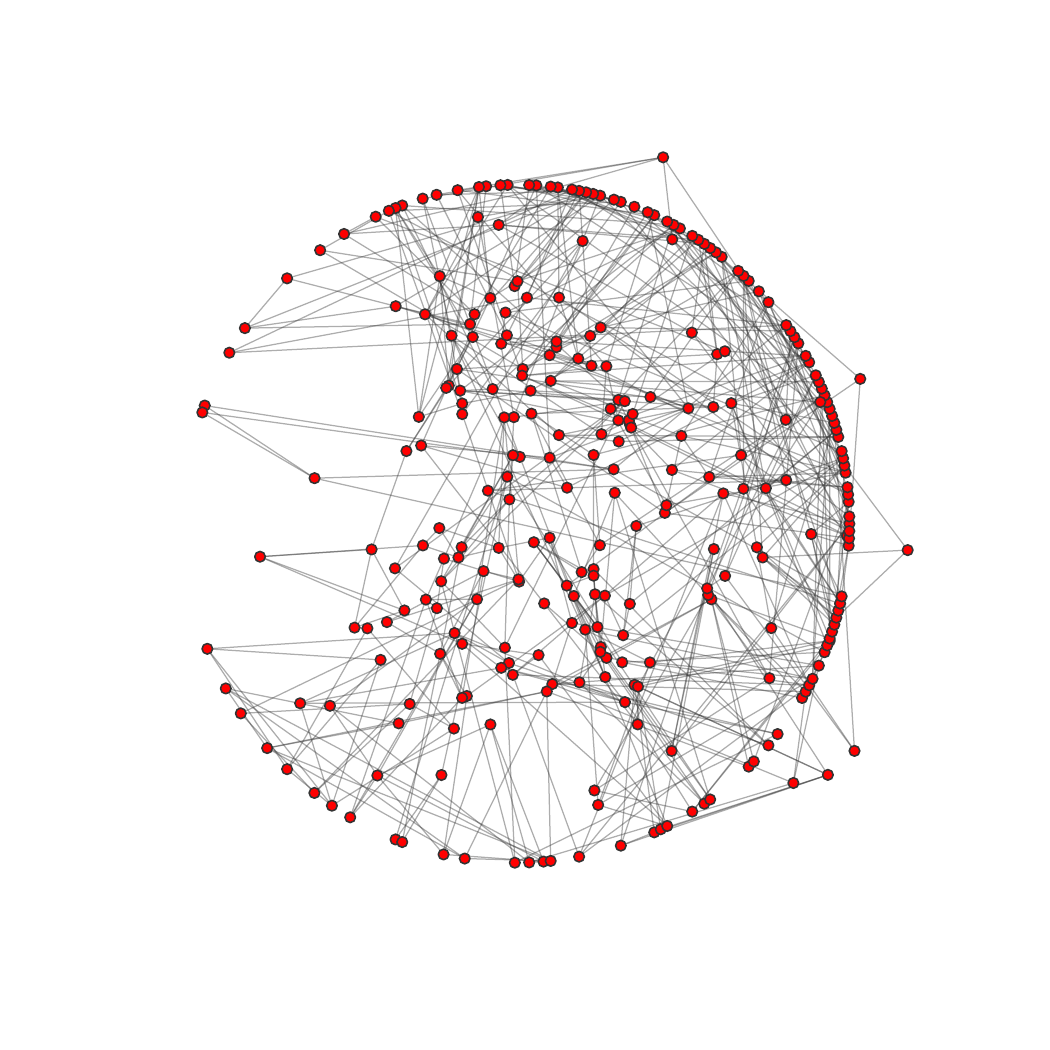}
\end{minipage}
  \hspace{-2.2cm}
\begin{minipage}[t]{0.55\textwidth}
  \hspace{2.0cm} Multi-View, $C^{*}=2$ \vspace{-1.0cm}
\\
  \centering
  \includegraphics[width=\textwidth]{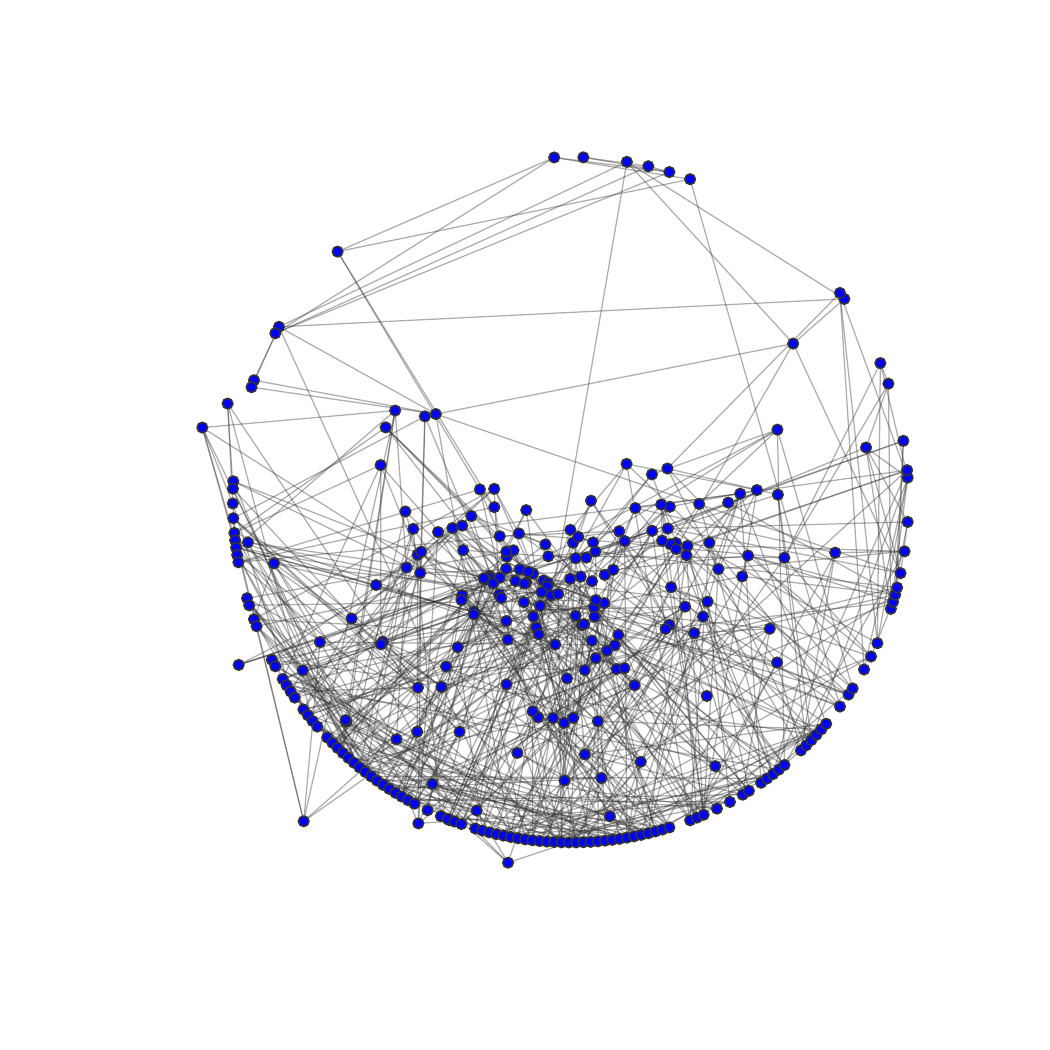}
\end{minipage}
\vspace{-1.3cm}
\caption{\small
Similarity graphs constructed from the direct AHS estimates and external auxiliary information. Single graph obtained from Algorithm \ref{alg:SC} using R-MPI (left) and Multi-View graph obtained from Algorithm \ref{alg:SC2} using R-MPI and U-MPI (right).}\label{fig:network}
\vspace{-0.5cm}
\end{figure}

\subsection{Simulation study and evaluation metrics}
\label{sec:evaluation}

To inform the model evaluation in Section \ref{sec:results}, we first
conduct a data-based simulation study. The simulation provides a supporting
assessment of the evaluation metrics used in the application and is not a
central component of the methodological development. We therefore provide
the full simulation design and results in Supplementary Section
\ref{sec_Simulations} and summarize the main findings here.

The simulation was designed to resemble the motivating application while
retaining known small area parameters. It uses covariates and other features of the Colombian
application under different data-generating models. This setting allows us to compare the performance of the models when the true small area parameters are known, examine whether conclusions based on the
Deviance Information Criterion (DIC) \citep{spiegelhalter2002bayesian, tang2018modeling, tang2026measurement} are consistent with these truth-based comparisons, and assess whether the estimated Conditional Posterior Mean Square Error
(CPMSE) provides a reasonable approximation to the empirical MSE. The
simulation results support the use of DIC for model assessment and CPMSE for
evaluating the precision of small area estimates in the application, where the
true municipality-level AHS parameters are unknown.

Specifically, to quantify uncertainty in the benchmarked small area estimators, we use
the Conditional Posterior Mean Square Error (CPMSE) proposed by
\citet{fuquenefully}. Let $\mathcal{M}$ denote the corresponding
unbenchmarked model, FH, FH-L, or FH-MVL, and let
$\hat{\theta}_i^{\mathcal{M}\textsf{-B-RB}}$ and
$\hat{\theta}_i^{\mathcal{M}\textsf{-RB}}$ denote the Rao-Blackwell
estimators under the benchmarked and unbenchmarked models, respectively.
The CPMSE of the benchmarked estimator is defined as 
\begin{align}
\operatorname{CPMSE}
\left(
\hat{\theta}_i^{\mathcal{M}\textsf{-B-RB}}
\mid \boldsymbol y,\boldsymbol X
\right)
&=
\operatorname{CPMSE}
\left(
\hat{\theta}_i^{\mathcal{M}\textsf{-RB}}
\mid \boldsymbol y,\boldsymbol X
\right)
+
\left(
\hat{\theta}_i^{\mathcal{M}\textsf{-B-RB}}
-
\hat{\theta}_i^{\mathcal{M}\textsf{-RB}}
\right)^2.
\label{eq:cpmse}
\end{align}
Thus, the CPMSE combines the uncertainty associated with the corresponding
unbenchmarked small area estimator with the squared discrepancy introduced
by enforcing the benchmarking constraint. We use the estimated CPMSE to
compute the coefficient of variation for each benchmarked small area
estimate as
\begin{equation}
\widehat{\operatorname{CV}}_i
=
\frac{
\sqrt{
\widehat{\operatorname{CPMSE}}
\left(
\hat{\theta}_i^{\mathcal{M}\textsf{-B-RB}}
\mid \boldsymbol y,\boldsymbol X
\right)
}
}{
\hat{\theta}_i^{\mathcal{M}\textsf{-B-RB}}
}.
\end{equation}
Similarly, the coefficient of variation can be obtained for $\hat{\theta}_i^{\mathcal{M}\textsf{-RB}}$.  The Rao-Blackwell estimators and Monte Carlo computation of the CPMSE are
provided in Supplementary Section \ref{sec:CPMSE}.

To compare the performance of the different proposed Bayesian models we consider the DIC  given by 
\begin{align}
\text{DIC} =  \dfrac{2}{S} \sum_{t=1}^{S}\sum_{i=1}^{m} \dfrac{(y_{i}-\theta_i^{\textsf{FH-B},(s)})^2}{D_{i}} -\sum_{i=1}^{m} \dfrac{(y_{i}-\hat{\theta}_i^{\textsf{FH-B},(s)})^2}{D_{i}},
\end{align}
where $\theta_i^{\textsf{FH-B},(s)}$ is a posterior sample obtained using Algorithm \ref{mcmc:alg1} under  
FH-BP, FH-L-BP or FH-MVL-BP and  $\hat{\theta}_i^{\textsf{FH-B},(s)} = \sum_{s=1}^{S}\theta_i^{\textsf{FH-B},(s)}/S$ is the posterior sample  mean of $\theta_{i}$. 
Similarly, DIC 
can be computed under FH, FH-L or FH-MVL. 
Among all candidate models, the model with the smallest DIC is selected as the best-fitting model for the final analysis. 
\subsection{Results} \label{sec:results}

In this section, we evaluate the proposed models for the AHS application. To obtain posterior samples, we implement Algorithm \ref{mcmc:alg1}. For each model in Table \ref{tab:models}, we run two MCMC chains for 40,000 iterations, discard the first 40,000 iterations as burn-in, and retain every fourth draw thereafter. We use weakly informative priors for $\lambda_j$, with $a_{\lambda_j}=b_{\lambda_j}=1$, a non-informative prior for $\phi$ with $a_\phi=b_\phi=1.01$, and small hyperparameter values for the prior of $\sigma_u^2$, $a_{\sigma_u^2}=b_{\sigma_u^2}=0.0001$. Convergence is satisfactory across all model parameters. See Supplementary Section \ref{MCMC_diagnostic} for convergence diagnostics under the  FH-MVL-BP model.

Table \ref{tab:EPD} displays the DIC values for the six models considered in the AHS application. These results favor the models  that incorporate graph-based regularization. Among the
models without benchmarking, both FH-L and
FH-MVL have smaller DIC values than FH, with FH-MVL attaining the smallest
value. The same pattern is observed among the benchmarked models. FH-L-BP and
FH-MVL-BP have smaller DIC values than FH-BP, with FH-MVL-BP displaying the
smallest value. These results indicate that incorporating the graph-based 
similarity structure improves model performance in this application, 
with further improvement when both R-MPI and U-MPI are incorporated through the Multi-View specification. 

\begin{table}[htpb]
\renewcommand{\arraystretch}{1.3}
\centering
\begin{tabular}{clcclccc} 
  \hline
 & & DIC &  $\#(\widehat{\operatorname{CV}}\%)<7$ &  &  & \\
  \hline \hline  
 & Direct & --------- & 191 & & \\
&  FH &   516.9  &      222   && \\
&  FH-BP &    579.9 &           206   && \\ \hline\hline 
&  FH-L &      446.1&        254     && \\
& FH-L-BP &   512.3 &      235   &&\\ \hline\hline 
& FH-MVL &   434.7  &      262     &&\\
&  FH-MVL-BP &  500.0 &     242  && \\ \hline
\end{tabular}
\caption{\small Model assessment and precision of the AHS estimates under the
FH, FH-L, FH-MVL, FH-BP, FH-L-BP, and FH-MVL-BP models. Reported are the DIC
and the number of municipalities (out of 291) with estimated coefficients of variation below
7\%. The corresponding number for the direct estimates is included for
comparison.}
\label{tab:EPD}
\end{table}

Figure \ref{fig:phi} displays the marginal posterior distributions of the
regularization parameter $\phi$ for the models incorporating graph
Laplacians. The posterior distributions under FH-MVL and FH-MVL-BP are
concentrated at smaller values of $\phi$ than those under FH-L and
FH-L-BP, respectively. Since smaller values of $\phi$ correspond to
greater regularization, these results indicate a stronger influence of
the regularization structure under the Multi-View specifications.

\begin{figure}[htpb]
\begin{center}
\begin{tabular}{ccc}
 \hspace{-0.2cm}  \includegraphics[width=0.47\textwidth]{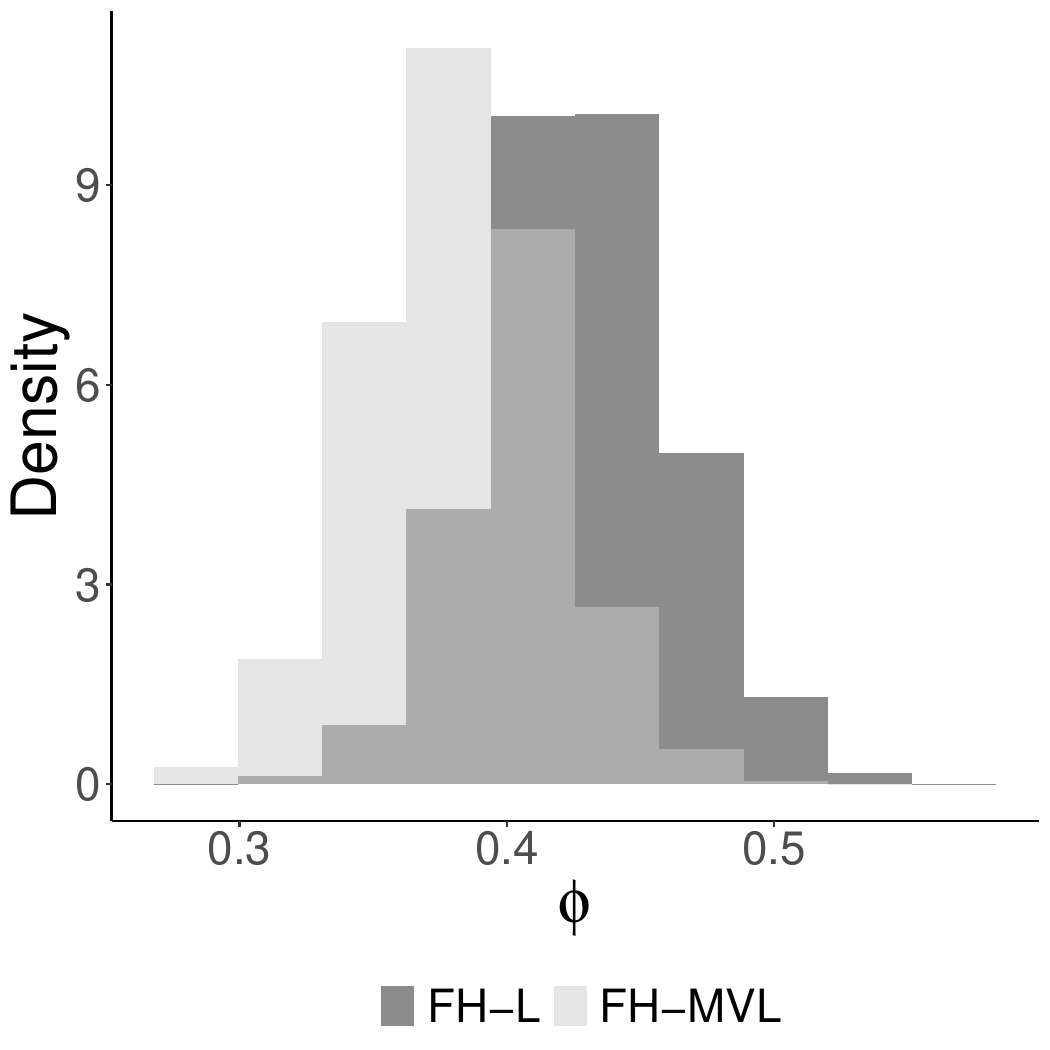} & \includegraphics[width=0.47\textwidth]{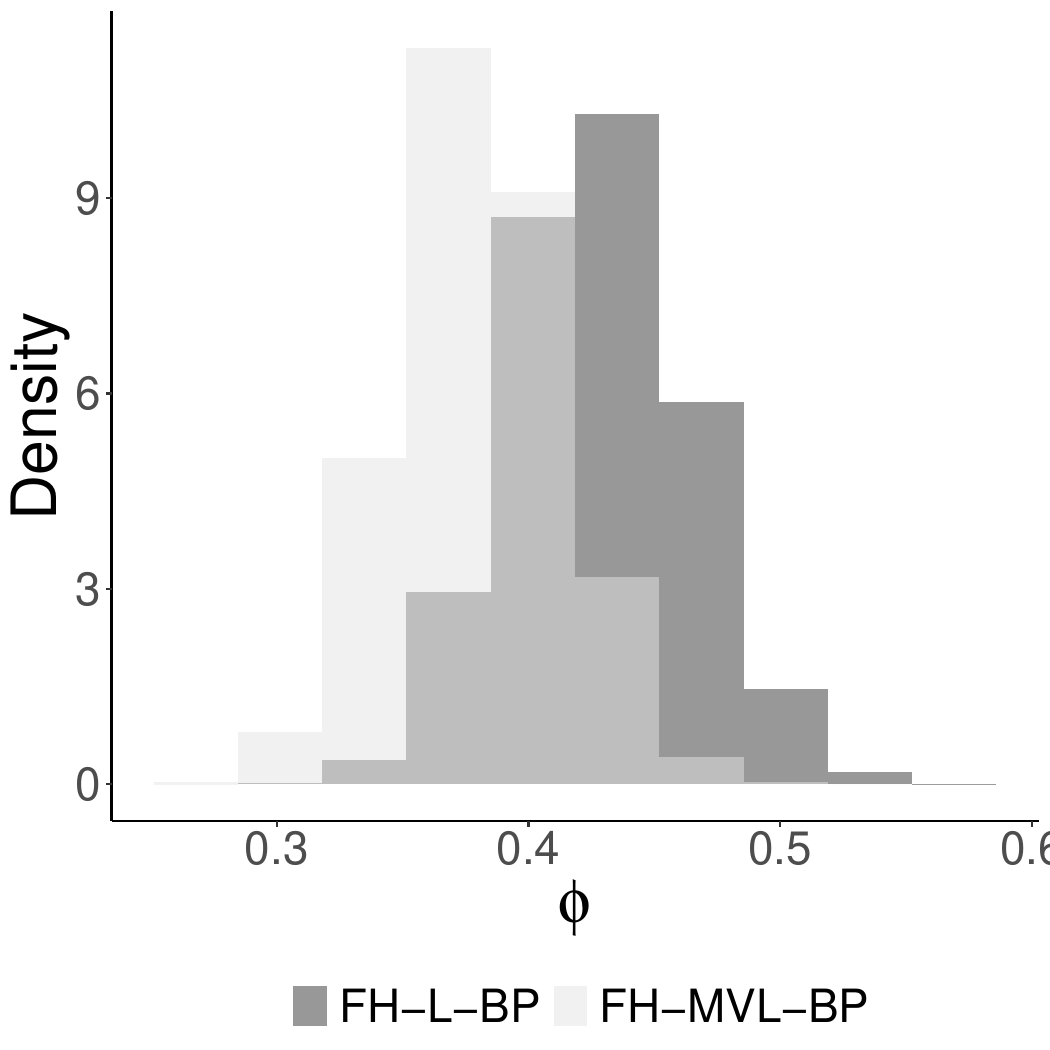}  \\
\end{tabular}
\end{center}
\vspace{-0.3cm}
\caption{\small
Posterior estimates of $\phi$ under FH-L and FH-MVL (left) and under   FH-L-BP and  FH-MVL-BP (right).} \label{fig:phi}
\end{figure}
  
We next evaluate the uncertainty and precision of the small area
estimators under the models in Table \ref{tab:models}. Figure \ref{fig:result1} displays the Conditional Posterior Mean Square Error (CPMSE) and the respective 
estimated Coefficients of Variation (CVs) for the benchmarked municipality-level AHS estimates. 
According to Figure \ref{fig:result1} (left), the unbenchmarked models
FH, FH-L, and FH-MVL generally produce smaller posterior uncertainty
measures than their benchmarked counterparts, FH-BP, FH-L-BP, and
FH-MVL-BP. This increase is expected because the CPMSE for the
benchmarked estimators accounts for the additional discrepancy induced
by enforcing the benchmarking constraint. Despite this increase,
Figure \ref{fig:result1} (right) shows that incorporating graph-based
regularization improves the precision of the small area estimates.
In particular, FH-MVL and FH-MVL-BP generally produce smaller
coefficients of variation than the corresponding models without
Laplacian or Multi-View regularization. Overall, the results suggest that incorporating multiple sources of similarity
through the Multi-View Laplacian can improve the precision of the AHS
small area estimates.

\begin{figure}[htpb]
\begin{center}
\begin{tabular}{ccc}
\hspace{-0.5cm} \includegraphics[width=0.48\textwidth]{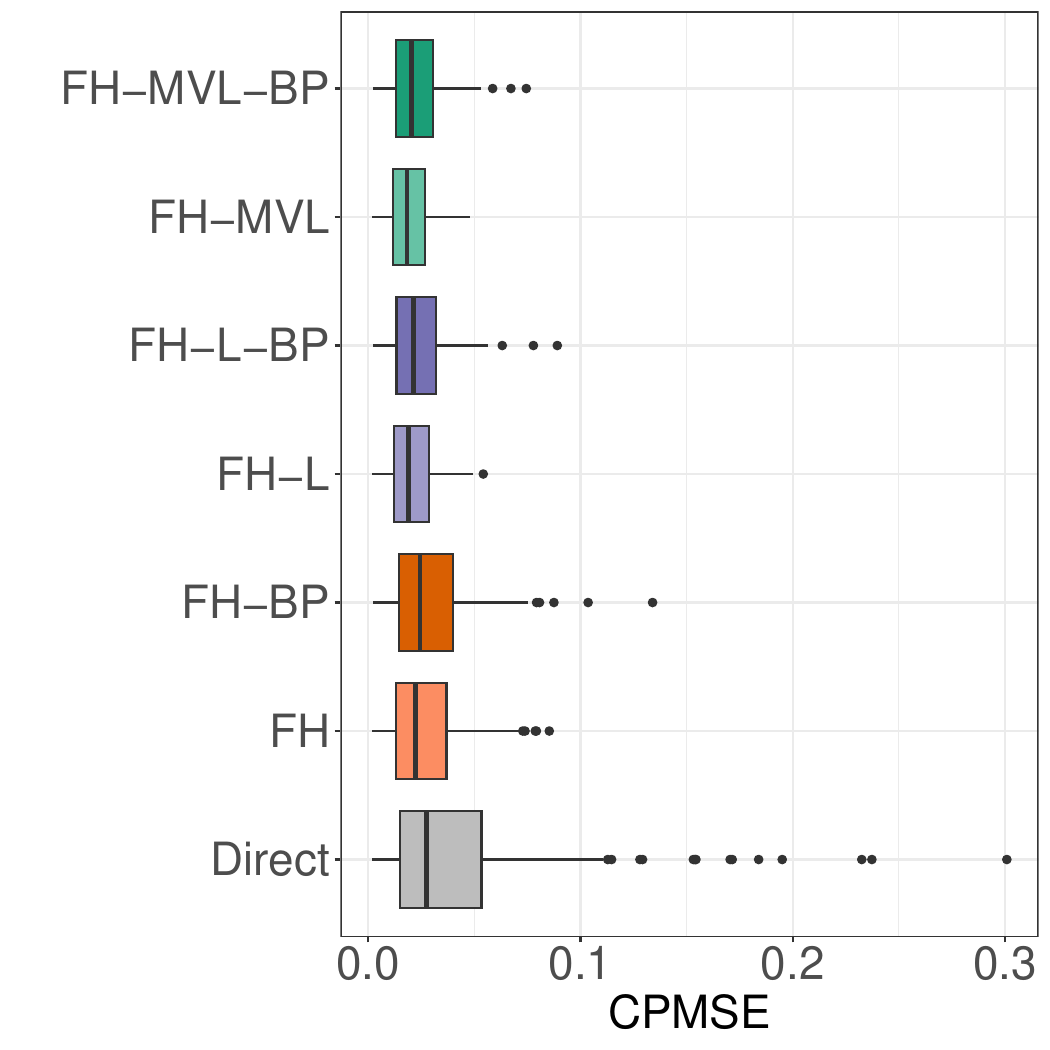} & \hspace{-0.5cm} \includegraphics[width=0.48\textwidth]{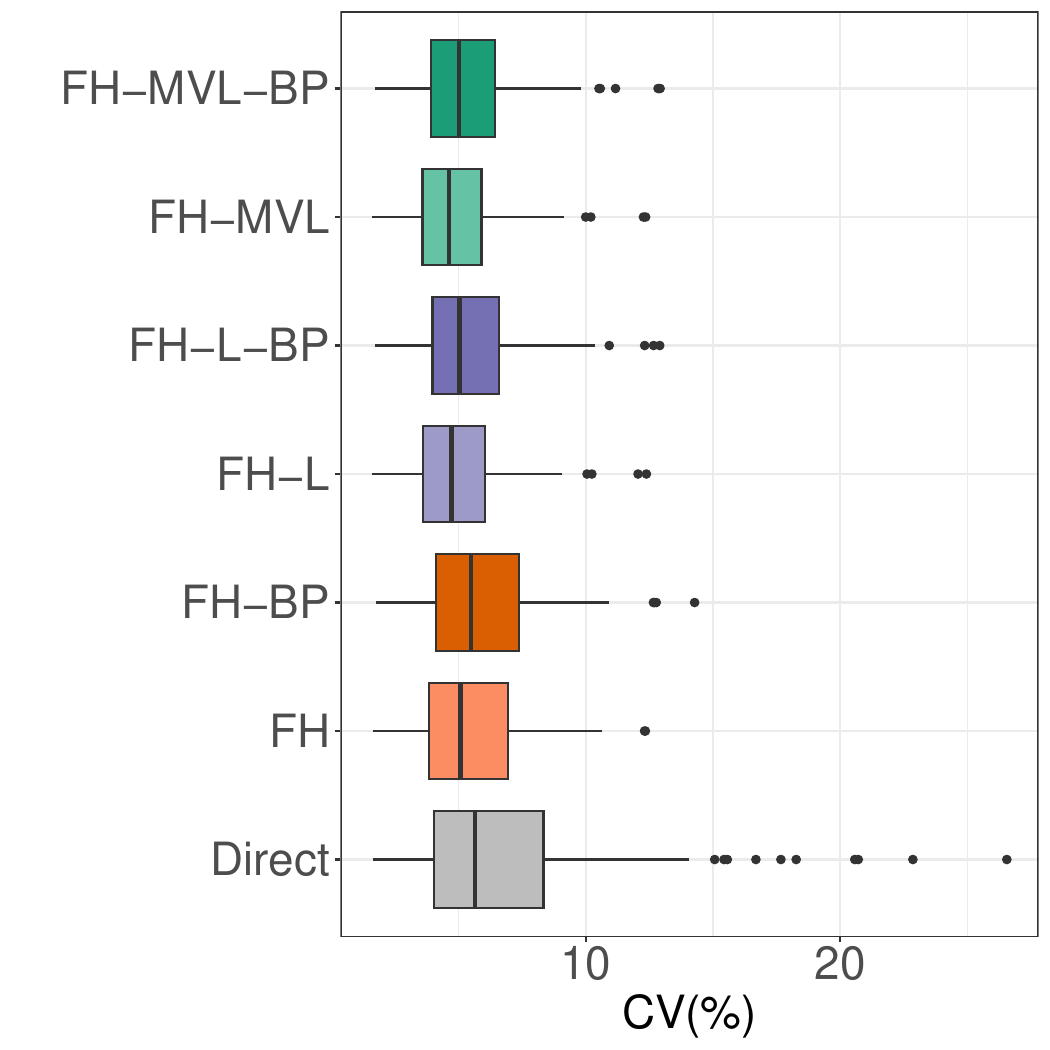} \\
\end{tabular}
\end{center}
\vspace{-0.3cm}
\caption{\small
Conditional Posterior Mean Square Error estimates (left) and
estimated Coefficients of Variation (CVs) (right) for models in Table \ref{tab:models}.} \label{fig:result1}
\end{figure}

\begin{figure}[htpb]
\centering
\begin{tabular}{cc}
AHS Estimates & \hspace{-0.5cm}  $\widehat{\text{CV}}$ (\%) \\[-0.8cm]
\hspace{-1cm} \includegraphics[width=0.6\textwidth]{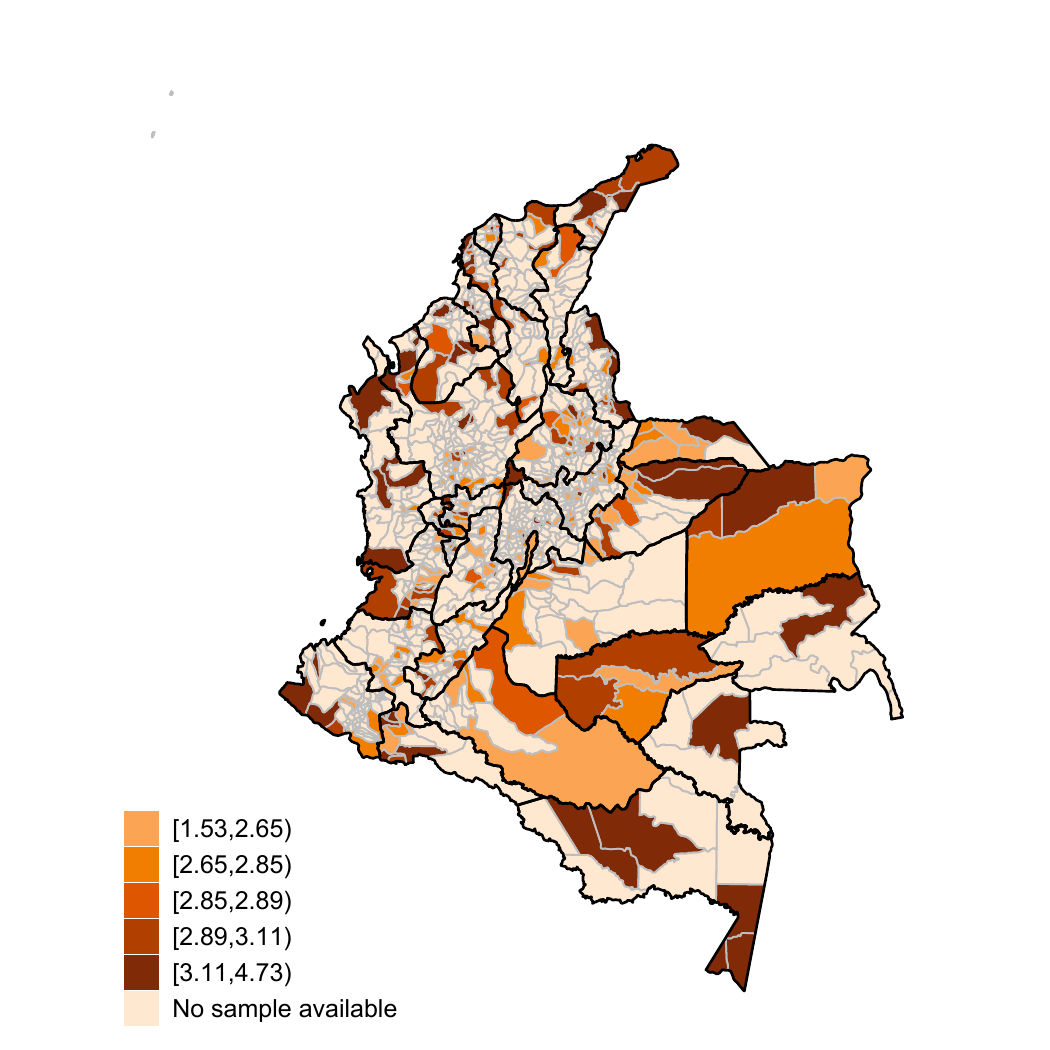} & \hspace{-1.5cm}
\includegraphics[width=0.6\textwidth]{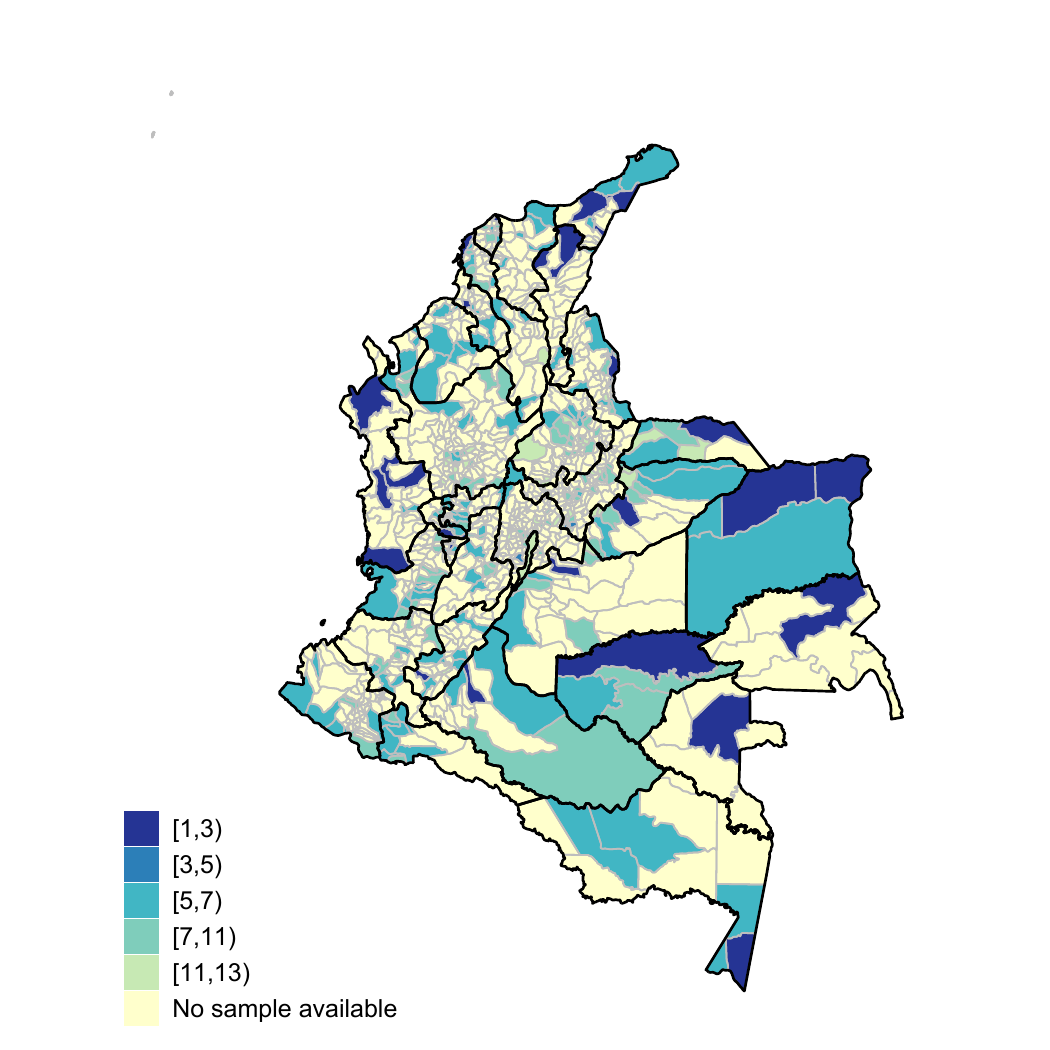}
\end{tabular}
\vspace{-0.3cm}
\caption{\small Municipality-level AHS estimates and corresponding coefficients of variation under the FH-MVL-BP model in Colombia.}
\label{fig:post2}
\end{figure}

Figure \ref{fig:post2} displays the posterior mean
AHS estimates under FH-MVL-BP and their corresponding coefficients of
variation. The results show substantial variation in household size across
municipalities, with posterior mean estimates ranging from approximately
1.5 to 4.7 persons per household. Table \ref{tab:EPD} shows that the graph regularized models consistently increase the number of municipality-level estimates with coefficients of variation below 7\%. Specifically, under FH-MVL-BP, approximately 83\% of municipality-level estimates have estimated
coefficients of variation below 7\%. The 7\% threshold is particularly relevant in this application
because estimates of key indicators with coefficients of variation below
7\% are considered ideal for publication in Colombian official statistics
\citep{DANE2008CV}. In addition to their municipality-level precision, the
FH-MVL-BP estimates satisfy the benchmarking constraint and are therefore
consistent with the national AHS estimate.

\section{Concluding remarks} \label{sec_Conclusions}

This work shows, for the first time, how graph-based regularization and
benchmarking restrictions can be jointly incorporated into the prior
distribution of a Bayesian SAE model. To enable posterior computation
under the resulting degenerate benchmarking priors, we developed a reduced
parameterization of the constraint space that avoids direct sampling from
a singular distribution and guarantees that the benchmarking restrictions
are satisfied at every posterior draw. We also introduced a Multi-View
Laplacian construction that allows multiple sources of external information
to inform the similarity structure among small areas. In the Colombian
application, the Multi-View models improved posterior predictive performance
and generally produced more precise municipality-level AHS estimates.

More broadly, the proposed construction provides a general way to
incorporate benchmarking into regularized Bayesian SAE beyond the particular
graph structures considered here. This opens several directions for future
work, including alternative similarity structures and Benchmarking priors,
more complex benchmarking restrictions, and extensions to other small area
models and applications.

\clearpage

\renewcommand{\thesection}{\Alph{section}}
\setcounter{section}{0}
\counterwithin{figure}{section}
\counterwithin{table}{section}
\setcounter{algorithm}{0}
\renewcommand{\thealgorithm}{A\arabic{algorithm}}


\appendix

\section*{Supplementary Material}

The following sections provide supplementary material for the paper. Section \ref{theory}  contains the proof of Lemma \ref{lemma1}. Section \ref{sec:CPMSE} provides additional details on the computation of the Conditional Posterior Mean Square Error (CPMSE) for the small area estimators considered in the paper. Section \ref{sec_Simulations} presents the data-based simulation study used to evaluate model performance and assess the CPMSE as a measure of uncertainty for the benchmarked small area estimators. Finally, Section \ref{MCMC_diagnostic} presents MCMC diagnostics for the models used in the application.

\section{Proof of Lemma \ref{lemma1}} \label{theory}

Under the FH model, the random variable  $\boldsymbol{\theta}$  follows a Normal distribution:
\begin{align}
\boldsymbol{\theta} \sim N_m(\mathbf{X}\boldsymbol{\beta},\sigma_u^2\mathbf{I}_m). 
\end{align}
Taking the expectation and variance of the random variable $\bdt^{\textsf{FH-B}}$ in Definition  \ref{def:FH_B} we obtain
\begin{align}\label{FH_B_mean}
\mathbb{E}(\bdt^{\textsf{FH-B}} \mid \mathbf{X}, \boldsymbol{\beta}, \sigma_u^2, \phi) &= 
\boldsymbol{A}_{\phi} \mathbb{E}(\bdt \mid \mathbf{X}, \boldsymbol{\beta}, \sigma_u^2, \phi) + 
\mathbb{E}(\boldsymbol{b}_{\phi} \mid \mathbf{X}, \boldsymbol{\beta}, \sigma_u^2, \phi) \\
&= \boldsymbol{A}_{\phi} \mathbf{X}\boldsymbol{\beta} + \boldsymbol{b}_{\phi}, \notag
\end{align} 
and
\begin{align}\label{FH_B_variance}
\operatorname{Var}(\bdt^{\textsf{FH-B}} \mid \mathbf{X}, \boldsymbol{\beta}, \sigma_u^2, \phi)&=
\operatorname{Var}(\boldsymbol{A}_{\phi} \bdt+ \boldsymbol{b}_{\phi} \mid  \mathbf{X}, \boldsymbol{\beta}, \sigma_u^2,  \phi)\\ &=
\sigma_u^2
\boldsymbol{A}_{\phi}\boldsymbol{A}_{\phi}^{\mathsf{T}}. \notag
\end{align} 

Note that  $ \boldsymbol{W}\boldsymbol{A}_{\phi} =\boldsymbol{W} \boldsymbol{Q}_{\phi}^{-1}+ \boldsymbol{W} \boldsymbol{Q}_{\phi}^{-1}
\boldsymbol{W}^T\left(
\boldsymbol{W}
\boldsymbol{Q}_{\phi}^{-1}
\boldsymbol{W}^\top
\right)^{-1}\boldsymbol{W} \boldsymbol{Q}_{\phi}^{-1} =\bm 0$. Therefore, the first property follows by noting that
\begin{align*}
\boldsymbol{W}(\boldsymbol{A}_{\phi}\mathbf{X}\boldsymbol{\beta}
+\boldsymbol{b}_{\phi})
&=    \boldsymbol{W} \boldsymbol{Q}_{\phi}^{-1}\boldsymbol{W}^T\left(
\boldsymbol{W}
\boldsymbol{Q}_{\phi}^{-1}
\boldsymbol{W}^\top
\right)^{-1}\boldsymbol{p}\\
&=\boldsymbol{p}.
\end{align*} 
For the second property, consider  $\boldsymbol{A}_{\phi,\sigma_u^2 } =\sigma_u^2 \boldsymbol{A}_{\phi}\boldsymbol{A}_{\phi}^{\mathsf{T}}$
and
\[
\bm W\in\mathbb{R}^{k\times m},\qquad \operatorname{rank}(\bm W)=k,
\qquad \boldsymbol{A}_{\phi,\sigma_u^2 }  \in\mathbb{R}^{m\times m},
\qquad 
\boldsymbol{W} \boldsymbol{A}_{\phi,\sigma_u^2 } =0.
\]
Since
\[
\operatorname{Col}(\bm W^{\mathsf{T}})\subseteq\ker(\boldsymbol{A}_{\phi,\sigma_u^2 }^{\mathsf{T}} ),
\]
we have
\[
\dim\ker(\boldsymbol{A}_{\phi,\sigma_u^2 }^{\mathsf{T}} )\geq \operatorname{rank}( \boldsymbol{W})=k.
\]
By the rank--nullity theorem,
\[
\operatorname{rank}(\boldsymbol{A}_{\phi,\sigma_u^2 }^{\mathsf{T}} )
+\dim\ker(\boldsymbol{A}_{\phi,\sigma_u^2 }^{\mathsf{T}} )=m,
\]
and therefore
\[
\operatorname{rank}(\boldsymbol{A}_{\phi,\sigma_u^2 }^{\mathsf{T}}) \leq m-k<m.
\]
Hence, for \(k\geq 1\), $\boldsymbol{A}_{\phi,\sigma_u^2 }$  is singular
and $\bdt^{\textsf{FH-B}}$ follows a degenerate Normal distribution with mean and variance given by
(\ref{FH_B_mean}) and (\ref{FH_B_variance}), respectively.

\section{Conditional Posterior Mean Square Error for Benchmarked Estimators}
\label{sec:CPMSE}

The CPMSE proposed by \citet{fuquenefully} applies to the FH, FH-L, and FH-MVL models and their benchmarked counterparts. Let $\M$ denote the corresponding unbenchmarked model, that is, FH, FH-L or FH-MVL. The CPMSE for the benchmarked Rao-Blackwell (RB) estimator is given in Definition \ref{def3}. 
 
 \begin{definition}[\cite{fuquenefully}]
\label{def3}
The Conditional Posterior Mean Square Error (CPMSE) for the benchmarked Rao-Blackwell  estimator $\hat{\theta}_i^{\M\textsf{-B-RB}}$  of the  posterior mean of $\theta_{i}^{\M}$ 
under model $\M$ is given by
\begin{align}
   \label{CPMSE_projection_estimator}
   \emph{CPMSE}(\hat{\theta}_i^{\M\textsf{-B-RB}} \mid \boldsymbol{y}, \boldsymbol{X}
      ) &=  E_{(\phi, \bm\beta, \sigma_u^2)}(E_{\bdt^{\M}}((\hat{\theta}_i^{\M\textsf{-B-RB}} - \theta_i^{\M})^{2}\mid \boldsymbol{y}, \boldsymbol{X}, \phi, \bm\beta, \sigma_u^2))       \\ \notag
        &   =   \emph{CPMSE}(\hat{\theta}_i^{\M\textsf{-RB}} \mid \boldsymbol{y}, \boldsymbol{X}) + (\hat{\theta}_i^{\M\textsf{-B-RB}} - \hat{\theta}_i^{\M\textsf{-RB}})^{2}, 
        \end{align} 
where  $\hat{\theta}_i^{\M\textsf{-RB}}$ is the unbenchmarked Rao-Blackwell estimator of the posterior mean of $\theta_{i}^{\M}$.  
 \end{definition}
To implement the CPMSE, we consider RB estimators, $\hat{\bdt}^{\M\textsf{-B-RB}}$,  of the posterior mean of $\bdt^{\M\textsf{-B}}=(\theta_{i}^{\M\textsf{-B}},...,\theta_{m}^{\M\textsf{-B}})^{\top}$  obtained with the output of Algorithm \ref{mcmc:alg1}. Specifically, consider a Rao-Blackwellization of $E_{\bdt^{\M\textsf{-B}}}(\bdt^{\M\textsf{-B}} \mid \bm y, \bm X)$,
 \begin{align}  \label{RB}
\bdt^{\M\textsf{-B-RB}}&= E_{(\phi, \bm\beta, \sigma_u^2)}(E_{\bdt^{\M\textsf{-B}}}(\bdt^{\M\textsf{-B}} \mid \bm y, \bm X, \phi, \bm\beta, \sigma_u^2))  \\ \notag
&= \bm A_{\phi}\bm X\bm\beta
+
\bm b_{\phi}
+
\bm B \bm\mu_{\boldsymbol{\omega}},
 \end{align} 
 where $\bm \mu_{\boldsymbol{\omega}}$ is the conditional posterior mean of $ \bm \omega$ provided in  Algorithm \ref{mcmc:alg1}. 
A Monte Carlo estimator of $\bdt^{\M\textsf{-B-RB}}$ using  the output
of Algorithm \ref{mcmc:alg1} can be computed as follows:
 \begin{align}  \label{RB_est}
\hat{\bdt}^{\M\textsf{-B-RB}} = \dfrac{1}{S}\sum_{s=1}^{S}  (\bm A_{\phi^{(s)}}\bm X\bm\beta^{(s)}
+
\bm b_{\phi^{(s)}}
+
\bm B \bm \mu_{\boldsymbol{\omega}^{(s)}}),
\end{align} 
where $\hat{\bdt}^{\M\textsf{-B-RB}}=(\hat{\theta}_i^{\M\textsf{-B-RB}},...,\hat{\theta}_m^{\M\textsf{-B-RB}})^{\top}$. 

Note that in Definition \ref{def3},  $\hat{\theta}_i^{\M\textsf{-RB}} $ denotes the unbenchmarked RB estimator under the corresponding FH, FH-L and FH-MVL model. More specifically, 
the first part of the CPMSE in   (\ref{CPMSE_projection_estimator}) is the uncertainty associated with the unbenchmarked  small area parameter 
measured with the CPMSE of the unbenchmarked  RB estimator, $\hat{\theta}_i^{\M\textsf{-RB}}$, of the posterior mean of $\theta_i^{\M}$ and computed by
\begin{align}
   \label{CPMSE_projection_estimator2}
   \text{CPMSE}(\hat{\theta}_i^{\M\textsf{-RB}} \mid \boldsymbol{y}, \boldsymbol{X}
      ) &=   
            E_{(\phi, \bm\beta, \sigma_u^2)}(V_{\bdt^{\M}}( \theta_i^{\M} \mid \boldsymbol{y}, \boldsymbol{X}, \phi, \bm\beta, \sigma_u^2))    \\ \notag
        &  \hspace{-0.5cm} + E_{(\phi, \bm\beta, \sigma_u^2)}((E_{\bdt^{\M}}((\hat{\theta}_i^{\M\textsf{-RB}}- \theta_i^{\M})\mid \boldsymbol{y}, \boldsymbol{X}, \phi, \bm\beta, \sigma_u^2))^{2}).
        \end{align}

The second term in the CPMSE in Definition \ref{def3}, $(\hat{\theta}_i^{\M\textsf{-B-RB}} - \hat{\theta}_i^{\M\textsf{-RB}} )^{2}$,
measures the additional discrepancy induced by enforcing the benchmarking
constraint. Using the posterior samples obtained from Algorithm
\ref{mcmc:alg1}, the CPMSE used in the simulation and application studies can
be estimated by Monte Carlo approximation as
\begin{align}
   \label{CPMSE_projection_estimator3}
   \widehat{\text{CPMSE}}(\hat{\theta}_i^{\M\textsf{-B-RB}} \mid \boldsymbol{y}, \boldsymbol{X}
      ) &=  (\hat{\theta}_i^{\M\textsf{-B-RB}} - \hat{\theta}_i^{\M\textsf{-RB}} )^{2} + \dfrac{1}{S}\sum_{s=1}^{S} V_{\bdt^{\M}}( (\theta_i^{\M})^{(s)} \mid \boldsymbol{y}, \boldsymbol{X}, \phi, \bm\beta, \sigma_u^2))   \notag  \\  
        &  \hspace{-0.5cm} + \dfrac{1}{S}\sum_{s=1}^{S}  (\hat{\theta}_i^{\M\textsf{-RB}}- E_{\bdt^{\M}}((\theta_i^{\M})^{(s)}\mid \boldsymbol{y}, \boldsymbol{X}, \phi, \bm\beta, \sigma_u^2))^{2}.
        \end{align} 
The estimated coefficient of variation for the benchmarked RB estimator
under model $\M$ is then computed using
(\ref{RB_est}) and (\ref{CPMSE_projection_estimator3}) as
\begin{align}
   \widehat{\text{CV}}_i=\dfrac{\sqrt{\widehat{\text{CPMSE}}(\hat{\theta}_i^{\M\textsf{-B-RB}} \mid \boldsymbol{y}, \boldsymbol{X}
      )}}{\hat{\theta}_i^{\M\textsf{-B-RB}}}.  
        \end{align}       
Similarly, the coefficient of variation can be obtained for the unbenchmarked RB estimator.

\section{Data-based simulation study} \label{sec_Simulations}

To illustrate the performance of the proposed models in Table \ref{tab:models} and the performance of the CPMSE, we implement a data-based simulation study. To this end, we select $m=50$ small areas from our motivating case study in Section \ref{sec_Application}. We use the FH-MVL  model as the data-generating model with the prior settings described in Section \ref{MCMC_algorithms}. The generating model under the FH-MVL model is given by 
\begin{align}
\label{model_FHL_m}
    \by &\sim \text{Normal}((\boldsymbol{Q}_{\phi}^{\textsf{L}_{C^{*}}^{\textsf{MV}}})^{-1}\bdt, \boldsymbol{D}),\\
    \bdt &\sim \text{Normal}(\mathbf{X}\boldsymbol{\beta},  \sigma_u^2\mathbf{I}_m).
\end{align} 

We use the same covariates as in the application, an intercept and the HPP  \citep{CNA2014}. We consider the matrix  $\boldsymbol{Q}_{\phi}^{\textsf{L}_{C^{*}}^{\textsf{MV}}}$ as in (\ref{LapRLS_2}).
To set the regression parameters we consider similar values
to those obtained under the FH model in the real application. Specifically, we set $\beta_{0}=4$ , $\beta_{1}=-4$ and  $\sigma^2_{u}=0.25$. As in the application, we  also assume the benchmarking constraint as $\sum_{i=1}^{m} w_{i} \theta_{i} = 2.82$.  We use $\phi=0.3$ similar to that obtained in the application. Since the FH-MVL model  is the underlying model we calculate a similarity graph using Algorithm \ref{alg:SC2}  (Figure \ref{fig:CPMSE1} (left)) with the 50 observations of AHS, R-MPI and U-MPI  selected randomly. To obtain posterior samples under the various models in Table \ref{tab:models}, we use two MCMC chains for 20,000 iterations, discard the first 20,000 iterations as burn-in, and retain every fourth draw thereafter. We also use the hyperparameters for the priors of $\lambda_j$, $\phi$ and  $\sigma_u^2$ utilized in Section \ref{sec_Application}.

We simulate $T=100$ data sets and evaluate the various models using the difference of the posterior mean and the simulated $ \theta_{i}$. Specifically, we compute the average absolute deviation (AAD) and average squared deviation (ASD)  of  $\theta_{i}$ and $\theta_{i}^{(t)}$ across data sets as follows,
\begin{align}\label{deviance}
\text{AAD}_{\bm \theta}^{(t)} &= \dfrac{1}{m} \sum_{i=1}^{m} \mid (\hat{\theta}_{i}^{\M})^{(t)}
 - \theta_{i}^{(t)} \mid, 
& \text{ASD}_{\bm \theta}^{(t)} &=\dfrac{1}{m} \sum_{i=1}^{m}  ( (\hat{\theta}_{i}^{\M})^{(t)}
 - \theta_{i}^{(t)} )^{2},
\end{align}

To evaluate the performance of the CPMSE for the benchmarked RB estimators   
we compute the average of the CPMSE  of $\hat{\theta}_i^{\mathcal{M}\textsf{-B-RB}}$ 
 across the simulated data sets for each small area, as follows
 \begin{align}
   \label{eq:CPMSE}
 \widehat{\text{CPMSE}}(\hat{\theta}_i^{\mathcal{M}\textsf{-B-RB}}) &= \dfrac{1}{T}\sum_{t=1}^{T}(\text{CPMSE}( \hat{\theta}_i^{\mathcal{M}\textsf{-B-RB}}))^{(t)}, 
\end{align}
and the MSE of $\hat{\theta}_i^{\mathcal{M}\textsf{-B-RB}}$ given by 
 \begin{align}
   \label{eq:MSE}
 \widehat{\text{MSE}}(\hat{\theta}_i^{\mathcal{M}\textsf{-B-RB}}) &= \dfrac{1}{T}\sum_{t=1}^{T}((\hat{\theta}_i^{\mathcal{M}\textsf{-B-RB}})^{(t)}-  \theta_{i}^{(t)})^{2},
\end{align}
In addition, to measure the quality of the approximation of the MSE provided by CPMSE, we consider the average of  (\ref{eq:CPMSE}) and (\ref{eq:MSE}), given by  
 \begin{align}
   \label{eq:CPMSE_MSE}
\widehat{\text{CPMSE}}^{\textsf{Avg}-\mathcal{M}} &= \dfrac{1}{m}\sum_{i=1}^{m}  \widehat{\text{CPMSE}}(\hat{\theta}_i^{\mathcal{M}\textsf{-B-RB}}), \\
\widehat{\text{MSE}}^{\textsf{Avg}-\mathcal{M}} &= \dfrac{1}{m}\sum_{i=1}^{m} \widehat{\text{MSE}}(\hat{\theta}_i^{\mathcal{M}\textsf{-B-RB}}).
\end{align} 
Similarly, we compute (\ref{eq:CPMSE})-(\ref{eq:CPMSE_MSE}) under the unbenchmarked models.

Multi-View models  produce more precise small area estimates. Specifically, in Figure \ref{fig:CPMSE1} (right) we illustrate that the RB estimators under FH-MVL and
FH-MVL-BP, respectively, produce the smallest values of CPMSE.
 Figure \ref{fig:CPMSE2} illustrates the difference between $\widehat{ \text{CPMSE}}$ and $ \widehat{\text{MSE}}$. According to the plots,  $\widehat{ \text{CPMSE}}$ provides a good approximation of the MSE under the various models.  Table \ref{tab:CPMSEs} 
shows the absolute discrepancies between $\widehat{\text{CPMSE}}^{\textsf{Avg}-\mathcal{M}}$ and $\widehat{\text{MSE}}^{\textsf{Avg}-\mathcal{M}}$. The results in Table \ref{tab:CPMSEs} support the findings in Figure \ref{fig:CPMSE2}, showing that CPMSE produces suitable estimates of the MSE.

\begin{table}[htbp]
    \centering
    \caption{Absolute difference between $\widehat{\text{CPMSE}}^{\textsf{Avg}-\mathcal{M}}$ and $\widehat{\text{MSE}}^{\textsf{Avg}-\mathcal{M}}$ with the various models.}
    \renewcommand{\arraystretch}{1.15}
    \setlength{\tabcolsep}{17pt}
    \begin{tabular}{lc}
        \toprule
  SAE Model & $ \mid \widehat{\text{CPMSE}}^{\textsf{Avg}-\mathcal{M}} -\widehat{\text{MSE}}^{\textsf{Avg}-\mathcal{M}} \mid$ \\
        \midrule
FH  & 0.00290 \\ 
         FH-L & 0.00229 \\ 
         FH-MVL  &  0.00009 \\ 
      FH-BP & 0.00195 \\ 
      FH-L-BP   & 0.00229 \\ 
      FH-MVL-BP  & 0.00008 \\ 
        \bottomrule
    \end{tabular}
\label{tab:CPMSEs}
\end{table}

Figure \ref{fig:fig1_Dev} shows that, among the unbenchmarked and benchmarked
models, the proposed FH-MVL and FH-MVL-BP models attain the
smallest values of $\text{AAD}_{\bm \theta}^{(t)}$ and $\text{ASD}_{\bm \theta}^{(t)}$, respectively. The figure also identifies the model selected by DIC within each group. Importantly, DIC selects FH-MVL among the unbenchmarked
models and FH-MVL-BP among the benchmarked models, which are also the models
with the smallest deviance values.

\begin{figure}[h!]
\centering
 \hspace{-1.2cm}
\begin{minipage}[t]{0.55\textwidth}
  \hspace{2.0cm} Multi-View, $C^{*}=2$ \vspace{-1.0cm}
\\
  \centering
  \includegraphics[width=\textwidth]{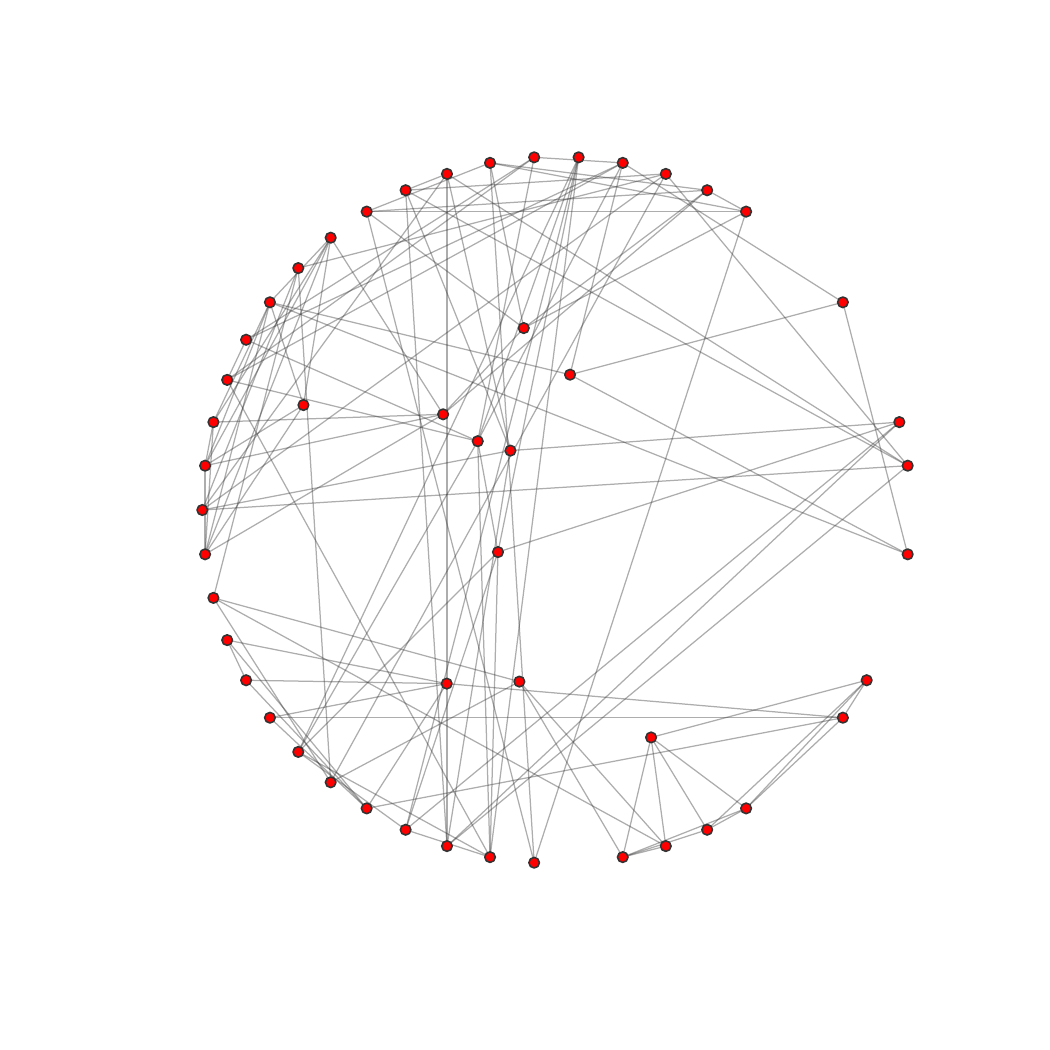}
\end{minipage}
  \hspace{-2cm}
\begin{minipage}[t]{0.60\textwidth}
  \hspace{2.0cm} Multi-View, $C^{*}=2$ \vspace{-0.4cm}
\\
  \centering
 \includegraphics[width=0.8\textwidth]{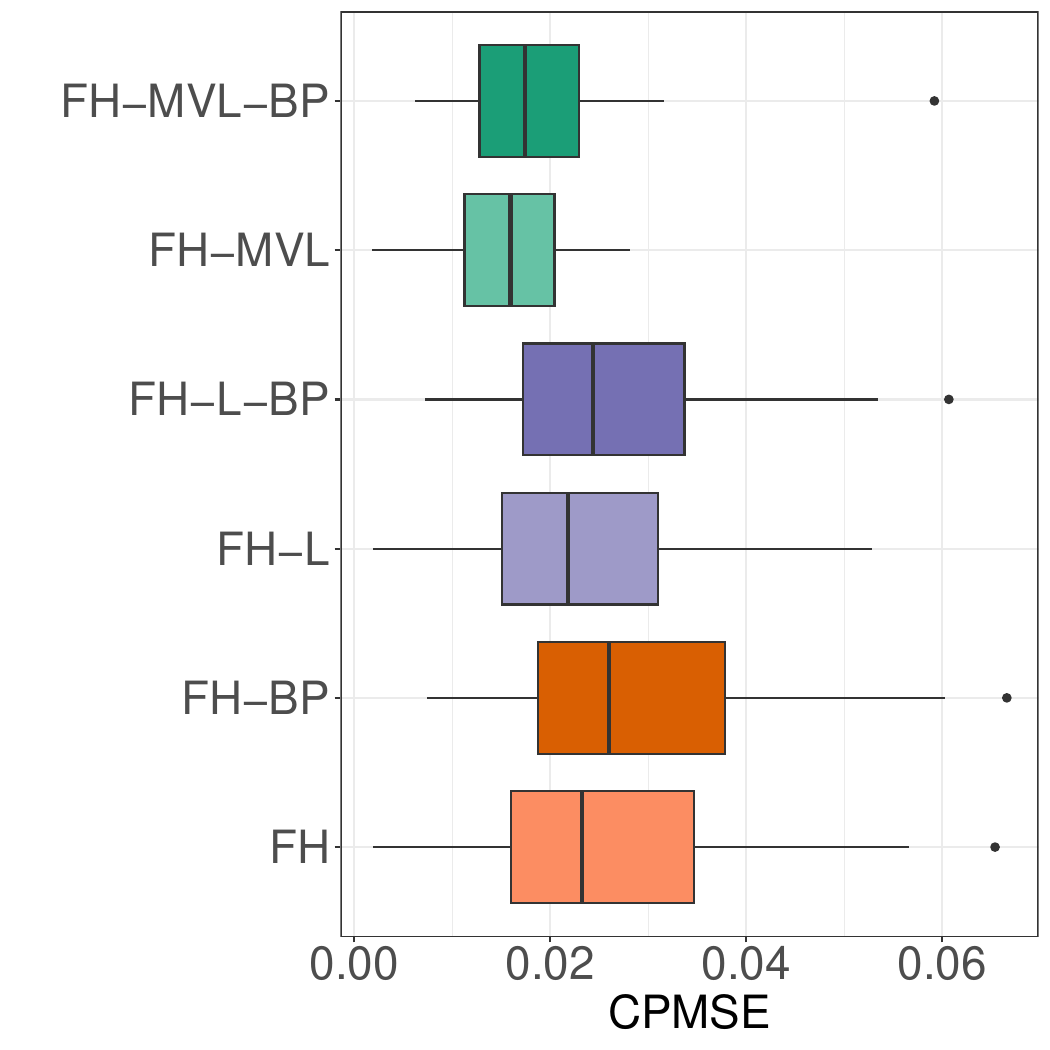}
\end{minipage}
\vspace{-0.7cm}
\caption{\small
Multi-View graph for the simulation study obtained from Algorithm \ref{alg:SC2} with 50 observations selected randomly  from AHS, R-MPI and U-MPI (left) 
and   small area estimates of CPMSE obtained 
with FH, FH-L, FH-MVL, FH-BP, FH-L-BP and FH-MVL-BP (right)\label{fig:CPMSE1}.}
\vspace{0.5cm}
\begin{center}
\begin{tabular}{ccc}
\hspace{-1cm}
\includegraphics[width=0.48\textwidth]{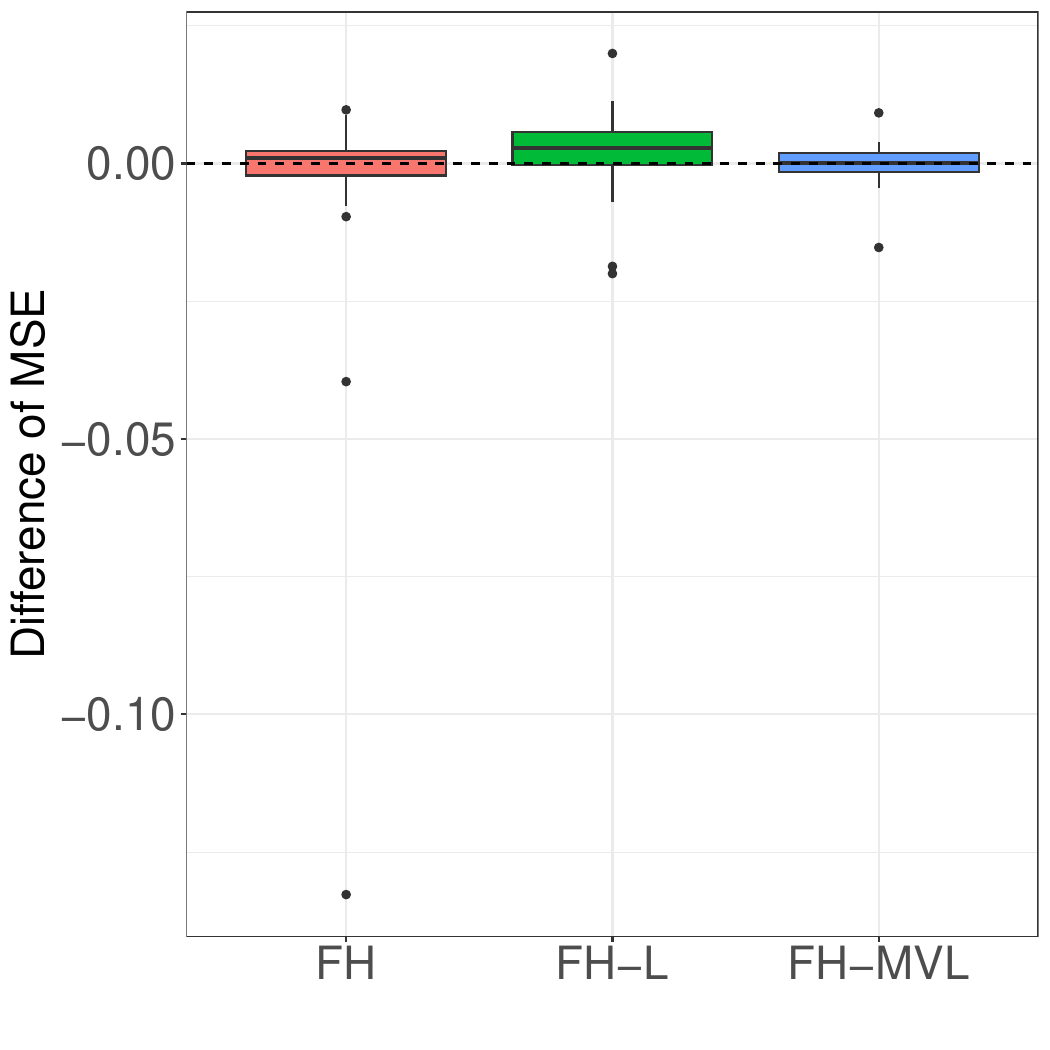}  & \includegraphics[width=0.48\textwidth]{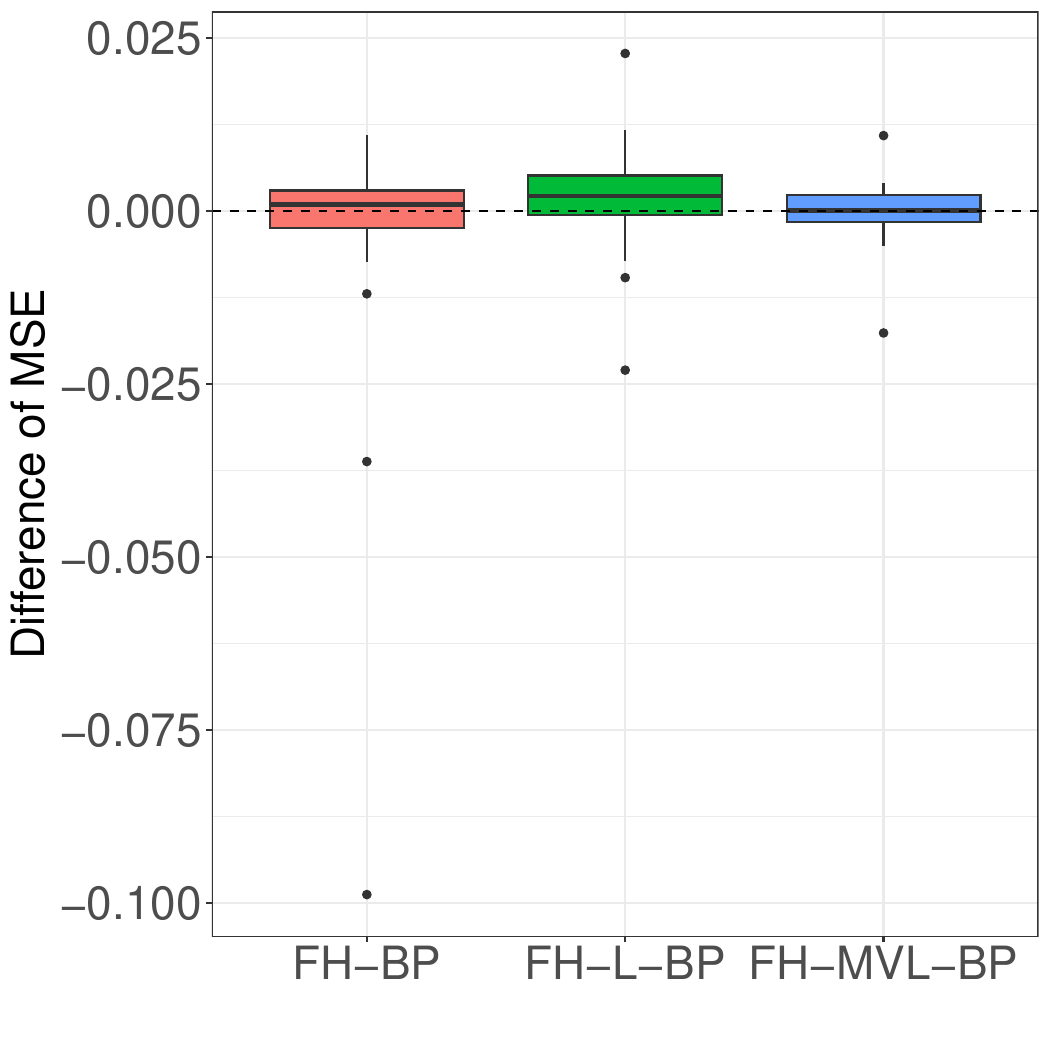}
\end{tabular}
\end{center}
\caption{Difference between $\widehat{\text{CPMSE}}(\hat{\theta}_{i}^{\M\textsf{-RB}})$ and $\widehat{\text{MSE}}(\hat{\theta}_{i}^{\M\textsf{-RB}})$ under
FH, FH-L and FH-MVL (left)  and difference between  $\widehat{\text{CPMSE}}(\hat{\theta}_{i}^{\M\textsf{-B-RB}})$ and $\widehat{\text{MSE}}(\hat{\theta}_{i}^{\M\textsf{-B-RB}})$ under
FH-BP, FH-L-BP and FH-MVL-BP (right). \label{fig:CPMSE2}}
\end{figure}

\clearpage

\begin{figure}[ht]
\begin{center}
\begin{tabular}{ccc}
\hspace{-0.5cm}\includegraphics[width=0.5\textwidth]{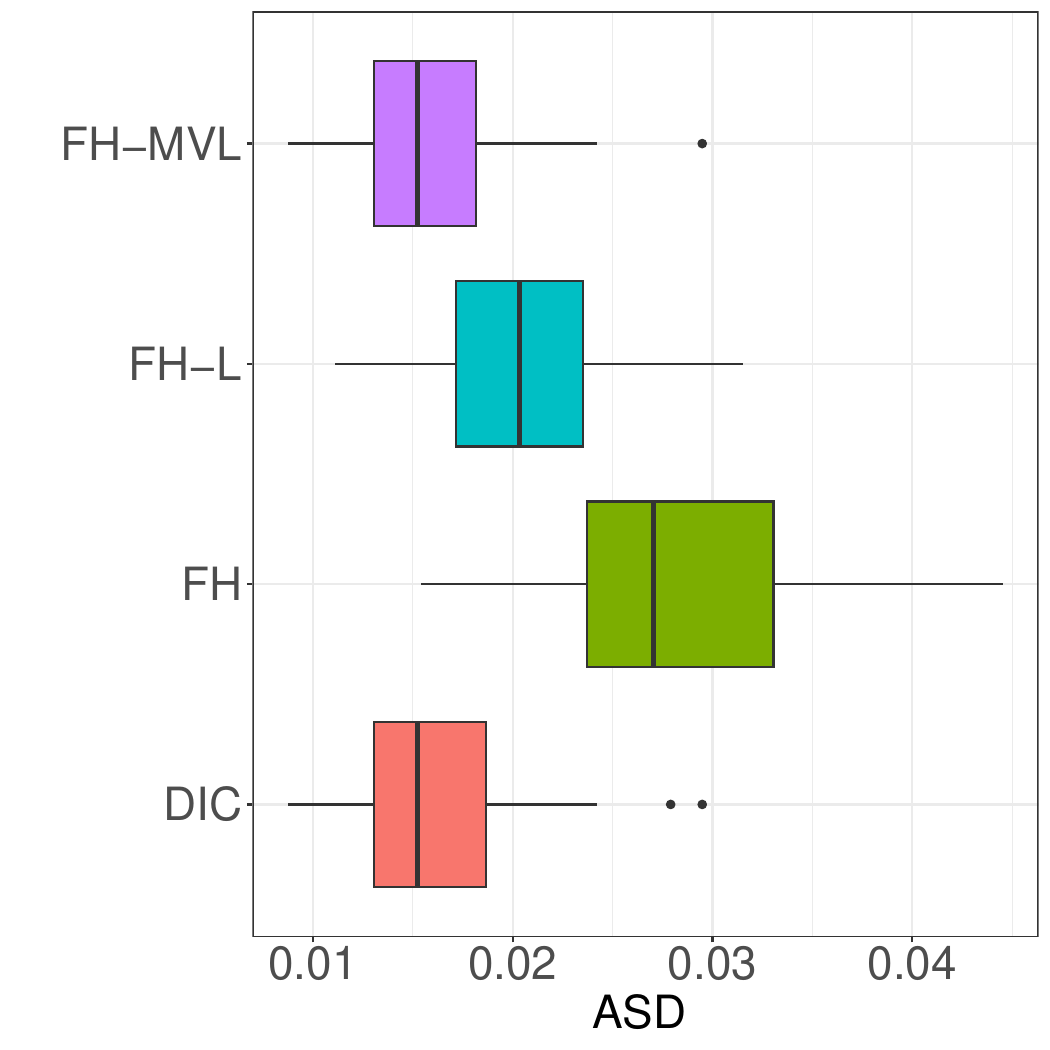}  & \hspace{-0.5cm} \includegraphics[width=0.5\textwidth]{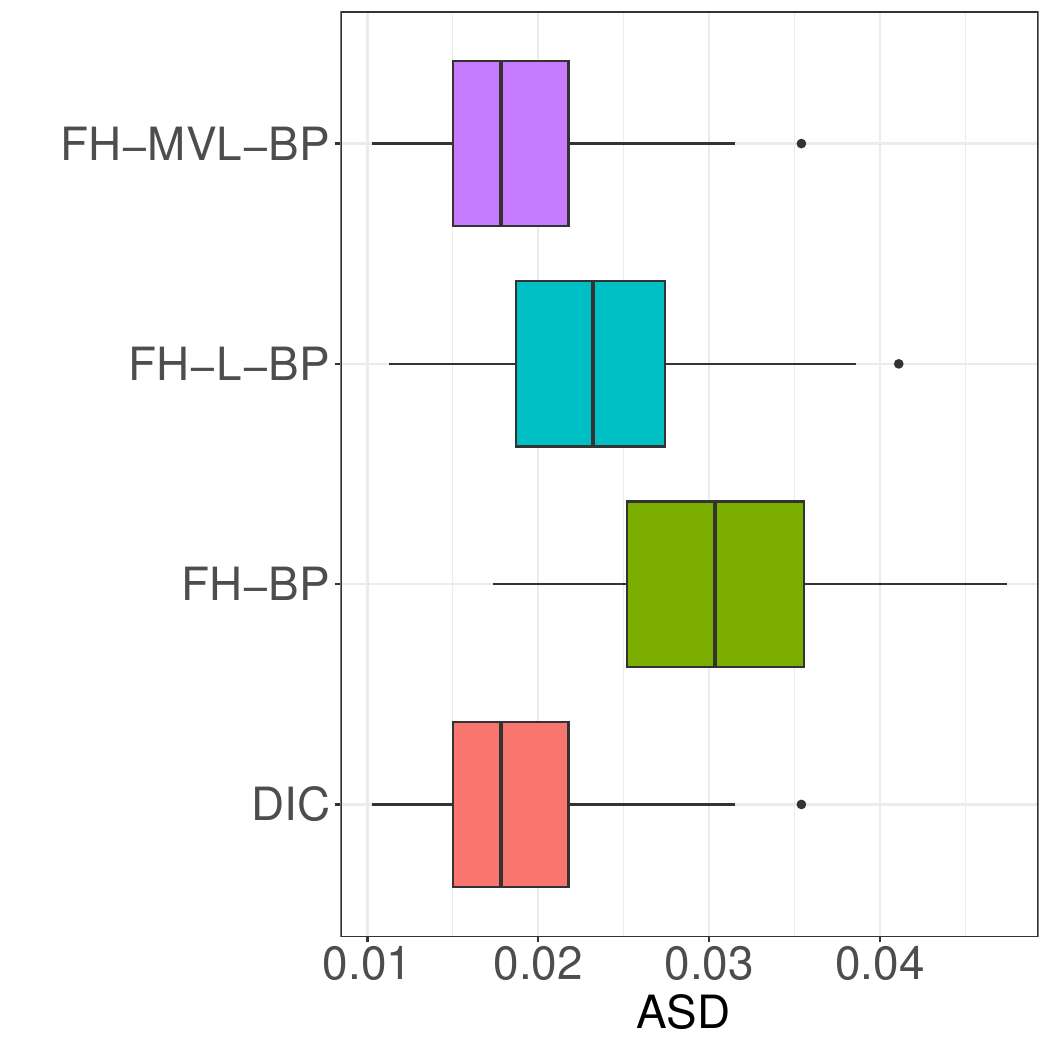} \\
\hspace{-0.5cm}\includegraphics[width=0.5\textwidth]{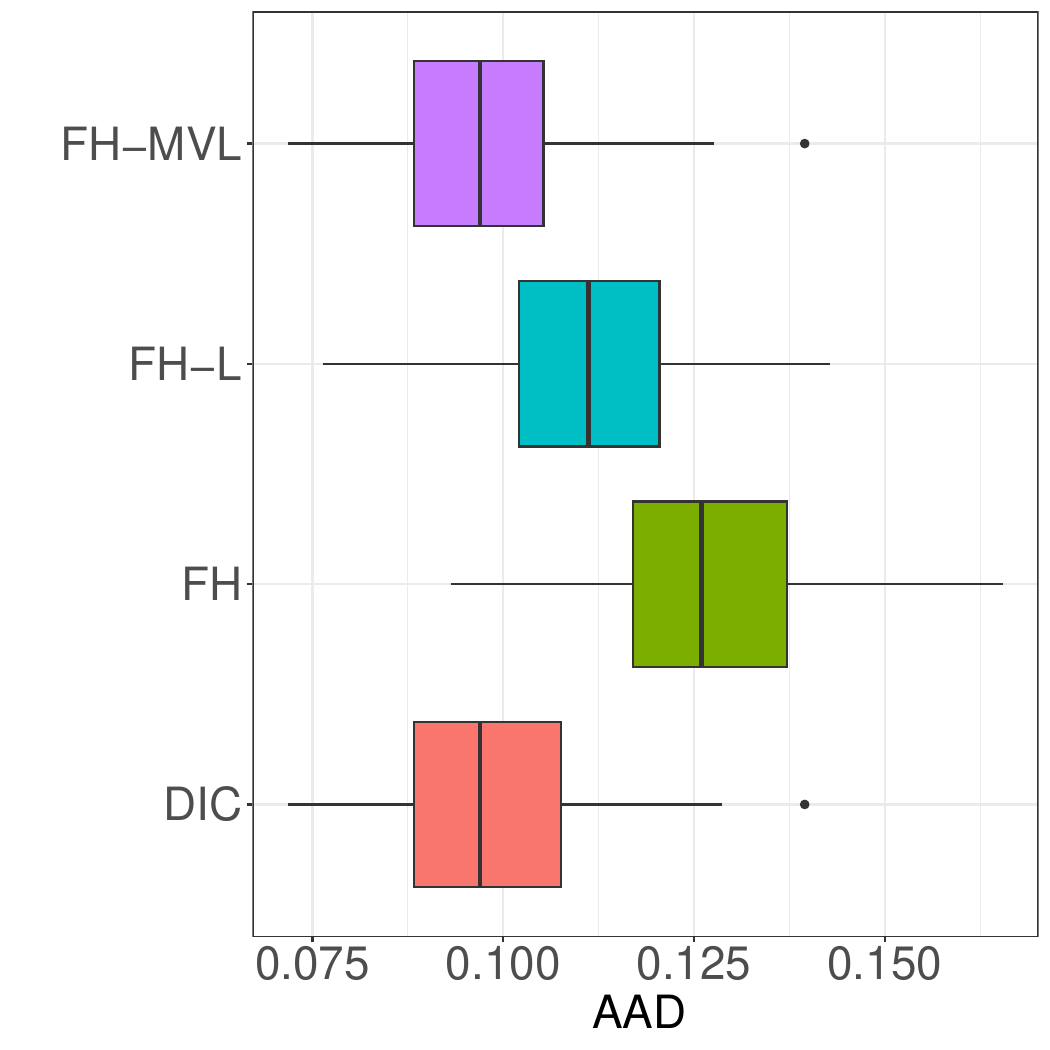}  & \hspace{-0.5cm} \includegraphics[width=0.5\textwidth]{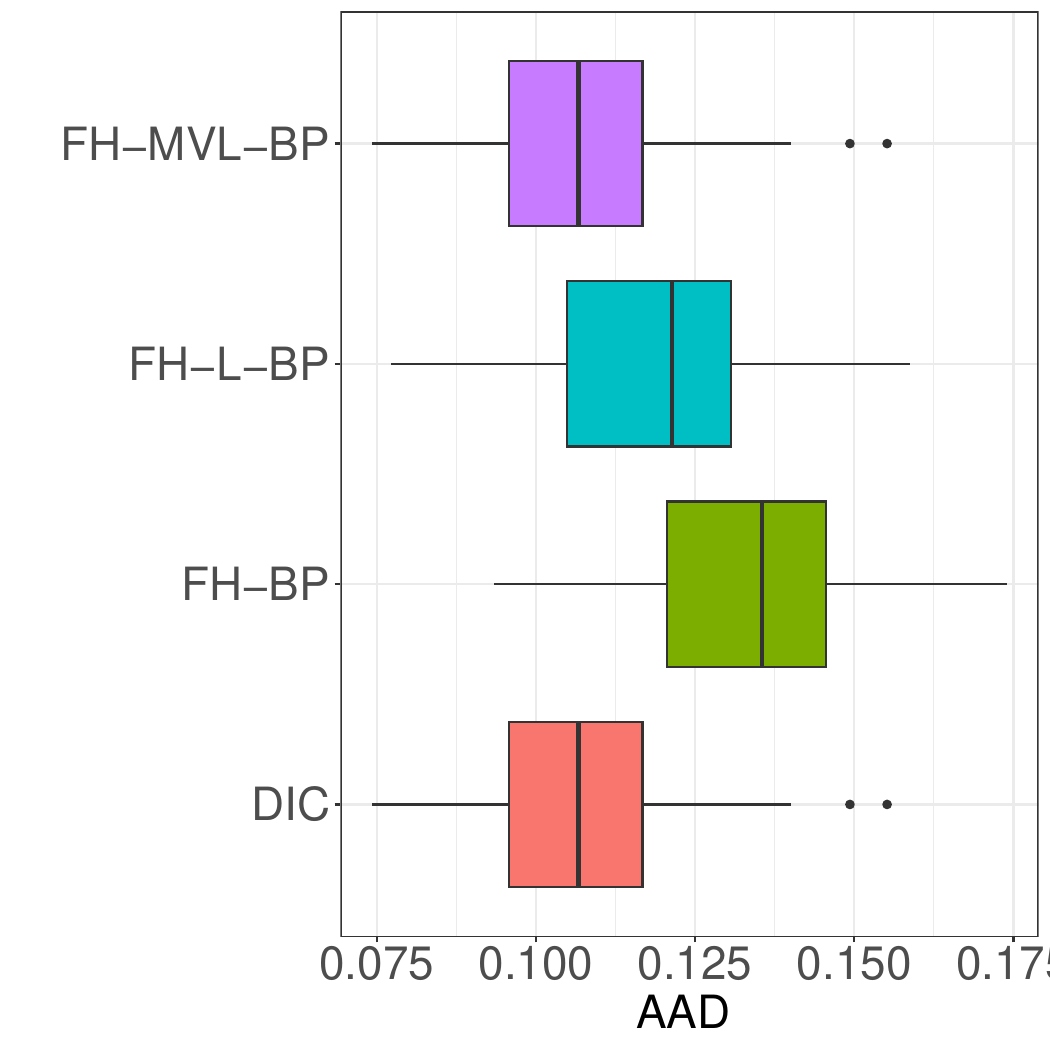} \\
\end{tabular}
\end{center}
\caption{Plots of $\text{AAD}_{\bm \theta}^{(t)} $ and $\text{ASD}_{\bm \theta}^{(t)} $ under
FH, FH-L and FH-MVL (left)  and under FH-BP, FH-L-BP and FH-MVL-BP (right). The plots also present the deviation measures for the model selected by DIC.   \label{fig:fig1_Dev}}
\end{figure}

\clearpage

\section{MCMC diagnostics for the model in the Application}
\label{MCMC_diagnostic}

In this section, we illustrate the MCMC diagnostics for the FH-MVL-BP model implemented in Section \ref{sec_Application}. We examine trace plots, density plots, and ergodic mean plots based on 40,000 MCMC iterations obtained after a burn-in period of 40,000 iterations. Every fourth iteration is retained, resulting in 10,000 posterior draws for inference. The plots in Figure \ref{fig:diagnostic_FH2} display the trace plots, density plots, and ergodic mean diagnostics. Overall, the results indicate satisfactory mixing and convergence, with no apparent issues in the MCMC sampling for  the Multi-View model with benchmarking, FH-MVL-BP. Similar results were observed for the other models.

\begin{figure}[htpb]
\small
\begin{center}
\begin{tabular}{cccc}
\hspace{-2.0cm}  $\beta_0$  & $\beta_1$  & $\sigma_u$  & $\phi$ \\
\hspace{-2.0cm} \includegraphics[width=0.3\textwidth]{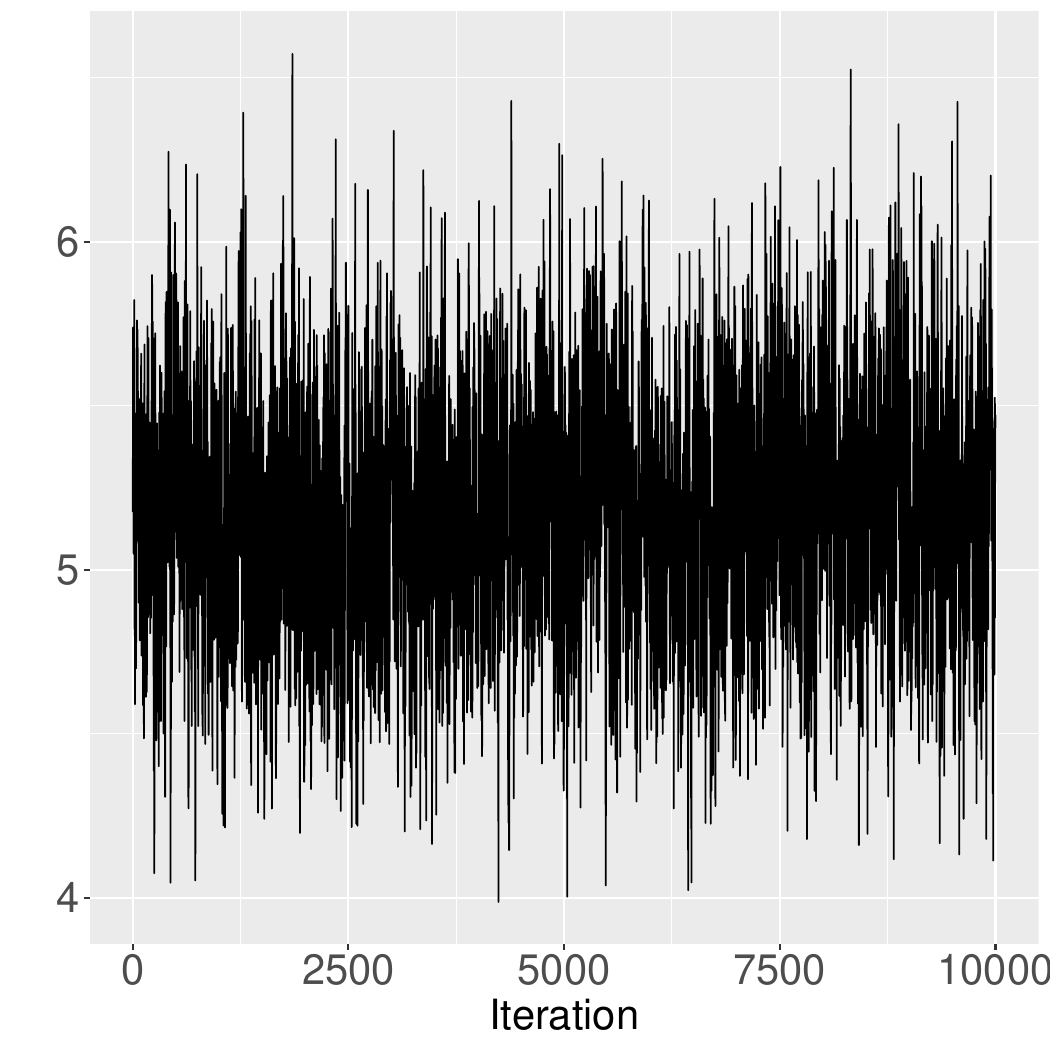} 
& \includegraphics[width=0.3\textwidth]{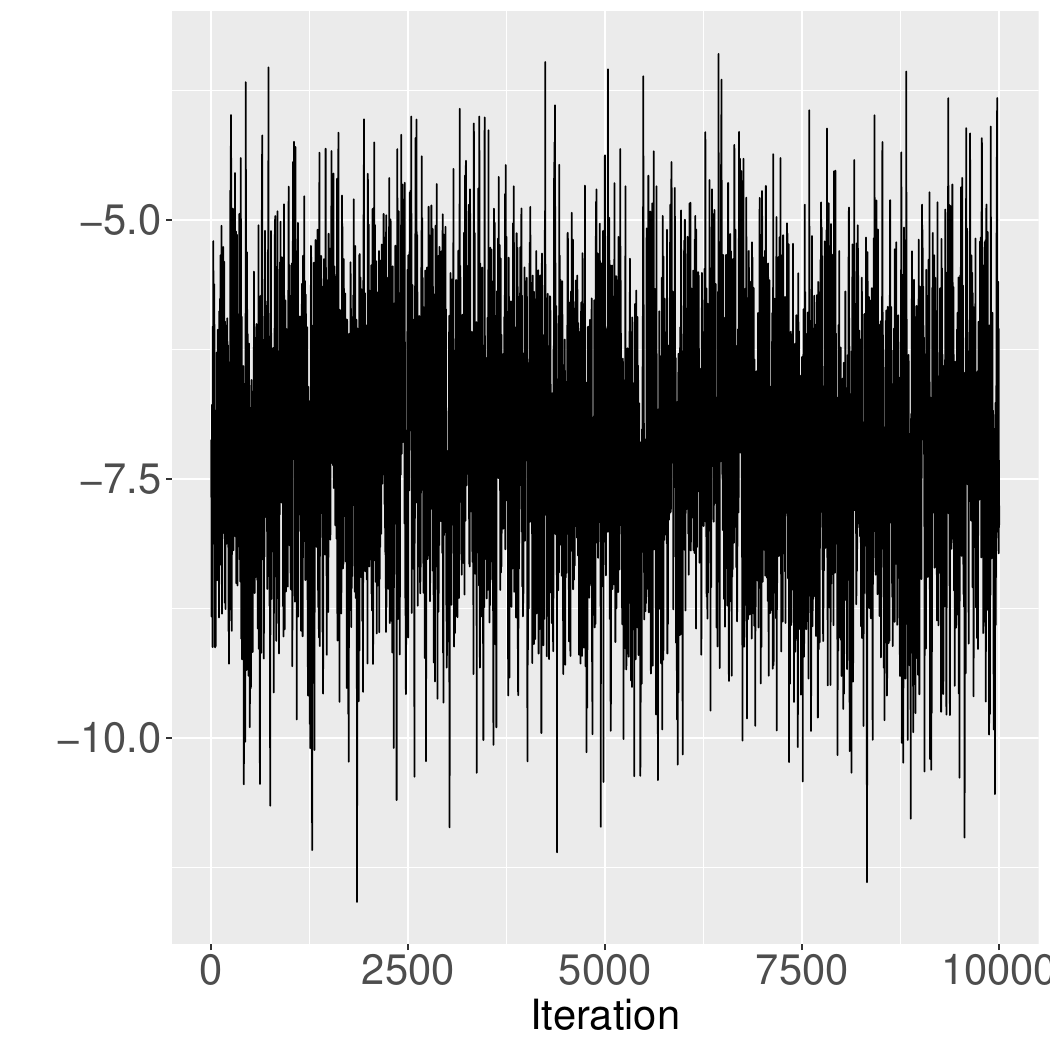} & \includegraphics[width=0.3\textwidth]{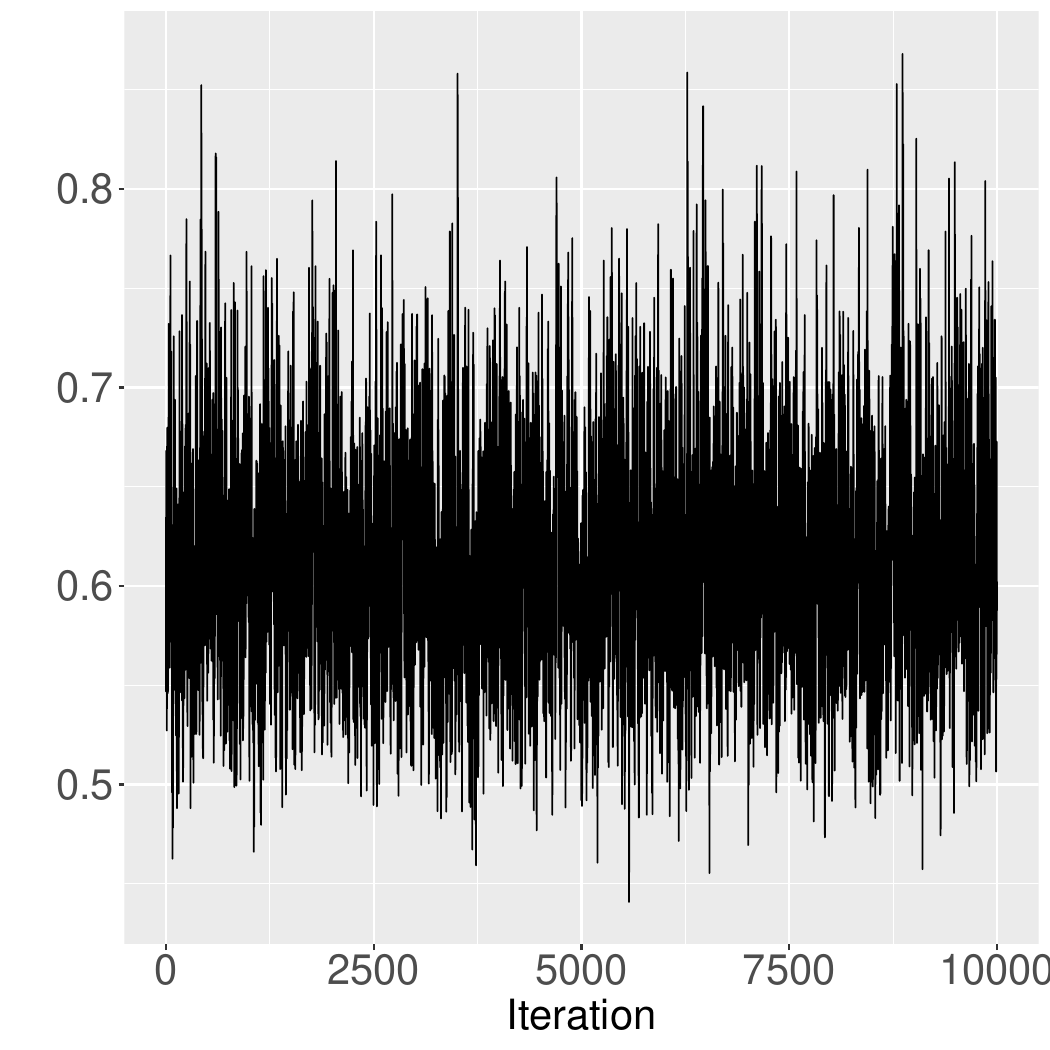} &  \includegraphics[width=0.3\textwidth]{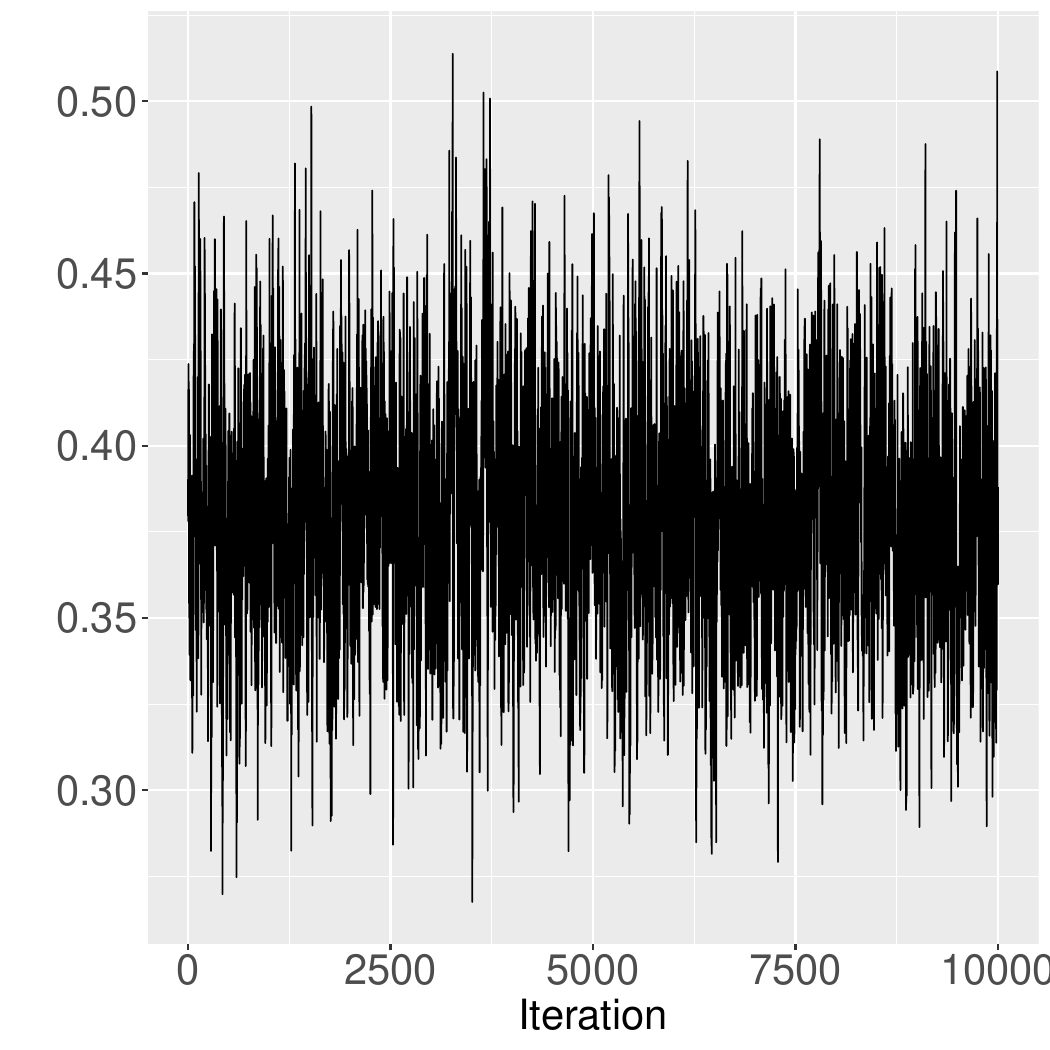}  \\
\hspace{-2.0cm} \includegraphics[width=0.3\textwidth]{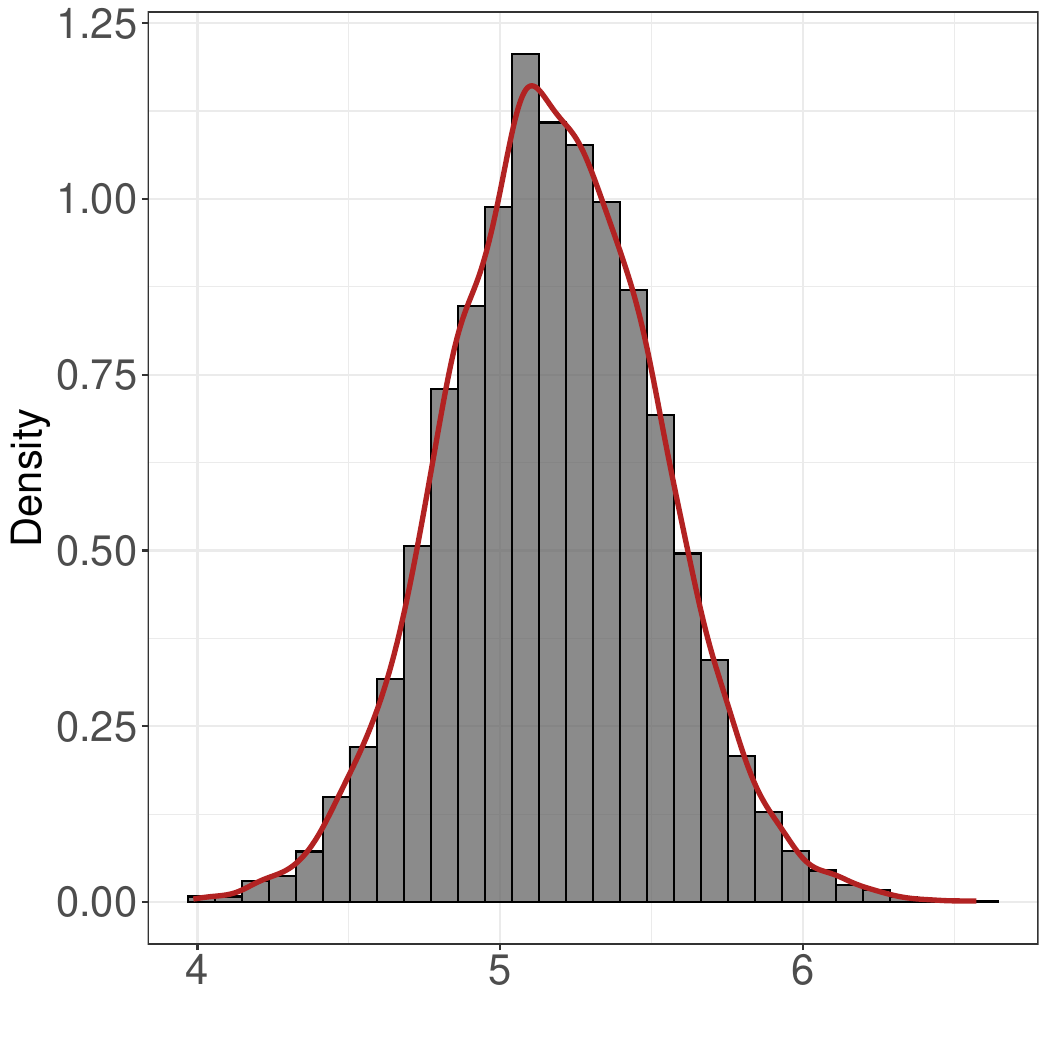} 
& \includegraphics[width=0.3\textwidth]{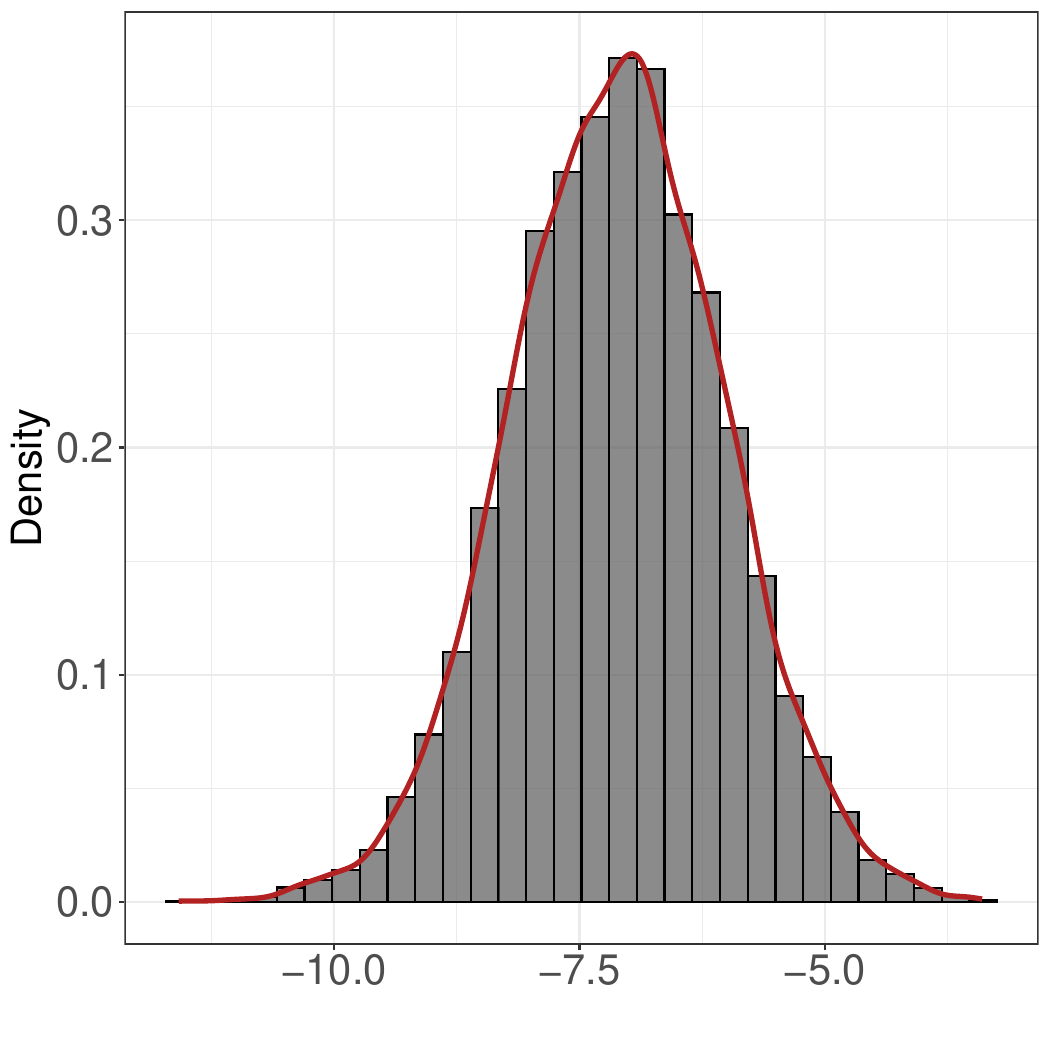} & \includegraphics[width=0.3\textwidth]{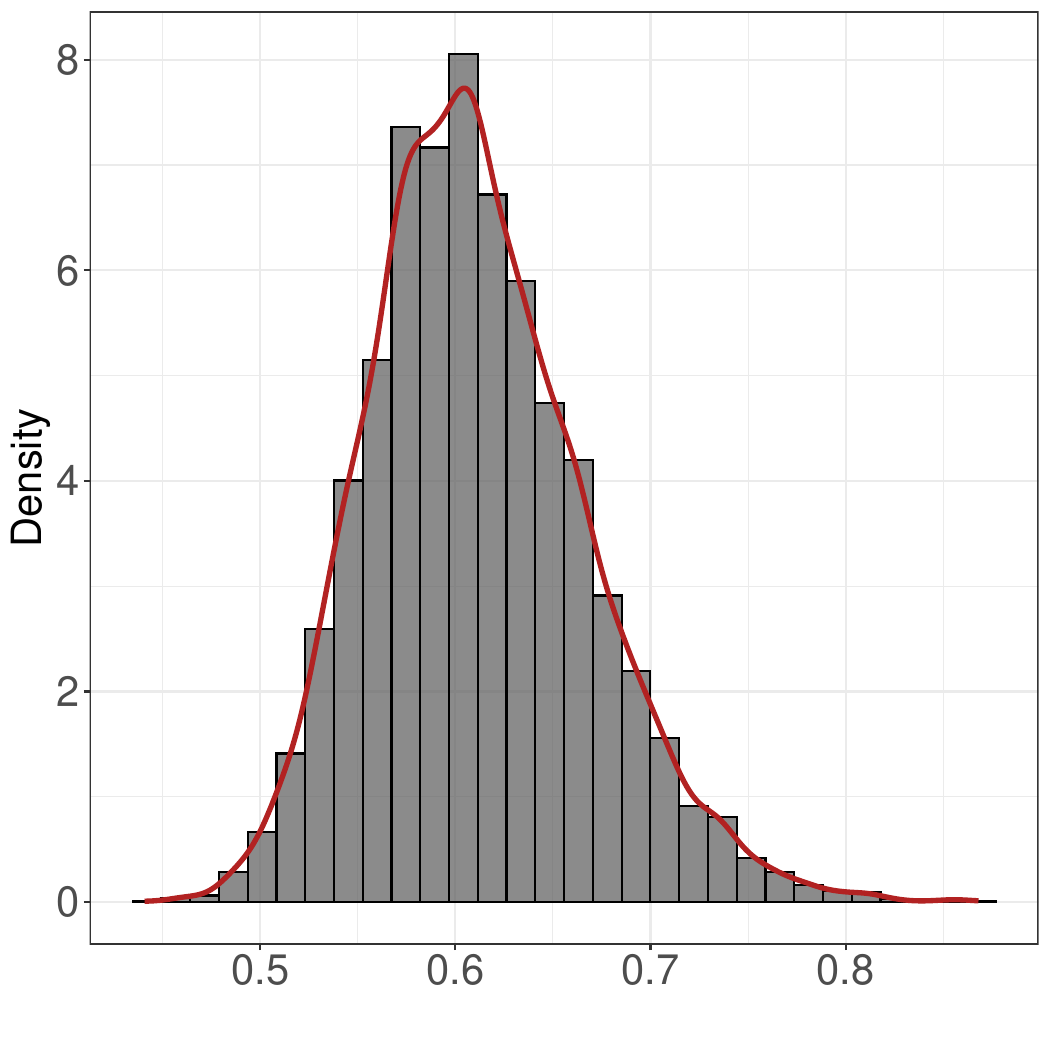} & \includegraphics[width=0.3\textwidth]{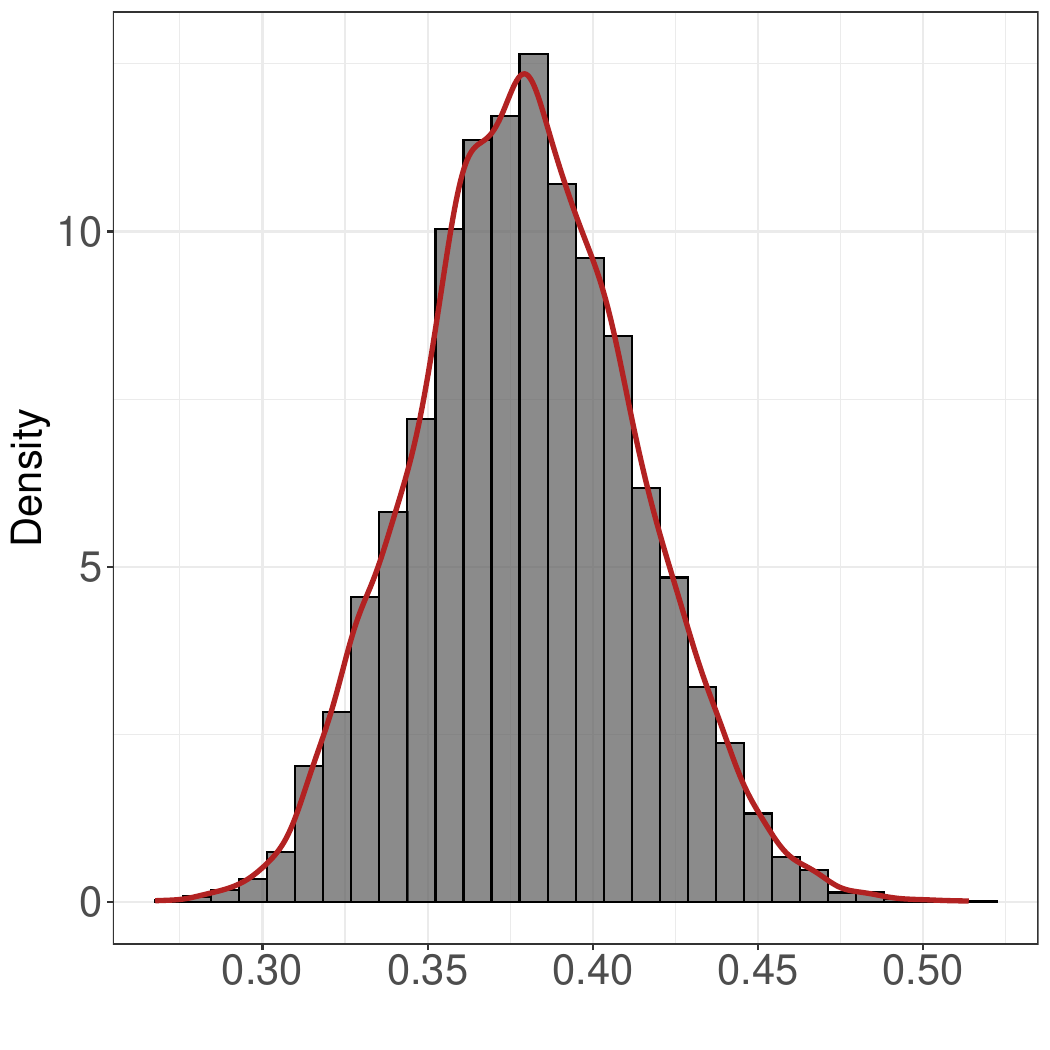}  \\
\hspace{-2.0cm} \includegraphics[width=0.3\textwidth]{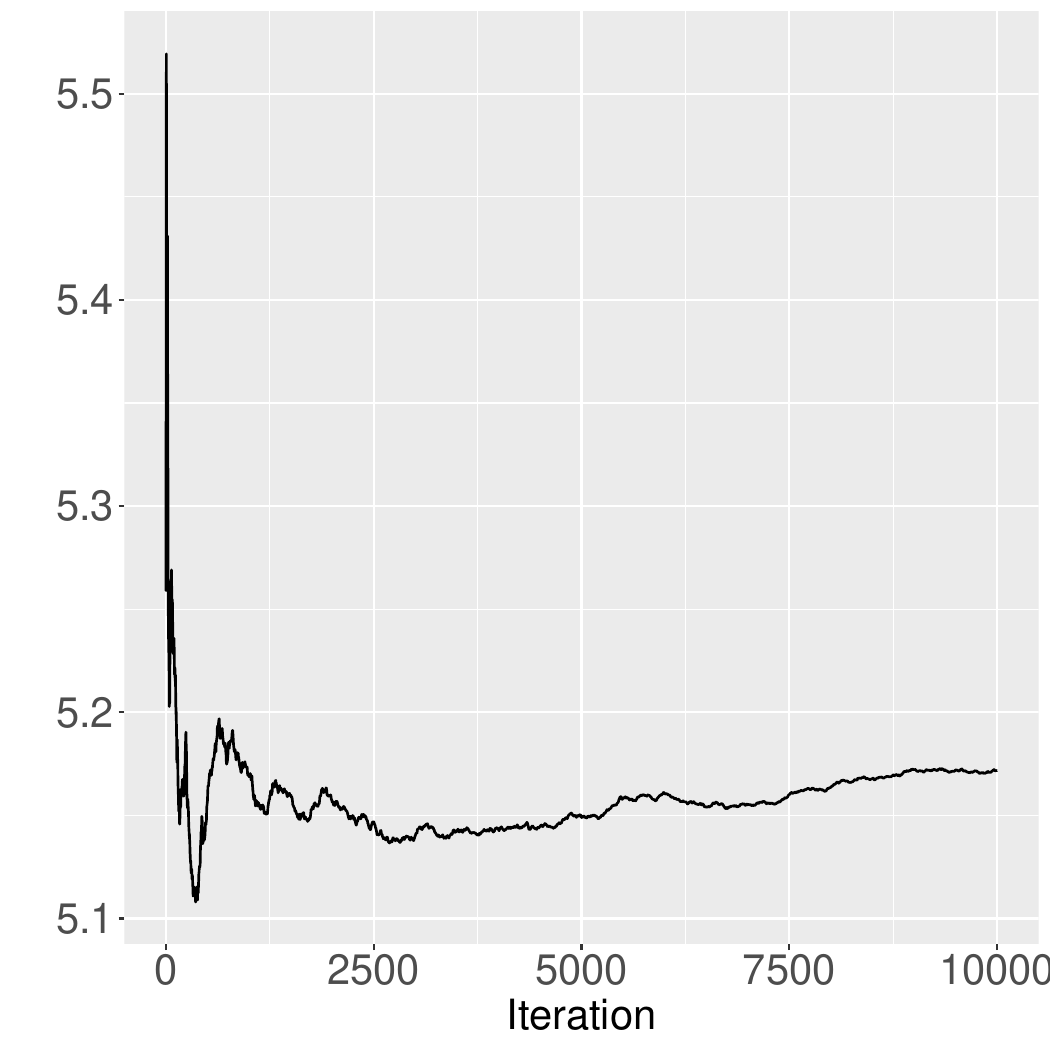} 
& \includegraphics[width=0.3\textwidth]{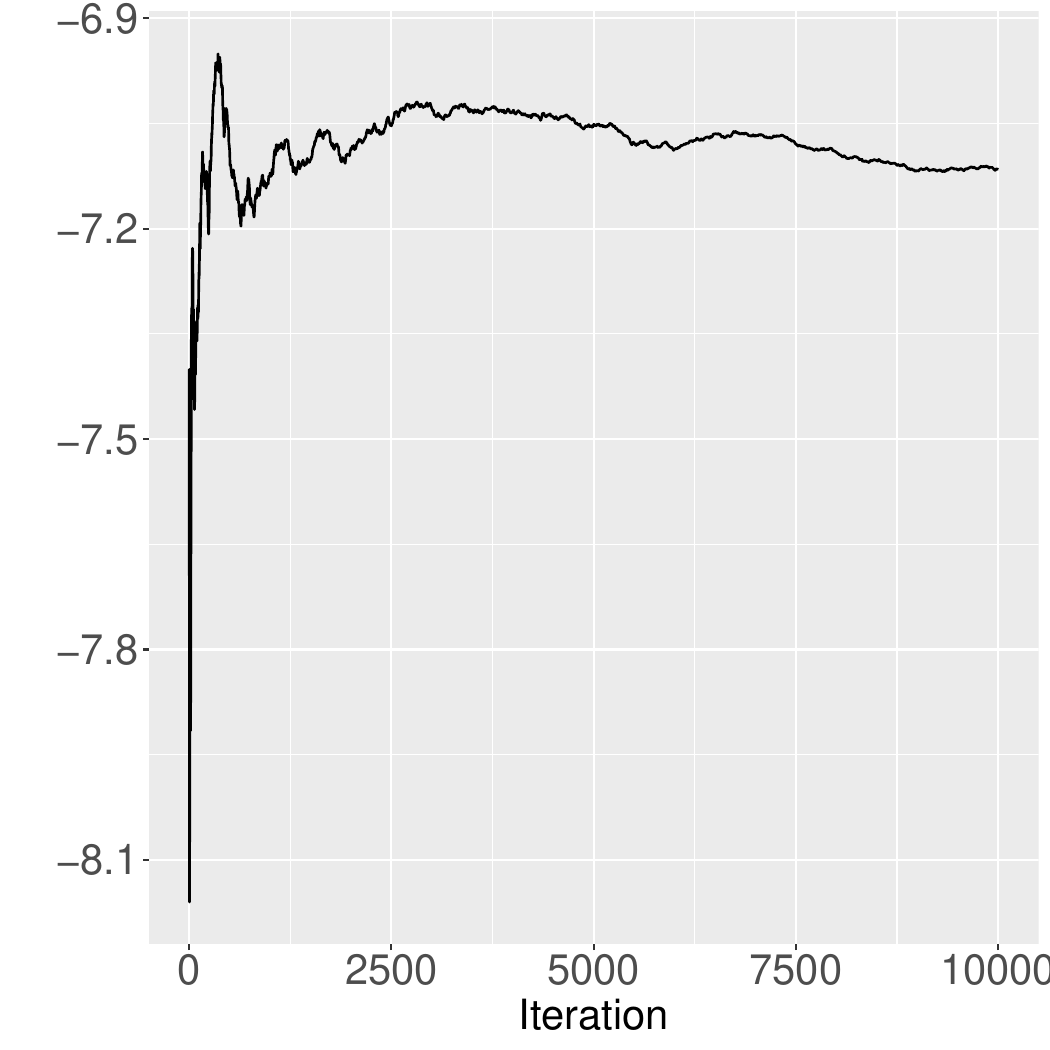} & \includegraphics[width=0.3\textwidth]{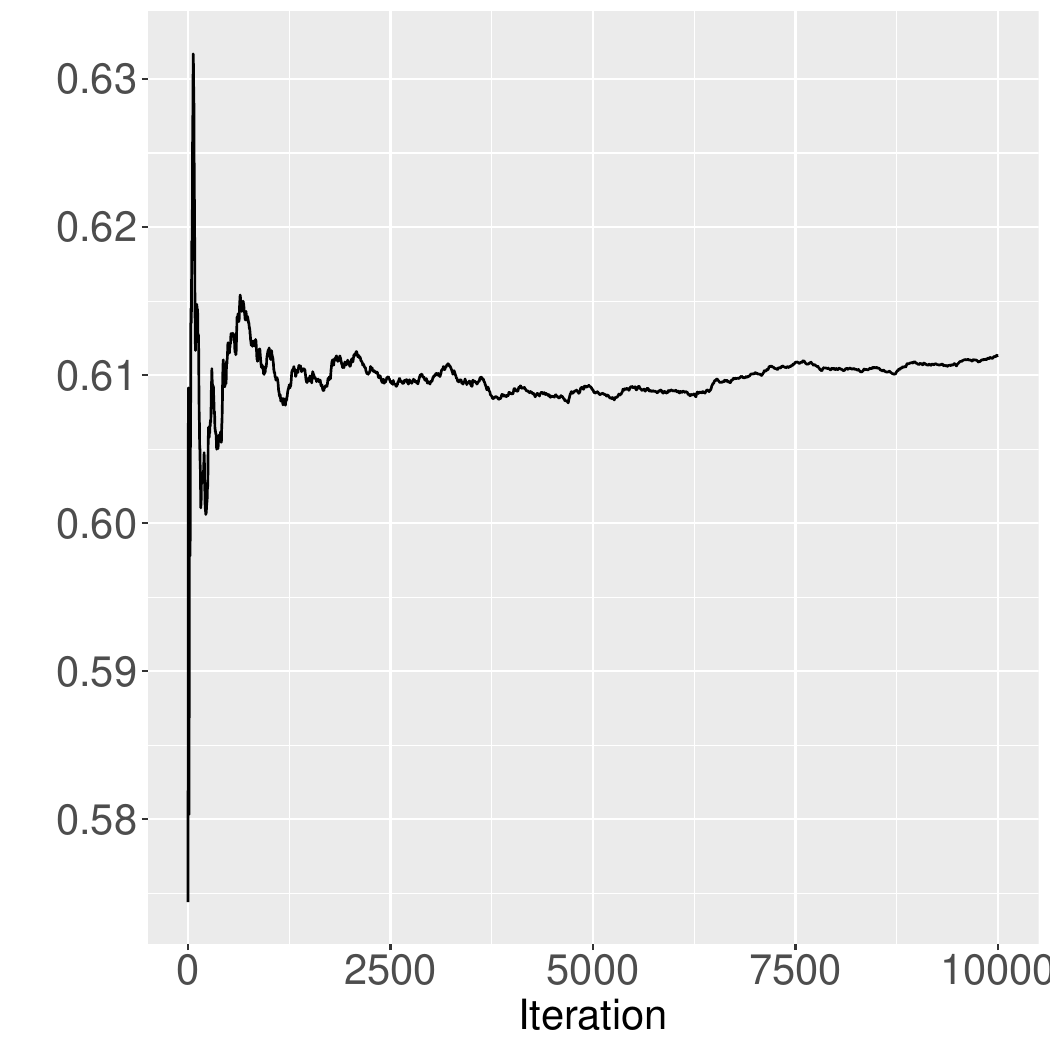} & \includegraphics[width=0.3\textwidth]{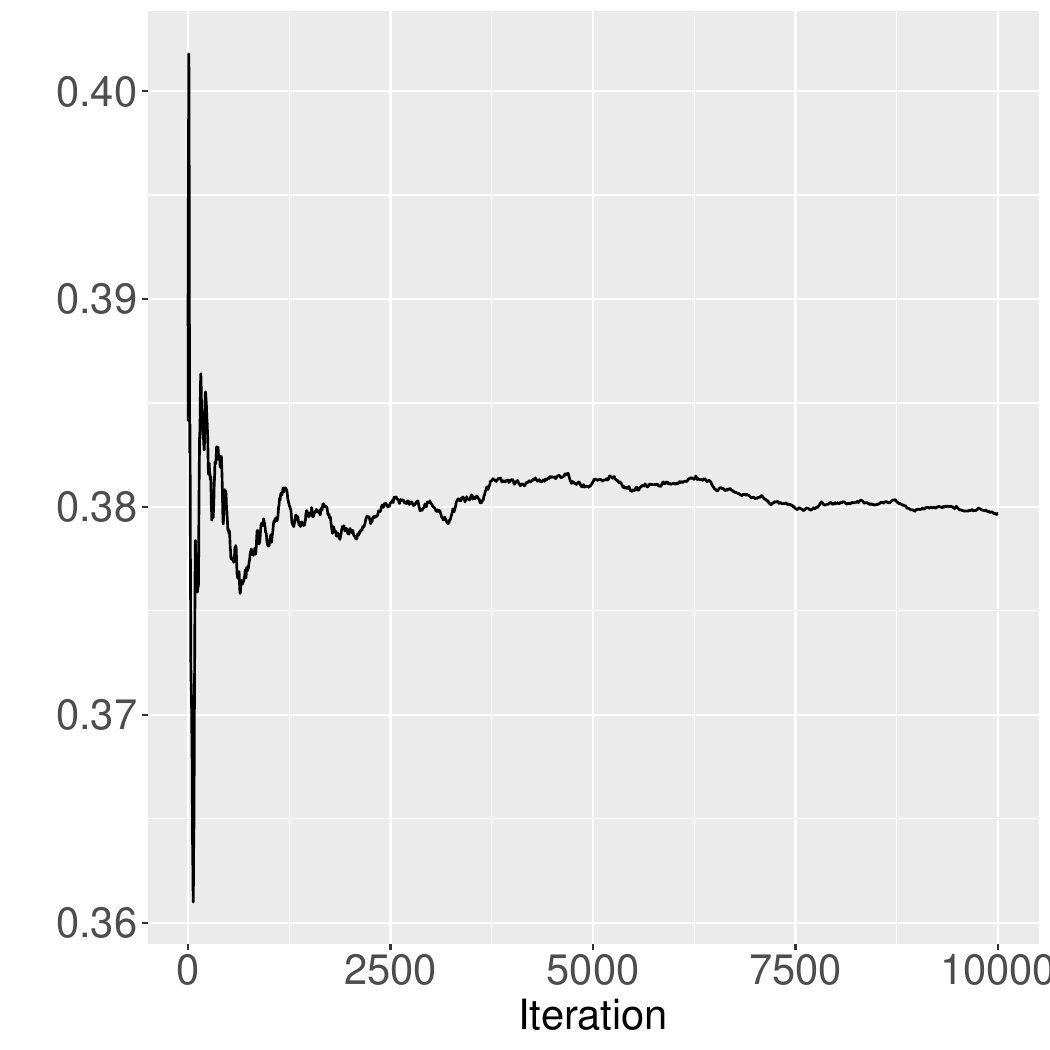}  \\
\end{tabular}
\end{center}
\caption{MCMC diagnostic plots: trace plots (top), density plots (middle), and ergodic mean plots (bottom)
for  $\beta_0$, $\beta_1$,  $\sigma_u$  and $\phi$  under FH-MVL-BP model.} 
\label{fig:diagnostic_FH2}
\end{figure}


\clearpage

\bibliographystyle{ba}
\bibliography{reference,new_bibliof2,refs}

\begin{thebibliography}{39}
\newcommand{\enquote}[1]{``#1''}
\expandafter\ifx\csname natexlab\endcsname\relax\def\natexlab#1{#1}\fi
\expandafter\ifx\csname url\endcsname\relax
  \def\url#1{{\tt #1}}\fi
\expandafter\ifx\csname urlprefix\endcsname\relax\def\urlprefix{URL }\fi
\ifx\endbibitem\undefined \let\endbibitem\relax\fi

\bibitem[{Alkire et~al.(2015)Alkire, Roche, Ballon, Foster, Santos, and
  Seth}]{alkire2015multidimensional}
Alkire, S., Roche, J.~M., Ballon, P., Foster, J., Santos, M.~E., and Seth, S.
  (2015).
\newblock {\em Multidimensional poverty measurement and analysis\/}.
\newblock Oxford University Press, USA.
\endbibitem

\bibitem[{Belkin et~al.(2006)Belkin, Niyogi, and
  Sindhwani}]{belkin2006manifold}
Belkin, M., Niyogi, P., and Sindhwani, V. (2006).
\newblock \enquote{Manifold regularization: A geometric framework for learning
  from labeled and unlabeled examples.}
\newblock {\em Journal of Machine Learning Research\/}, 7(11): 2399--2434.
\endbibitem

\bibitem[{Bell et~al.(2013)Bell, Datta, and Ghosh}]{bell_2013}
Bell, W.~R., Datta, G.~S., and Ghosh, M. (2013).
\newblock \enquote{Benchmarked Small Area Estimators.}
\newblock {\em Biometrika\/}, 100(1): 189--202.
\endbibitem

\bibitem[{Bolin and Wallin(2021)}]{bolin2021efficient}
Bolin, D. and Wallin, J. (2021).
\newblock \enquote{Efficient methods for Gaussian Markov random fields under
  sparse linear constraints.}
\newblock {\em Advances in Neural Information Processing Systems\/}, 34:
  9882--9894.
\endbibitem

\bibitem[{CNA(2014)}]{CNA2014}
CNA (2014).
\newblock \enquote{Censo Nacional Agropecuario.}
\newblock Technical report, Departamento Administrativo Nacional de
  Estad{\'{i}}stica.
\endbibitem

\bibitem[{DANE(2005)}]{DANE2005}
DANE (2005).
\newblock \enquote{Censo Nacional de Poblaci{\'{o}}n y Vivienda.}
\newblock Technical report, Departamento Administrativo Nacional de
  Estad{\'{i}}stica.
\endbibitem

\bibitem[{DANE(2008)}]{DANE2008CV}
--- (2008).
\newblock \enquote{Estimación e interpretación del coeficiente de variación
  de la encuesta cocensal.}
\newblock Technical report, Departamento Administrativo Nacional de
  Estadística (DANE).
\endbibitem

\bibitem[{DANE(2018)}]{DANE2018}
--- (2018).
\newblock \enquote{Censo Nacional de Poblaci{\'o}n y Vivienda - CNPV - 2018.}
\newblock Technical report, Departamento Administrativo Nacional de
  Estad{\'{i}}stica.
\newline\urlprefix\url{https://microdatos.dane.gov.co/index.php/catalog/643/get_microdata}
\endbibitem

\bibitem[{DANE(2025)}]{DANE_ECV_2025}
--- (2025).
\newblock \enquote{Encuesta Nacional de Calidad de Vida 2025.}
\newblock Technical report, Departamento Administrativo Nacional de
  Estad{\'{i}}stica.
\endbibitem

\bibitem[{Esteve et~al.(2024)Esteve, Pohl, Becca, Fang, Galeano,
  Garc{\'\i}a-Rom{\'a}n, Reher, Trias-Prats, and Turu}]{esteve2024global}
Esteve, A., Pohl, M., Becca, F., Fang, H., Galeano, J., Garc{\'\i}a-Rom{\'a}n,
  J., Reher, D., Trias-Prats, R., and Turu, A. (2024).
\newblock \enquote{A global perspective on household size and composition,
  1970--2020.}
\newblock {\em Genus\/}, 80(1): 2.
\endbibitem

\bibitem[{Evgeniou et~al.(2005)Evgeniou, Micchelli, Pontil, and
  Shawe-Taylor}]{evgeniou2005learning}
Evgeniou, T., Micchelli, C.~A., Pontil, M., and Shawe-Taylor, J. (2005).
\newblock \enquote{Learning multiple tasks with kernel methods.}
\newblock {\em Journal of Machine Learning Research\/}, 6(4): 615--637.
\endbibitem

\bibitem[{FAO-OPHI(2022)}]{FAO_OPHI_2022_RMPI}
FAO-OPHI (2022).
\newblock \enquote{Measuring rural poverty with a multidimensional approach:
  {T}he {R}ural {M}ultidimensional {P}overty {I}ndex.}
\newblock Technical Report Statistical Development Series No. 19, Food and
  Agriculture Organization of the United Nations (FAO), Rome.
\endbibitem

\bibitem[{Fay and Herriot(1979)}]{fay_1979}
Fay, R.~E. and Herriot, R.~A. (1979).
\newblock \enquote{Estimates of income from small places: An application of
  {James-Stein} procedures to census data.}
\newblock {\em Journal of the American Statistical Association\/}, 74:
  269--277.
\endbibitem

\bibitem[{F{\'u}quene-Pati{\~n}o(2025)}]{fuquenefully}
F{\'u}quene-Pati{\~n}o, J. (2025).
\newblock \enquote{Fully Bayesian Spectral Clustering and Benchmarking with
  Uncertainty Quantification for Small Area Estimation.}
\newline\urlprefix\url{https://arxiv.org/abs/2512.15643}
\endbibitem

\bibitem[{F{\'u}quene-Pati{\~n}o and Betancourt(2025)}]{fuquene2025ultra}
F{\'u}quene-Pati{\~n}o, J. and Betancourt, B. (2025).
\newblock \enquote{Ultra-sparse small area estimation with super heavy-tailed
  priors for internal migration flows.}
\newblock {\em The Annals of Applied Statistics\/}, 19(1): 121--146.
\endbibitem

\bibitem[{Garza-Rodr{\'\i}guez(2016)}]{garza2016determinants}
Garza-Rodr{\'\i}guez, J. (2016).
\newblock \enquote{The determinants of poverty in the Mexican states of the
  US-Mexico border.}
\newblock {\em Estudios fronterizos\/}, 17(33): 141--167.
\endbibitem

\bibitem[{Ghosh et~al.(2015)Ghosh, Kubokawa, and
  Kawakubo}]{ghosh2015benchmarked}
Ghosh, M., Kubokawa, T., and Kawakubo, Y. (2015).
\newblock \enquote{Benchmarked empirical Bayes methods in multiplicative
  area-level models with risk evaluation.}
\newblock {\em Biometrika\/}, 102(3): 647--659.
\endbibitem

\bibitem[{G{\"o}nen and Alpaydın(2011)}]{gonen2011multiple}
G{\"o}nen, M. and Alpaydın, E. (2011).
\newblock \enquote{Multiple kernel learning algorithms.}
\newblock {\em The Journal of Machine Learning Research\/}, 12: 2211--2268.
\endbibitem

\bibitem[{H{\'a}jek(1971)}]{Hajek1971}
H{\'a}jek, J. (1971).
\newblock \enquote{Comment on ``{A}n {E}ssay on the {L}ogical {F}oundations of
  {S}urvey {S}ampling, {P}art {O}ne''.}
\newblock In Godambe, V.~P. and Sprott, D.~A. (eds.), {\em The Foundations of
  Survey Sampling\/}, 236. New York: Holt, Rinehart, and Winston.
\endbibitem

\bibitem[{Karush(1939)}]{karush1939minima}
Karush, W. (1939).
\newblock \enquote{Minima of functions of several variables with inequalities
  as side constraints.}
\newblock M.sc. dissertation, University of Chicago.
\endbibitem

\bibitem[{Kuhn and Tucker(1951)}]{kuhn1951nonlinear}
Kuhn, H.~W. and Tucker, A.~W. (1951).
\newblock \enquote{Nonlinear programming.}
\newblock In {\em Proceedings of the Second Berkeley Symposium on Mathematical
  Statistics and Probability\/}. University of California Press.
\endbibitem

\bibitem[{Kumar et~al.(2011)Kumar, Rai, and Daume}]{kumar2011co}
Kumar, A., Rai, P., and Daume, H. (2011).
\newblock \enquote{Co-regularized multi-view spectral clustering.}
\newblock {\em Advances in neural information processing systems\/}, 24.
\endbibitem

\bibitem[{Mahmood et~al.(2026)Mahmood, Khan, Akram, Ul~Hassan, and
  Khan}]{mahmood2026multidimensional}
Mahmood, T., Khan, A., Akram, M., Ul~Hassan, M., and Khan, U. (2026).
\newblock \enquote{Multidimensional energy poverty in Pakistan: Determinants
  and spatial patterns at the provincial and district levels.}
\newblock {\em Energy Strategy Reviews\/}, 64: 102127.
\endbibitem

\bibitem[{Mercer et~al.(2015)Mercer, Wakefield, Pantazis, Lutambi, Masanja, and
  Clark}]{mercer2015space}
Mercer, L.~D., Wakefield, J., Pantazis, A., Lutambi, A.~M., Masanja, H., and
  Clark, S. (2015).
\newblock \enquote{Space-time smoothing of complex survey data: Small area
  estimation for child mortality.}
\newblock {\em The Annals of Applied Statistics\/}, 9(4): 1889.
\endbibitem

\bibitem[{Miller and Zhang(2023)}]{miller2023additive}
Miller, N. and Zhang, G. (2023).
\newblock \enquote{Additive multi-task learning models and task diagnostics.}
\newblock {\em Communications in Statistics: Simulation and Computation\/},
  1--18.
\endbibitem

\bibitem[{{Ministerio de Salud y Protección Social}(2025)}]{DHS2025}
{Ministerio de Salud y Protección Social} (2025).
\newblock \enquote{Encuesta Nacional de Demografía y Salud 2025 (ENDS 2025).}
\newblock Technical report, Ministerio de Salud y Protección Social, Bogotá,
  Colombia.
\endbibitem

\bibitem[{N{\'u}{\~n}ez et~al.(2005)N{\'u}{\~n}ez, Ram{\'\i}rez, and
  Cuesta}]{nunez2005determinantes}
N{\'u}{\~n}ez, J., Ram{\'\i}rez, J.~C., and Cuesta, L. (2005).
\newblock \enquote{Determinantes de la pobreza en Colombia, 1996-2004.}
\newblock {\em documento CEDE\/}, 60.
\endbibitem

\bibitem[{Okonek and Wakefield(2024)}]{Okonek2024}
Okonek, T. and Wakefield, J. (2024).
\newblock \enquote{A Computationally Efficient Approach to Fully Bayesian
  Benchmarking.}
\newblock {\em Journal of Official Statistics\/}, 40(2): 283--316.
\newline\urlprefix\url{https://doi.org/10.1177/0282423X241244670}
\endbibitem

\bibitem[{Profamilia(2015)}]{Profamilia}
Profamilia (2015).
\newblock \enquote{Encuesta Nacional de Demograf{\'{i}}a y Salud - ENDS.}
\endbibitem

\bibitem[{Ross and Wakefield(2015)}]{ross2015bayesian}
Ross, M. and Wakefield, J. (2015).
\newblock \enquote{Bayesian hierarchical models for smoothing in two-phase
  studies, with application to small area estimation.}
\newblock {\em Journal of the Royal Statistical Society: Series A (Statistics
  in Society)\/}, 178(4): 1009.
\endbibitem

\bibitem[{Spiegelhalter et~al.(2002)Spiegelhalter, Best, Carlin, and Van
  Der~Linde}]{spiegelhalter2002bayesian}
Spiegelhalter, D.~J., Best, N.~G., Carlin, B.~P., and Van Der~Linde, A. (2002).
\newblock \enquote{Bayesian measures of model complexity and fit.}
\newblock {\em Journal of the Royal Statistical Society: Series B (Statistical
  Methodology)\/}, 64(4): 583--639.
\endbibitem

\bibitem[{Sugasawa et~al.(2025)Sugasawa, Kobayashi, and
  Kawakubo}]{sugasawa2024bayesian}
Sugasawa, S., Kobayashi, G., and Kawakubo, Y. (2025).
\newblock \enquote{Bayesian Benchmarking Small Area Estimation via Entropic
  Tilting.}
\newblock {\em arXiv preprint arXiv:2407.17848\/}.
\endbibitem

\bibitem[{Tang and F{\'u}quene-Pati{\~n}o(2026)}]{tang2026measurement}
Tang, X. and F{\'u}quene-Pati{\~n}o, J. (2026).
\newblock \enquote{Measurement Error Models with Global-Local Priors for Small
  Area Estimation.}
\newblock {\em Bayesian Analysis\/}, 21(3): 1059–1084.
\endbibitem

\bibitem[{Tang et~al.(2018)Tang, Ghosh, Ha, and Sedransk}]{tang2018modeling}
Tang, X., Ghosh, M., Ha, N.~S., and Sedransk, J. (2018).
\newblock \enquote{Modeling random effects using global--local shrinkage priors
  in small area estimation.}
\newblock {\em Journal of the American Statistical Association\/}, 113(524):
  1476--1489.
\endbibitem

\bibitem[{UN(2026)}]{UN2026}
UN (2026).
\newblock \enquote{United Nations, Department of Economic and Social Affairs
  Population Division, 2026.}
\newline\urlprefix\url{https://population.un.org/household/}
\endbibitem

\bibitem[{Von~Luxburg(2007)}]{von2007tutorial}
Von~Luxburg, U. (2007).
\newblock \enquote{A tutorial on spectral clustering.}
\newblock {\em Statistics and Computing\/}, 17: 395--416.
\endbibitem

\bibitem[{Wolter(1985)}]{wolter1985introduction}
Wolter, K.~M. (1985).
\newblock {\em Introduction to Variance Estimation\/}.
\newblock Springer-Verlag, New York.
\endbibitem

\bibitem[{You and Hidiroglou(2023)}]{you2023application}
You, Y. and Hidiroglou, M. (2023).
\newblock \enquote{Application of sampling variance smoothing methods for small
  area proportion estimation.}
\newblock {\em Journal of Official Statistics\/}, 39(4): 571--590.
\endbibitem

\bibitem[{Zhang and Bryant(2020)}]{zhang2020fully}
Zhang, J.~L. and Bryant, J. (2020).
\newblock \enquote{Fully {B}ayesian benchmarking of small area estimation
  models.}
\newblock {\em Journal of Official Statistics\/}, 36(1): 197--223.
\endbibitem

\end{thebibliography}


\end{document}